\documentclass[11pt]{amsart}  
                        
\usepackage{amsmath}
\usepackage{amsfonts}
\usepackage{amssymb}
\usepackage{amsthm}
\usepackage{mathrsfs}
\usepackage{mathtools}
\usepackage{dsfont}
\usepackage{tikz-cd}
\usepackage{physics}
\usepackage{bm}
\usepackage{bbm}
\usepackage[all]{xy}
\usepackage{pdfpages}
\usepackage{enumerate}
\usepackage{float}
\usepackage{todonotes}
\usepackage{graphicx}
\usepackage{color}
\usepackage[colorlinks,linkcolor=blue,urlcolor=blue,citecolor=blue,destlabel=true]{hyperref}

\usepackage[capitalise]{cleveref}
\numberwithin{equation}{section}

\newtheorem{thm}{Theorem}[section]
\newtheorem*{thm*}{Theorem}
\newtheorem{theorem}[thm]{Theorem}
\newtheorem{cor}[thm]{Corollary}

\newtheorem{prop}[thm]{Proposition}
\newtheorem{proposition}[thm]{Proposition}
\newtheorem{lem}[thm]{Lemma}
\newtheorem{lemma}[thm]{Lemma}

\theoremstyle{definition}
\newtheorem{defn}[thm]{Definition}
\newtheorem*{defn*}{Definition}
\newtheorem{definition}[thm]{Definition}

\newtheorem{ex}[thm]{Example}

\newtheorem{rem}[thm]{Remark}
\newtheorem{remark}[thm]{Remark}

\newtheorem{convention}[thm]{Convention}

\newtheorem{steps}[thm]{Steps}

\newcommand{\sT}{\mathscr{T}}

\newcommand{\sF}{\mathscr{F}}
\newcommand{\sE}{\mathscr{E}}

\newcommand{\ra}{\rightarrow}

\newcommand{\N}{\mathbb{N}}

\newcommand{\id}{\mathrm{id}}

\DeclareMathOperator{\Aut}{Aut}
\DeclareMathOperator*{\colim}{colim}

\DeclareMathOperator{\Hom}{Hom}

\DeclareMathOperator{\Map}{Map}

\newcommand{\category}[1]{\mathrm{#1}}

\newcommand{\Sp}{\category{Sp}}
\newcommand{\CGTop}{\category{cgTop}}
\newcommand{\CWTop}{\category{cwTop}}
\newcommand{\sSets}{\category{sSets}}
\newcommand{\fHilb}{\category{Hilb}_\textup{fin}}
\newcommand{\fHilbo}{\category{Hilb}_\textup{fin}^{\otimes}}

\newcommand{\bCGTop}{\category{cgTop}_*}

\newcommand{\Z}{\mathbb{Z}}

\newcommand{\C}{\mathbb{C}}
\newcommand{\R}{\mathbb{R}}

\newcommand{\cC}{\mathcal{C}}

\renewcommand{\to}{\longrightarrow}

\newcommand{\cL}{\mathcal{L}}

\newcommand{\fA}{\mathfrak{A}}
\newcommand{\fB}{\mathfrak{B}}

\newcommand{\fC}{\mathfrak{C}}

\newcommand{\cH}{\mathcal{H}}
\newcommand{\cP}{\mathcal{P}}
\newcommand{\cE}{\mathcal{E}}

\newcommand{\bP}{\mathbb{P}}
\newcommand{\cB}{\mathcal{B}}

\newcommand{\sL}{{\mathscr{L}_{\oplus}}}
\newcommand{\I}{\mathscr{I}}

\newcommand{\sJ}{\mathscr{I}^{\otimes}}

\newcommand{\sI}{\mathscr{I}^{\oplus}}

\newcommand{\sK}{{\mathscr{L}^{\otimes}}}
\newcommand{\bsK}{{\mathscr{L}^{\otimes}_{u}}}

\newcommand{\cT}{\mathcal{T}}

\newcommand{\GP}{\mathcal{GP}}

\newcommand{\QS}{\mathscr{Q}}
\newcommand{\GS}{\mathscr{G}}

\newcommand{\Hslab}[1]{H^{\mathrm{slab},a}}
\newcommand{\instack}{\ominus}

\newcommand{\pt}{\mathrm{pt}}
\renewcommand{\tr}{\mathrm{tr}}

\newcommand{\sP}{\mathscr{P}}

\newcommand{\Kzero}{K_0}
\newcommand{\cO}{\mathcal{O}}
\newcommand{\bY}{\boldsymbol{Y}}
\newcommand{\bGP}{\boldsymbol{\mathcal{GP}}}
\newcommand{\Sing}{\mathrm{Sing}}
\newcommand{\rep}{R}

\renewcommand{\dd}{d}
\newcommand{\pone}{+1}

\newcommand{\tn}[1]{\textnormal{#1}}
\newcommand{\sS}{\mathscr{S}}
\newcommand{\cK}{\mathcal{K}}
\newcommand{\defeq}{\vcentcolon=}
\renewcommand{\1}{\mathds{1}}
\DeclareMathOperator{\Ad}{\mathrm{Ad}}
\newcommand{\bbC}{\mathbb{C}}
\newcommand{\hilbone}{\mathcal{H}}
\newcommand{\hilbtwo}{\mathcal{K}}
\newcommand{\isoone}{f}
\newcommand{\isotwo}{g}

\newcommand{\cV}{\mathcal{V}}

\newcommand{\bbP}{\mathbb{P}}
\newcommand{\fN}{\mathfrak{N}}
\newcommand{\bbN}{\mathbb{N}}
\newcommand{\Unitary}{\mathsf{U}}
\newcommand{\interval}{I}
\newcommand{\denmat}[1]{DM_{#1}(\bbC)}
\newcommand{\htpy}[1]{\simeq_{#1}}
\newcommand{\pathhtpy}[1]{\simeq_{#1,\textup{rel}}}
\DeclareMathOperator{\diag}{\mathrm{diag}}

\newcommand{\cW}{\mathcal{W}}

\newcommand{\cc}{c}

\newcommand{\ff}{\Phi}

\newcommand{\QSC}{\psi^{-1}\QS_\Gamma^{\sqcup}}

\newcommand{\PGS}{\mathcal{P}\mathrm{ic}\mathscr{G}}
\newcommand{\etagP}{\eta^{\mathscr{P}_\Gamma}}
\newcommand{\etagQ}{\eta^{\mathscr{Q}_\Gamma}}
\newcommand{\cpCAlg}{C^*\text{-}\category{Alg}_{1,+}}
\newcommand{\cpCAlgo}{C^*\text{-}\category{Alg}_{1,+}^{\otimes}}

\newcommand{\closure}[1]{\overline{#1}}

\title[Homotopical foundations of quantum spin systems]{Homotopical Foundations of Parametrized Quantum Spin Systems}

\author[Beaudry]{Agn\`es Beaudry\textsuperscript{1}}

\author[Hermele]{Michael Hermele\textsuperscript{2,3}}
\author[Moreno]{Juan Moreno\textsuperscript{1}}
\author[Pflaum]{Markus J.\ Pflaum\textsuperscript{1,3}}
\author[Qi]{Marvin Qi\textsuperscript{2,3}}
\author[Spiegel]{Daniel D.\ Spiegel\textsuperscript{4,5}}

\thanks{\textsuperscript{1}Department of Mathematics, University of Colorado Boulder}
\thanks{\textsuperscript{2}Department of Physics, University of Colorado Boulder}
\thanks{\textsuperscript{3}Center for Theory of Quantum Matter, University of Colorado Boulder}
\thanks{\textsuperscript{4}Department of Mathematics, University of California, Davis}
\thanks{\textsuperscript{5}Center for Quantum Mathematics and Physics, University of California, Davis}


\begin{document}
\begin{abstract}
In this paper, we present a homotopical framework for studying invertible gapped phases of matter from the point of view of infinite spin lattice systems, using the framework of algebraic quantum mechanics. We define the notion of  \emph{quantum state types}. These are certain lax-monoidal functors from the category of finite dimensional Hilbert spaces to the category of topological spaces. The universal example takes a finite dimensional Hilbert space $\mathcal{H}$ to the pure state space of the quasi-local algebra of the quantum spin system with Hilbert space $\mathcal{H}$ at each site of a specified lattice. The lax-monoidal structure encodes the tensor product of states, which corresponds to stacking for quantum systems. We then explain how to formally extract parametrized phases of matter from quantum state types, and how they naturally give rise to $\mathscr{E}_\infty$-spaces for an operad we call the ``multiplicative'' linear isometry operad.  We define the notion of invertible quantum state types and explain how the passage to phases for these is related to group completion. We also explain how invertible quantum state types give rise to loop-spectra. Our motivation is to provide a framework for constructing Kitaev's loop-spectrum of bosonic invertible gapped phases of matter.  Finally, as a first step towards understanding the homotopy types of the loop-spectra associated to invertible quantum state types, we prove that the pure state space of any UHF algebra is simply connected. 
\end{abstract}
\maketitle

\tableofcontents


\section{Introduction}
Gapped invertible quantum phases of matter are a topic of major interest in condensed matter physics.  Selected key examples and partial classifications of such phases were introduced in the following works:  \cite{Laughlin1981,haldane83a, haldane83b, affleck87rigorous, Read2000, Kane2005, Moore2007, Roy2009, Fu2007, Kitaev2009, Ryu2009, gu09tensor,pollmann10,fidkowski11,turner11,chen11a,schuch11,Chen2013,levin12braiding}.

In recent years, there has been significant work done specifically towards better understanding Kitaev's conjecture that gapped invertible phases form a loop-spectrum in the sense of homotopy theory \cite{kitaevSimonsCenter,kitaevIPAM,kitaev,Gaiotto2019}. A loop-spectrum $\bY$ is a sequence of pointed topological spaces $Y_0, Y_1, Y_2, \ldots$ together with weak equivalences (see \cref{defn:weakequivalence})
\[ Y_\dd \xrightarrow{\simeq} \Omega Y_{\dd+1}\]
where $\Omega Y_{\dd+1}$ is the space of based loops in $Y_{\dd+1}$.
The space $Y_\dd$ is to be interpreted as the \emph{classifying space} for phases in spacetime dimension $\dd\pone$, so that its path components $\pi_0Y_\dd$ is the set of phases. Furthermore, phases in spacetime dimension $\dd\pone$ should be in one-to-one correspondence with homotopy classes of loops in the classifying space for phases one dimension higher.

\bigskip
One successful approach aimed at describing the loop-spectrum that classifies phases
has been to use topological quantum field theories (TQFTs); in particular, cobordism invariants play a key role \cite{Kapustin}.  Work of Freed--Hopkins followed by many others settles the question in the TQFT framework. They identify representing loop-spectra depending on the choice of bosonic versus fermionic phases and on the symmetry group of interest \cite{freed_hopkins,FreedHopkinsSym}. This has led to many computations and predictions, for example, in work of Wan--Wang \cite{wan_wang_2019, WAN2020115016}.

However, while it is widely accepted that invertible gapped phases are  classified by TQFTs, there are still many mysteries associated to this correspondence. For instance,  there is a need for a better understanding of Kitaev's insight directly from lattice spin systems.
Progress has been made in this direction in several recent papers, especially by Kapustin, Sopenko, Spodyneiko and Yang
\cite{KS1,KS2,KapustinSopenkoYang,Sopenko,KSo_local_Noether}. Work of Xiong \cite{Xiong} also gives a framework for Kitaev's conjecture, including many details on the connections between homotopy theory and condensed matter theory.  We also note recent work of Marcolli \cite{Marcolli}, which explores the connections between Segal's theory of Gamma-spaces and quantum information theory.\footnote{In this paper, we use $\sE_\infty$-spaces and operads to produce loop-spectra. Gamma-spaces are a different (but related) framework which can also be used to construct loop-spectra. Background on operads and $\sE_\infty$-spaces is given in \cref{sec:opandeinfty}.}

\bigskip
In this article, we do not settle the main question of identifying Kitaev's loop-spectrum. Our goal is to introduce a homotopical framework that we think is well-suited to the problem. This framework is not new to homotopy theory, but its application to lattice spin systems we think is new.

To explain the framework, it is helpful to think of the \emph{generalized cohomology theory} associated with the loop-spectrum of phases. We say what we mean here. 

Given a loop-spectrum $\bY=\{Y_0, Y_1, \cdots \}$ and a topological space $X$, one defines
\[\bY^\dd(X) := [X, Y_\dd]\]
where the right-hand side is the homotopy classes of maps.
The sets $\bY^\dd(X)$ are in fact abelian groups. This is a consequence of the fact that $Y_{\dd}$ is an infinite loop space 
\[Y_{\dd} \simeq \Omega Y_{\dd+1}\simeq \Omega^2 Y_{\dd+2} \simeq \Omega^3 Y_{\dd+3} \simeq \ldots\] 
and, already, $ \Omega^2 Y_{\dd+2}$ is a homotopy associative and commutative $H$-space, and so homotopy classes of maps into it form an abelian group.\footnote{An $H$-space is a based topological space $Y$ with a multiplication $Y \times Y \to Y$ and such that the base point of $Y$ is a unit up to homotopy. Properties of $H$-spaces such as associativity and commutativity are always only required to hold up to homotopy. Background on this is given in \cref{sec:topmonhspace}.}

In the context of invertible gapped phases, there should be a spectrum $\boldsymbol{\GP}$ of gapped phases with spaces $\GP_0$, $\GP_1$, \ldots, and then $\boldsymbol{\GP}^\dd(X)$ is to be thought of as the group of phases of \emph{quantum systems parametrized} by the space $X$. When $X=\pt$, a single point, we recover the non-parametrized gapped phases.

In practice, when trying to understand phases in spacetime dimension $\dd\pone$, one often starts with a space $\QS$ of  \emph{quantum systems} as opposed to phases. This space could be model dependent. For example, one might choose to work with Hamiltonians on certain lattices, or to already restrict to spaces of states with special properties. Furthermore, the space $\QS$ may depend on the symmetry group of interest, or whether or not one is studying fermions or bosons.

For a given dimension, and a space $\QS$ of systems, a quantum system parametrized by $X$ is a continuous function $f\colon X \ra \QS$. Given two such systems $f_1$ and $f_2$, one must be able to stack them to obtain a new system $f_1\instack f_2 \colon X \ra \QS$. For example, for systems modeled by states, the stacking operation takes two states $\omega_1$ and $\omega_2$ and forms their tensor product $\omega_1\otimes \omega_2$. For systems modeled by Hamiltonians $H_1$ and $H_2$, stacking corresponds to forming the operator $H_1\otimes \1 + \1 \otimes H_2 $.

Some quantum systems are considered trivial. 
For example, in a ``states" model, trivial systems will often be completely factorized states. 
In this paper, we will say a parametrized system $f \colon X \ra \QS$ is trivial if $f$ is constant with value a completely factorized state.

\begin{steps}\label{steps}
To pass from quantum systems to phases, one should perform the following identifications, where $f,g \colon X \ra \QS$ are quantum systems:
 \begin{enumerate}
 \setcounter{enumi}{-1}
  \item\label{iso} (Isomorphism) If $f$ and $g$ differ by a change of presentation, such as a change of basis, they are in the same phase;
 \item\label{def} (Deformation) If $f$ and $g$ are homotopic, they represent the same phase;
 \item\label{stab} (Stacking stabilization) If $f$ and $g$ become the equal after stacking with trivial systems, then they represent the same phase.
 \end{enumerate}
 \end{steps}

Often, the first type of equivalence, that of isomorphism, is subsumed by the second in the sense that isomorphisms of systems are realized via homotopies.
This will be assumed to be the case in our framework, so we will not discuss the notion of isomorphism further.


 \begin{ex}
Anyone who has worked with topological $K$-theory will feel a sense of \emph{d\'ej\`a-vu}. We review this connection in \cref{sec:Ktheory} and discuss the relationship of the framework with Kitaev's classification of free fermion phases.   
\end{ex}

 \begin{ex}\label{ex:0d}
The classification of bosonic systems without symmetry in spacetime dimension $0\pone$  is well understood. There are no non-trivial phases if one does not allow for varying parameters. That is, phases over a point are all trivial. The parametrized phases over more interesting topological spaces, however, are not all trivial. Indeed, over the sphere $S^2$, there is a non-trivial phase modeled by the quantum system
\[ H(w_1, w_2, w_3) = w_1\sigma^1+w_2\sigma^2+w_3 \sigma^3\]
where the $\sigma^i$ are the Pauli matrices and $w=(w_1,w_2,w_3) \in S^2$. The ground states of this system form a non-trivial line bundle over $S^2$ and this is the only phase invariant. It is the line bundle associated with the Berry curvature form. 

For any parameter space, the line bundle of ground states is always a complete  phase invariant in spacetime dimension $0\pone$. Since $BGL(1)$ classifies complex line bundles, there is an isomorphism
\begin{equation*}
 \bGP^{0}(X) \cong [X, BGL(1)]
\end{equation*}
and, for bosonic systems with no-symmetry, we have
$\GP_0 \simeq BGL(1)$.
We review this more in \cref{sec:0d}. 
\end{ex}
\begin{ex}[{\cite{qpump}}]\label{ex:1d}
Consider an infinite chain of spin-$1/2$ particles in spacetime dimension $1\pone$. This is modeled by taking  a lattice $\Gamma=\Z$ and, on each lattice site $v\in \Z$, placing a copy of $\mathfrak{B}(\C^2)$. 
Denote the Pauli operators on site $v$ by $\sigma^{1}_v,\sigma^{2}_v,\sigma^{3}_v$.
Model $S^3$ as pairs $(w, t)$ in the quotient of $S^2 \times [-1,1]$ which identifies all the points $(w,1)$ to one point, and similarly for the points $(w,-1)$.
Consider a Hamiltonian which is a sum of three terms
\begin{align*}
H(w,t) = \sum_{v \in \Z} H^1_v(w,t) + \sum_{v \in \Z | v \text{ even}} H^{2+}_{v,v+1} (t) +  \sum_{v \in \Z | v \text{ odd}} H^{2-}_{v,v+1} (t) \text{.}
\end{align*}
For the first term, 
\[H^1_v(w,t) = (-1)^v \sqrt{1-t^2} H(w)\] 
where $H(w)$ is the Hamiltonian for the $0+1$ dimensional family in \cref{ex:0d}.
For the other two terms, we
let $ g_+(t) = t-1/2$ for $t\in [1,1/2)$, and $g_-(t) = -t-1/2 $ for $ t\in (-1/2,-1]$.
Both functions are zero otherwise.
 Then 
 \[H^{2 \pm}_{v, v+1} = g_{\pm}(t)(\sigma_v^{1} \sigma^{1}_{v+1}+\sigma_v^{2} \sigma^{2}_{v+1}+ \sigma_v^{3} \sigma^{3}_{v+1}).\]
At each value of $(w,t)$ the system decomposes as a tensor product of decoupled $0+1$ dimensional quantum systems, which are either single or pairs of spins, making the ground state and spectrum of $H(w,t)$ easy to analyze.  In particular, $H(w,t)$ is always gapped with a unique ground state $\omega(w,t)$.  The pairing of the spins is not fixed and two different pairings  occur as one varies $t$, with a switch between them as we go from $t=1$ to $-1$.

Note that we get a function
\[\omega \colon S^3 \to \sP_{\Z}(\C^2)\]
where $\sP_{\Z}(\C^2)$ is the pure state space of the $C^*$-algebra 
$\bigotimes_{v\in \Z } \mathfrak{B}(\C^2)$ (as defined in  \cref{sec:Csetup}).
This is a parametrized family of states over $S^3$, and it is continuous if we give $ \sP_{\Z}(\C^2)$ the weak$^*$ topology (but not with the norm topology).
\end{ex}


In both Examples~\ref{ex:0d} and \ref{ex:1d}, the  specific Hamiltonian representing the system was not important. Modeling systems using Hamiltonians was a choice of presentation, and there are many such choices. In fact, Kitaev has argued that in the context of local gapped invertible phases with unique ground states on an infinite lattice, one should only study the states themselves, as the map which sends a Hamiltonian to its ground state will be a weak equivalence.  While the arguments, as described in \cite[\S 2.1.1]{Xiong}, are quite compelling,  to our knowledge, a mathematically rigorous treatment of local gapped Hamiltonians on an infinite lattice, and of this weak equivalence, still has not appeared in the literature.

\begin{ex}
In spacetime dimension $0\pone$, suppose we have a finite dimensional Hilbert space $\cH$.  Let $\bP(\cH)$ be its pure state space. Let $\mathscr{H}(\cH)$ be the space of Hermitian operators on $\cH$ with one-dimensional ground state. The map
\[ \mathscr{H}(\cH) \to \bP(\cH)\]
which sends $H $ to its ground state is a homotopy equivalence.
Therefore, there is no loss in only studying the pure state space instead of Hamiltonians.
\end{ex}

In this paper, we therefore focus on pure state spaces. In spacetime dimension $\dd\pone>0\pone$, the state spaces come from infinite dimensional objects. So there is a question of topology that needs to be resolved. In \cref{ex:1d}, we noted that the family of states parametrized by $S^3$ is not norm continuous, but rather weak$^*$ continuous. We expect that more general parametrized Hamiltonians should also give rise to weak* continuous families of ground states; for example this is true for the families of states considered in \cite{BachmannMichalakisNachtergaeleSimsAutomorphicEquivalence,KSo_local_Noether}. Furthermore, it is physically reasonable to require expectation values of local observables to be continuous. The weak* topology is the minimal topology with respect to this property. We therefore consider the weak* topology on our spaces of states $\QS$ and will henceforth restrict our attention to this topology.

\bigskip

We now get more specific about the content of our paper. Throughout, we work in the category of \emph{compactly generated topological spaces} $\CGTop$. We define this category in \cref{conv:top}.

We let 
$\fHilb$ 
be the category whose objects are the \emph{non-zero} finite-dimensional complex Hilbert spaces and morphisms are linear isometric embeddings. 
This is a symmetric monoidal category with respect to the tensor product. When considering the monoidal structure, we denote the category by $\fHilbo$. This is a topologically enriched category where the set of linear isometries, which we denote by $\I(\cH,\cK)$, is given the subspace topology with respect to all linear maps.

Now, let $\Gamma$ be a lattice. In spacetime dimension $\dd\pone$, we will usually use $\Gamma =\Z^d$ but this is not important for the general framework.
For each $\cH$ in $\fHilbo$, we consider the $C^*$-algebra $\fA_\Gamma(\cH)$, where the quantum degrees of freedom on each lattice site of $\Gamma$ are described by $\cH$.
So, 
\[\fA_\Gamma(\cH) = \bigotimes_{v\in \Gamma } \mathfrak{B}(\cH) \]
where the tensor product is understood as a UHF algebra. 
We let $\sP_\Gamma(\cH)$ be the pure state space of $\fA_\Gamma(\cH)$ with the weak$^*$ topology. Note that $\sP_\Gamma(\C)$ is a space with a single point, which we denote by $\pt$. This is all reviewed in detail in  \cref{sec:Csetup}. 

Given states $\omega_1 \in \sP_\Gamma(\cH_1)$ and  $\omega_2 \in \sP_\Gamma(\cH_2)$, we can define their tensor product $\omega_1\otimes \omega_2 \in \sP_\Gamma(\cH_1\otimes \cH_2)$. We note that this is the operation on ground states that corresponds to stacking gapped Hamiltonians with a unique ground state.

In  \cref{sec:Csetup}, we show the following result: 
\begin{thm*}[\ref{thm:uniquantstatetype}]
There is a topologically enriched functor
\[ \sP_\Gamma \colon \fHilbo \to \CGTop \]
which maps $\cH$ to $\sP_\Gamma(\cH)$ and with the property that for a linear isometric embedding $f\colon \hilbone \to \hilbtwo$,
\[\sP_\Gamma(f) \colon \sP_\Gamma(\cH) \to \sP_\Gamma(\hilbtwo) \] is a closed embedding. 
Furthermore, the map which takes states to their tensor product gives a natural transformation
\[ \etagP \colon \sP_\Gamma(\cH_1) \times \sP_\Gamma(\cH_2) \to  \sP_\Gamma(\cH_1 \otimes \cH_2) ,\]
making $\sP_\Gamma$ into a lax monoidal functor.
\end{thm*}

We construct $ \sP_\Gamma$ by first constructing a strong monoidal contravariant functor $(\fA_\Gamma, \eta^{\fA_\Gamma}) \colon \fHilbo \to \cpCAlgo$ which takes $\cH$ to $\fA_\Gamma(\cH)$, viewed as an object in
the category $\cpCAlgo$ of nuclear unital $C^*$-algebras and unital completely positive linear maps (which are automatically bounded), equipped with their unique tensor product. This  is \cref{prop:strongmonoidalfunctoriality}.
By composing $\fA_\Gamma$ with the functor that assigns the state space, we obtain a pair $(\sS_\Gamma, \eta^{\sS_\Gamma})$ consisting of a functor  $\sS_\Gamma \colon \fHilbo \to \CGTop$ and a natural transformation $\eta^{\sS_\Gamma}$ just as above, but here $\sS_\Gamma(\cH)$ is the space of all states (not just pure ones) in the weak$^*$ topology. This is \cref{thm:statespacefuntopoenriched}.  Then $\sP_\Gamma$ is obtained by restricting $\sS_\Gamma$ to the pure state space. Of course, there are many details to check to ensure this works, and \cref{sec:Csetup} is dedicated to that.

As we will see in the construction, for a linear isometric embedding
$f \colon \cH \to\cK, \sP_\Gamma(f)(\omega)$ is obtained via pre-composition of $\omega$ with a completely
positive map $\fA_\Gamma(f)$. The map $\fA_\Gamma(f)$ is the colimit of maps
$\fA_\Lambda(f) = \Ad\big(\bigotimes_{v \in \Lambda} f^*\big)$ defined for finite subsets
$\Lambda \subset \Gamma$. That is, one might write
$\fA_\Gamma(f)$ suggestively as $\Ad\big( \bigotimes_{v \in \Gamma} f^* \big)$.

\begin{ex}
If $\dd=0$ so that $\Gamma=\Z^0$ is a single point, then 
\[\sP_{\Z^0}(\cH) \cong \bP(\cH)\] and
\[\sP_{\Z^0}(f) \colon\bP(\cH)\to \bP(\cK)\]
is simply the map which takes the one dimensional subspace $\psi$ of $\cH$ to the  subspace $f(\psi)$ of $\hilbtwo$. On the other hand $ \eta^{\sP_{\Z^0}}(\psi, \phi) = \psi\otimes \phi$, the tensor product of the one-dimensional subspaces.
\end{ex}

We construct and study $\sP_\Gamma$ in \cref{sec:Csetup}. Once we have constructed $\sP_\Gamma$, we can introduce the main structure we study in the rest of the paper, that of a \emph{quantum state type}. We make this definition with bosonic systems with no symmetry in mind, in the hope that it can be adapted to more general settings.
\begin{defn*}[\ref{defn:quantumstatetype}]
A \emph{quantum state type} on a lattice $\Gamma$ is a topologically enriched functor 
\[\QS_\Gamma \colon \fHilbo \to  \CGTop \]
such that the following hold:
\begin{enumerate}[(a)]
\item For each object $\cH$, $\QS_\Gamma(\cH)$ is a subspace of $\sP_\Gamma(\cH)$
containing all completely factorized states. 
\item  For each morphism $f\colon \cH \to \cL$, $\QS_\Gamma(f)$  is the restriction of $\sP_\Gamma(f)$ and is a closed embedding.
\item The natural transformation $\etagP $ restricts to a natural transformation
\[
\etagQ  \colon  \QS_\Gamma(\cH_1)\times \QS_\Gamma(\cH_2) \to  \QS_\Gamma(\cH_1\otimes \cH_2).
  \]
  \end{enumerate}
  The quantum state type $\sP_\Gamma$ is called the \emph{universal quantum state type}. 
\end{defn*}

This setup will allow us to deduce many formal results about quantum state types. In this paper, we will focus on the universal quantum state type $\sP_\Gamma$. While example of physical interest on infinite lattices will take $\QS_\Gamma(\cH)$ to be a proper subspace of $\sP_\Gamma(\cH)$, the universal quantum state type $\sP_\Gamma$ gives us a toy example for developing the theory and the results we prove about $\sP_\Gamma$ are of independent interest for non-commutative geometry. We discuss proposals for more physically relevant quantum state types on infinite lattices in \cref{ex:physicalexamples} below.


\bigskip

Assume we have fixed a quantum state type $\QS_\Gamma$.  We let $\cH$ be an object of $\fHilbo$ of dimension at least two. We choose a distinguished vacuum state $\psi$ of $\cH$. We let $f_i \colon \cH^{\otimes i} \to \cH^{\otimes (i+1)}$ be the map obtained by tensoring on the right with a fixed choice of unit vector in $\psi$. This is a linear isometry, so the maps $\QS_\Gamma(f_i)$ are closed embeddings. The completely factorized state $\bigotimes_{\Gamma} \psi$ is in $\QS_\Gamma(\cH)$ by assumption. By an abuse of notation, we simply call this state  $\psi$.  We then have
\[\QS_\Gamma(f_i) (\omega) = \etagQ(\omega, \psi ) = \omega\otimes \psi.\]

We introduce two spaces that play a key role in our study of phases.\footnote{At this point, we apply a functorial CW-replacement to our quantum state type. This preserves the weak homotopy type but puts us in a context better suited for homotopy theory. The theorems are to be understood up to this functorial replacement. This is explained carefully in \cref{sec:CWreplacements}. Since we are interested in the homotopy theory of quantum state types, we don't lose anything through this replacement.} First,  we let
\[\QS_{\Gamma}^{\sqcup}(\cH):= \coprod_{i\geq 0} \QS_\Gamma(\cH^{\otimes i}).\]
Also let  $\psi_i = \psi^{\otimes i}\in  \QS_\Gamma(\cH^{\otimes i})$. 
The natural transformations $\etagQ$ give $\QS_{\Gamma}^{\sqcup}(\cH)$ the structure of a strictly associative topological monoid with unit $\psi_0$. It is homotopy commutative, but not strictly commutative. See \cref{prop:QStopmonoid}.

The space $\QS_{\Gamma}^{\sqcup}(\cH)$ is our classifying space for quantum systems. Henceforth, a \emph{parameter space} will always be assumed to be a compact Hausdorff space  $X$ which is homotopy equivalent to a CW complex. Then,  a continuous function
\[ \omega \colon X \to \QS_\Gamma^{\sqcup}(\cH)\]
is a \emph{quantum system parametrized by $X$}. See \cref{defn:classifying}. We say it is \emph{trivial} if it is constant with value the completely factorized state $\psi_i$ for some $i\geq 0$.  We continue our abuse of notation and call this constant function $\psi_i$.

Now that we have defined quantum systems, we also want to define \emph{quantum phases}. To this end, define
\[  \QS_\Gamma(\cH^{\otimes \infty}) :=  \colim_{i, \otimes \psi} \QS_\Gamma(\cH^{\otimes i}).\]
We give $\QS_\Gamma(\cH^{\otimes \infty})$ the base point corresponding to the image of $\psi_0$, which, by design, is equal to the image of every $\psi_i$.  Physically, we can roughly think of elements of $\QS_\Gamma(\cH^{\otimes \infty})$ as states in $\QS_{\Gamma}^{\sqcup}(\cH)$, stacked with an infinite number of copies of the trivial state $\psi$.  Upon passing to $\QS_\Gamma(\cH^{\otimes \infty})$, we no longer need to explicitly consider stacking stabilization, and all information about phases is encoded in the homotopy type of $\QS_\Gamma(\cH^{\otimes \infty})$.  We emphasize that $\QS_\Gamma(\cH^{\otimes \infty})$ is not a subspace of pure states with $\cH^{\otimes \infty}$ as the single-site vector space.

With this space in hand, we can define our notion of phases. Namely, the \emph{quantum phases of state type $\QS_\Gamma$ parametrized by $X$} is the set of homotopy classes
\[[X,  \QS_\Gamma(\cH^{\otimes \infty})].\]
We will justify this definition, but first we discuss the structure on the space $\QS_\Gamma(\cH^{\otimes \infty})$. We prove it is a homotopy associative and commutative $H$-space. This means that  $\QS_\Gamma(\cH^{\otimes \infty})$ has a product that is associative and commutative up to homotopy, with an element that also satisfies the conditions of being a unit up to homotopy.   In fact, we prove more, we show that it is an $\sE_\infty$-space, a concept which is reviewed in \cref{sec:opandeinfty}.
In particular, $\pi_0\QS_\Gamma(\cH^{\otimes \infty})$, the space of path components, is a monoid. However, it may not be a group. The cases when it \emph{is} a group are particularly interesting. We will say that a quantum state type $\QS_\Gamma$ is \emph{invertible} if $\pi_0\QS_\Gamma(\cH^{\otimes \infty})$ is a group (see \cref{defn:QSinvert}). We summarize this in the following theorem which is proved in this paper.

\begin{thm*}[\ref{thm:jmathHspace}]
The space  $ \QS_\Gamma(\cH^{\otimes \infty})$ is an $\sE_\infty$-space. As such, it is a homotopy commutative and associative $H$-space. If $\QS_\Gamma$ is invertible, then $ \QS_\Gamma(\cH^{\otimes \infty})$ is an $H$-group.

The map 
\[\jmath =\coprod \jmath_i \colon \QS_\Gamma^{\sqcup}(\cH) \to \QS_\Gamma(\cH^{\otimes \infty})\]
 obtained from the inclusions $\jmath_i\colon \QS_\Gamma(\cH^{\otimes i}) \to  \QS_\Gamma(\cH^{\otimes \infty})$ is a morphism of $H$-spaces. 
\end{thm*}

The reason $\QS_\Gamma(\cH^{\otimes \infty})$ is particularly interesting to us is the following result.
\begin{thm*}[\ref{thm:jquotientiso}]\label{thm:jquotientiso:intro}
For a parameter space $X$, the map $\jmath$ induces a map of monoids
\[\pi_0\Map(X,\jmath) \colon \Map(X,  \QS_\Gamma^{\sqcup}(\cH)) \to [X,  \QS_\Gamma(\cH^{\otimes \infty})], \]
which identifies $[X,  \QS_\Gamma(\cH^{\otimes \infty})]$ with the quotient of the space of quantum systems
\[ \Map(X,  \QS_\Gamma^{\sqcup}(\cH))\] 
by the relation $\omega_1 \approx \omega_2$ if and only if there are $i,j \geq 0$ such that
\[ \omega_1 \otimes \psi_i \simeq \omega_2 \otimes \psi_j.\]
\end{thm*}

Therefore, $ [X,  \QS_\Gamma(\cH^{\otimes \infty})]$ is precisely the result of passing from quantum systems to quantum phases by imposing the equivalences of deformation \eqref{def}  and of  stacking stabilization  \eqref{stab}. As noted before, we expect that equivalences corresponding to isomorphisms \eqref{iso} will be realized by deformations.\footnote{For example, because $\QS_\Gamma$ is a topologically enriched functor and the unitary group of $\cH$ is path connected, any isomorphism between systems obtained by conjugating each site in $\fA_\Gamma(\cH)$ by a fixed a unitary $U\in \fB(\cH)$ (the same on each site) is realized through a homotopy.} 

For invertible quantum state types, we can do even better and describe the passage from quantum systems to phases as a group completion process.
\begin{thm*}[\ref{thm:locisgroup}]
If $\QS_\Gamma$ is an invertible quantum state type,  then the monoid of phases $[X,  \QS_\Gamma(\cH^{\otimes \infty})]$ is an abelian group, and there is an isomorphism
\[ K_0([X,  \QS_\Gamma^{\sqcup}(\cH)]) \cong [X, \QS_\Gamma(\cH^{\otimes \infty}) ]\times [X,\Z]\]
where $K_0$ is the Grothendieck group completion functor.
\end{thm*}

The factor of $[X,\Z]$ appears because, on any path component $X_0$ of $X$, a quantum systems $f$ on $X$ must satisfy $f(X_0) \subset \QS_\Gamma(\cH^{\otimes i})$ for some specific $i \geq 0$, and the group completion of the non-negative integers is $\Z$. 

\begin{rem}
The definition of invertibility only depended on $\pi_0  \QS_\Gamma(\cH^{\otimes \infty})$. Since this formally implies that $[X,  \QS_\Gamma(\cH^{\otimes \infty})]$ is an abelian group, this means that a quantum system parametrized by $X$ is invertible if and only if it is invertible at each point of $X$. See \cref{rem:invertiblepointwisediscussion} for more discussion on this point.
\end{rem}

\begin{rem}
We can actually do more than this and relate the spaces $\QS_\Gamma^{\sqcup}(\cH)$ and $ \QS_\Gamma(\cH^{\otimes \infty})$ to Quillen's group completion $\Omega B \QS_\Gamma^{\sqcup}(\cH)$. This is a step-up in theoretical technicality so we do not discuss it further in the introduction, but the results are in \cref{thm:quillengroupcompletion} and further discussion in \cref{rem:quillengroupcompletion}.
\end{rem}

At this point, we want to return to the story with loop-spectra. As a direct consequence of the fact that $\QS_\Gamma(\cH^{\otimes \infty})$ is an $\sE_\infty$-space, we have the following result, which is discussed in \cref{sec:loopspectra}.
\begin{thm*}[\ref{thm:QSinfiniteloop}]
If $\QS_\Gamma$ is an invertible quantum state type, then $\QS_\Gamma(\cH^{\otimes \infty})$ is an infinite loop space, i.e., the zero space of a loop-spectrum.
\end{thm*}

Now that we have established our framework, we reconnect with the questions of computing parametrized phases. For an invertible quantum state type $\QS_\Gamma$, the group of phases  always corresponds to  $\pi_0\QS_\Gamma(\cH^{\otimes \infty})$. Furthermore, we have shown that $\QS_\Gamma(\cH^{\otimes \infty})$ is an $\sE_\infty$-space. 

May's Recognition Principle \cite{MayGILS} gives a correspondence between infinite loop spaces whose monoid of path components is a group and \emph{connective} loop-spectra.  These are loop-spectra $\bY$ whose homotopy groups vanish in all negative degrees. That is, $\pi_{\dd}\bY=\pi_0 Y_{-\dd}=0$ for all $\dd<0$. This feels very promising. Homotopy groups of spectra and other terminology about loop-spectra are reviewed in \cref{sec:loopspectra}.

\bigskip
We return to the physical context, fixing a space-time dimension $\dd\pone$ and letting $\Gamma$ be $\Z^\dd$. The hope is that with clever choices of invertible quantum state types $\QS_{\Z^\dd}$, the space $\QS_{\Z^d}(\cH^{\otimes \infty})$ will give a first approximation for the classification of gapped invertible phases in spacetime dimension $\dd\pone$ and that the $\sE_\infty$-space structure will be related to Kitaev's loop-spectrum. 
However, even with a reasonable choice of $\QS_{\Z^d}$, there will still be important open questions to resolve:
\begin{enumerate}
\item We have not established a connection between $\QS_{\Z^\dd}$ and $\QS_{\Z^{\dd+1}}$. So, all we have is a collection of loop-spectra that have not been shown to interact with one another. 
Establishing this connection will necessarily be non-formal input depending on the choice of invertible quantum state type $\QS_{\Z^d}$.

\item We have not imposed any spatial homogeneity on quantum systems, as would be required for instance by translation symmetry. It is possible that pathological behavior may arise from highly inhomogeneous infinite systems, which could result in a classification that does not model physical expectations.

\item   The spectra $\bGP$ believed to classify invertible gapped phases (of various flavors) will not be connective! Indeed, $\pi_{0}\GP_\dd$ is meant to classify gapped invertible phases in space-time dimension $\dd\pone$ and these are known to be non-trivial. However, the loop-spectrum corresponding to $\QS_{\Z^\dd}$ will be 
a shift of a ``Whitehead cover'' of $\bGP$. 
In particular, in the range $0\leq i\leq \dd$, we do recover the phases since
\[ \pi_0 \GP_{i} \cong    \pi_{i} \QS_{\Z^\dd}(\cH^{ \otimes \infty}) \]
and, in fact, if one is careful, one can also recover any Postnikov $k$-invariant of $\bGP$ from $\QS_{\Z^{\dd+N}}( \otimes \cH^{\infty})$ for $N$ large enough with respect to the degree of the Postnikov $k$-invariant.
\end{enumerate}

\begin{ex}
For gapped invertible bosonic phases with no symmetry, the loop-spectrum $\bGP$ is believed to be
\[\Sigma^2 \boldsymbol{  I_\Z MSO}.\]
This spectrum is almost the \emph{opposite} of connective since its homotopy groups are trivial for $\dd>2$ and non-trivial for infinitely many values of $\dd\leq 2$.  

In spacetime $0+1$, the universal quantum state type is physically relevant, and we can consider $\sP_{\Z^0}((\C^2)^{\otimes \infty})$. It is not hard to show that the loop-spectrum we obtained from this $\sE_\infty$-space is 
\[\Sigma^2\boldsymbol{H\Z} \simeq  \Sigma^2 \boldsymbol{I_\Z MSO}_{\tau {\geq 0}} ,\]
where $\boldsymbol{H\Z}$ is (up to weak equivalence) a spectrum 
\[(K(\Z,0),K(\Z,1),K(\Z,2), \ldots)\]
with choices of weak equivalences $K(\Z,n) \xrightarrow{\simeq}\Omega K(\Z,n+1)$ for each $n$.
\end{ex}

\begin{ex}\label{ex:physicalexamples}
To capture physically meaningful states, we can start by considering the quantum state type $\GS_{\Z^\dd}$, which assigns to $\cH$ the subspace  $\GS_{\Z^\dd}(\cH) \subset \sP_{\Z^\dd}(\cH)$ consisting of ground states of gapped local Hamiltonians. Not all such states are invertible, so we must restrict this space further. Note that all factorized states will be in the same path component of $\GS_{\Z^\dd}(\cH)$, and we define
\[ \PGS_{\Z^\dd}(\cH)\subset \GS_{\Z^\dd}(\cH)\]
to be the subspace of those states $\omega\in  \GS_{\Z^\dd}(\cH)$ such that there exists $\tau \in  \GS_{\Z^\dd}(\cK)$ with $\omega\otimes \tau$ in the path component of the factorized states. 

Another concrete proposal for modeling a subset of invertible local gapped quantum systems via a quantum state type would be to take Kapustin--Sopenko--Yang \cite{KapustinSopenkoYang} subspaces 
\[  \QS_{\Z^\dd}^{LGA}(\cH)  \subset  \sP_{\Z^\dd}(\cH)\] 
of states that are LGA-equivalent to a factorized state.  Physically, these can be thought of as those ground states of gapped local Hamiltonians for which the Hamiltonian is joined to that for a factorized ground state through a path of gapped Hamiltonians.  Some invertible states do not have this property (\emph{e.g.} the $E_8$ state in $d=2$; see \cref{rem:notPA}), so $\QS_{\Z^\dd}^{LGA}$ would not model all invertible states.  
\end{ex}

Given a quantum state type $\QS_\Gamma$, the first step to understanding the cohomology theory it represents is to compute the values of the reduced theory on spheres,  called the \emph{coefficients} of the cohomology theory.
This is equivalent to the problem of computing the homotopy groups of $\QS_\Gamma(\cH^{\otimes \infty})$.  Clearly, the answer depends on your choice of $\QS_\Gamma$ (although not on the choice of $(\cH, \psi)$).  The last section of this paper is dedicated to studying this question for the universal quantum state type $\sP_\Gamma$. We do not compute all of the homotopy groups, but we do make a step in this direction by computing the fundamental group.
In fact, \cref{sec:fundamentalgroup} is dedicated to proving the following result.

\begin{thm*}[\ref{thm:fundaUHFfinal}]
The pure state space of a UHF algebra in the weak$^*$ topology is simply connected. In particular, for any lattice $\Gamma$ and any state type $\cH$, $\sP_\Gamma(\cH)$ and $\sP_\Gamma(\cH^{\otimes \infty})$ are simply connected spaces.
\end{thm*}

The idea of the proof is to perform a homotopy of an arbitrary loop in $\sP_\Gamma(\cH)$ such that as the homotopy progresses we disentangle more and more sites from the rest of the lattice. We are inspired by the work of \cite{kishimoto_ozawa_sakai_2003}, in which it is shown that for any two pure states $\psi, \omega \in \sP(\fA)$ of a unital separable simple $C^*$-algebra $\fA$, there exists an automorphism $\alpha \colon \fA \to \fA$ and a continuous family of unitaries $U \colon [0,\infty) \to \Unitary(\fA)$ such that $\omega = \psi \circ \alpha$ and 
\[\alpha = \lim_{t \ra \infty} \Ad(U_t),\] where the limit is taken in the strong topology of $\Aut(\fA)$. In particular, there is a weak$^*$ continuous path from $\psi$ to $\omega$ implemented by a norm-continuous family of unitaries at every point except the last endpoint. Our proof that $\sP_\Gamma(\cH)$ is simply connected takes on a similar flavor.

\begin{remark}\label{rem:notPA}
Although the universal quantum state type $\sP_\Gamma$ has simply connected spaces $\sP_\Gamma(\cH^{\otimes \infty})$, we do not expect this to be true in general for the quantum state types that arise in condensed matter theory. 

For $\PGS_{\Z^\dd}$ as above, motivated by results in the physics literature \cite{Thouless1983,kitaevSimonsCenter, kitaevIPAM, Xiong,KS2,Shiozaki2022,qpump}, we expect that the fundamental groups $\pi_1(\PGS_{\Z^\dd}(\cH^{\otimes \infty}))$ are not generally trivial.  For example, gapped invertible phases of bosonic systems in $2\pone$ spacetime dimensions with no symmetry are believed to be classified by $\pi_0(\PGS_{\Z^2}(\cH^{\otimes \infty}))\cong \Z$, with the generator of the $\Z$ referred to as the $E_8$ state \cite{KitaevE8}.  In $3+1$ dimensions, we thus have an ``$E_8$ pump,'' which is a non-trivial system over $S^1$ that can be obtained from the $E_8$ state via a suspension construction \cite{qpump}. Therefore, we expect that $ \pi_1(\PGS_{\Z^3}(\cH^{\otimes \infty})) \cong \Z$ for $\PGS_{\Z^3}$ the quantum state type for gapped invertible ground states in $3\pone$ dimensions. 
In other words, we should have
\[\pi_0(\PGS_{\Z^2}(\cH^{\otimes \infty})) \cong\pi_0\Omega(\PGS_{\Z^3}(\cH^{\otimes \infty}))  \cong \pi_1(\PGS_{\Z^3}(\cH^{\otimes \infty})) \cong \Z\] 
However,  the fundamental groups are generally expected to be abelian. In fact, this expectation is already present when we assume that invertible phases are classified by an infinite loop space since all homotopy groups of an infinite loop space are abelian groups. For invertible quantum state types in the sense of this paper, we deduce this formally in \cref{lem:pi1abelian}.
\end{remark}


\bigskip
\begin{convention}\label{conv:top}
  In this paper, we work in the category $\CGTop$ of \emph{compactly generated} topological spaces in the sense of \cite[Appendix A]{LewisCGTOP}. These are also called \emph{weak-Hausdorff $k$-spaces}. There is a $k$-ification functor from weak-Hausdorff spaces to $k$-spaces. Products and mapping space objects in $\CGTop$ are defined to be the $k$-ification of the usual constructions. In particular, if $\CGTop$ is given the symmetric monoidal structure given by the $k$-ification of the cartesian product and mapping spaces are understood as the $k$-ification of the usual mapping space with the compact-open topology, then $\CGTop$ is a closed symmetric category in the sense that
  \[ \Map(X\times Y, Z) \cong \Map(X, \Map(Y,Z))\]
  in $\CGTop$.
  
   We let $\bCGTop$ be the category of based compactly generated topological spaces and based maps, with the convention that the base point of a space is non-degenerate (i.e., the inclusion of the base point is a closed embedding which is a neighborhood deformation retract).  In this case, we get an adjunction 
   \[ \Map_*(X\wedge Y, Z) \cong \Map_*(X, \Map_*(Y,Z))\]
   in $\bCGTop$, where $\Map_*$ denotes the $k$-ification of the space of based map topologized as a subspace of the space of all maps with the compact-open topology.  
\end{convention}


\subsection*{Acknowledgements}
This paper was the result of many conversations spanning over the last few years. We would like to particularly thank the following people for useful such conversations: Mike Hill, Mike Hopkins, Tyler Lawson, Thomas Nikolaus, Victor Nistor, Ang\'elica Osorno, Dylan Wilson, Inna Zakharevich.  

We would like to particularly thank J. Peter May for listening to our various proposals for the framework and helping us determine a good set-up for our quantum state types. His help has been invaluable for this project.

A special thanks also goes to one of our collaborators on closely related physics projects, Xueda Wen, who has spent many hours explaining various concepts to the mathematician co-authors of this paper.

The material in this paper is based upon work supported by the National Science Foundation (NSF) under Award No.\ DMS 2055501. 
D.~D.~Spiegel also acknowledges support in part by the NSF under Award No.\ DMS 2303063.
Finally, we thank the anonymous referee for useful comments.


\section{Functorial properties of quasi-local algebras and their states}\label{sec:Csetup}

The goal of this section is to construct the universal quantum state type functor $\sP_\Gamma$ of \cref{thm:uniquantstatetype}
and to define our notion of \emph{quantum state types}. 

\subsection{Quasi-local algebras}
In this section, we construct a strong monoidal functor
\[\fA_\Gamma \colon \fHilbo \to \cpCAlgo \]
from the category $\fHilbo$ of finite dimensional non-zero complex inner product spaces with linear isometries as morphisms
to the category $\cpCAlgo$ of nuclear unital $C^*$-algebras and unital completely positive linear maps (which are automatically bounded). 
Both categories are made symmetric monoidal using the tensor product, as the notation indicates, and we explain in detail how that works
for the codomain. In case we consider these categories without their monoidal structures we denote them by 
$\fHilb$ and $\cpCAlg$, respectively. Let us emphasize that the morphisms in $\fHilb$ are always embeddings, a
simple observation which will be crucial for the following.

\bigskip
We begin by reviewing definitions about quasi-local algebras and their states. 

In $C^*$-algebraic language, a \emph{quantum lattice system} consists of a
countable set $\Gamma$ called the \emph{lattice} and, for each
\emph{lattice site} $v\in \Gamma$, a finite-dimensional Hilbert space
$\hilbone_v$ which encodes the quantum mechanical states of a particle
or collection of particles at the site $v$. The space $\hilbone_v$ is often referred to as the \emph{single-site Hilbert space} associated with site $v$. The lattice $\Gamma$ is
endowed with the discrete topology.
For simplicity, we assume that all sites are described by the same Hilbert space
$\hilbone$. 

We first define the \emph{quasi-local algebra} of a quantum lattice system. 
This is a $C^*$-algebra whose self-adjoint elements are regarded as the
observables of the system. Later we will consider states on the quasi-local algebra. 

Let us write $\Lambda \Subset \Gamma$ to indicate that $\Lambda$ is a nonempty finite subset of $\Gamma$. Then for any $\Lambda \Subset \Gamma$, we define
\[
\hilbone_\Lambda = \bigotimes_{v \in \Lambda} \hilbone \qqtext{and} \fA_\Lambda(\hilbone) = \bigotimes_{v \in \Lambda} \cB(\hilbone) \cong \cB(\hilbone_\Lambda).
\]
The (self-adjoint) elements of the  $C^*$-algebra $\fA_\Lambda(\hilbone)$ can be understood as the observables of the quantum lattice system localized at $\Lambda$. 

When $\Lambda_1 \subset \Lambda_2 \Subset \Gamma$, there exists a canonical map
\[ \iota_{\Lambda_2\Lambda_1}^\hilbone \colon \fA_{\Lambda_1}(\hilbone) \rightarrow \fA_{\Lambda_2}(\hilbone)\]
defined by tensoring identity operators $\mathds{1}$  on lattice sites in $\Lambda_2 \setminus \Lambda_1$. The $C^*$-algebras $\fA_\Lambda(\hilbone)$ together with these inclusions then form a directed system in the category of $C^*$-algebras and $*$-homomorphisms. We define the \emph{quasi-local algebra}
of the quantum lattice system as the colimit 
\[
\fA_\Gamma (\hilbone) = \colim_{\Lambda \Subset \Gamma} \fA_\Lambda(\hilbone)
\]
together with the canonical maps $\iota^\hilbone_\Lambda \colon \fA_\Lambda(\hilbone) \rightarrow \fA_\Gamma(\hilbone)$.
Note that the inclusions $\iota_{\Lambda_2\Lambda_1}^\hilbone$ and $\iota_\Lambda^\hilbone$ are all unital and injective. The quasi-local algebra has a dense $*$-subalgebra
\[
\fA_{\Gamma,\tn{loc}}(\hilbone) = \bigcup_\Lambda \iota_\Lambda^\hilbone\qty(\fA_\Lambda(\hilbone))
\]
called the local algebra.

\bigskip
Given finite-dimensional Hilbert spaces $\hilbone_1$ and $\hilbone_2$, one can form their (completed) tensor product
$\hilbone_1 \otimes \hilbone_2$. Physically, the tensor product corresponds to the composite of the
quantum mechanical systems described by $\hilbone_1$ and $\hilbone_2$ at a given lattice site.
The $C^*$-algebra $\fA_\Gamma(\hilbone_1 \otimes \hilbone_2)$ is then interpreted as the  quasi-local algebra of
the quantum lattice system 
obtained by \emph{internally stacking} the systems $\fA_\Lambda(\hilbone_1)$ and $\fA_\Lambda(\hilbone_2)$.
But one can also form the (completed) tensor product $\fA_\Gamma (\hilbone_1) \otimes \fA_\Gamma (\hilbone_2)$
and regard it as the  quasi-local algebra obtained by
\emph{externally stacking} the systems $\fA_\Lambda(\hilbone_1)$ and $\fA_\Lambda(\hilbone_2)$.
Note that by nuclearity of the $C^*$-algebras $\fA_\Gamma (\hilbone_i)$, $i=1,2$, the completed tensor product
$\fA_\Gamma (\hilbone_1) \otimes \fA_\Gamma (\hilbone_2)$ is uniquely determined. 
That the two quasi-local algebras obtained by  internally and externally stacking
are  isomorphic is the first result we state. Although
we could not find a reference in the literature, this is most certainly folklore.

\begin{proposition}\label{prop:canisoalg}
  For any two finite dimensional Hilbert spaces  $\hilbone_1$ and $\hilbone_2$, there is a
  canonical isomorphism
\[
  \eta^{\fA_\Gamma}_{\hilbone_1,\hilbone_2} \colon \fA_\Gamma (\hilbone_1) \otimes \fA_\Gamma (\hilbone_2) \xrightarrow{\cong} \fA_\Gamma (\hilbone_1 \otimes \hilbone_2) \ .
\]
\end{proposition}

\begin{proof}
  Note first  that by the universal property of the tensor product there is a canonical embedding
  $\cB(\hilbone_1) \otimes \cB(\hilbone_2) \hookrightarrow \cB(\hilbone_1 \otimes \hilbone_2)$ which
   is an isomorphism since $\hilbone_1$ and $\hilbone_2$ are finite dimensional.
  For any $\Lambda \Subset \Gamma$, one therefore obtains a natural $*$-isomorphism as the composite
  \begin{align*}
    \eta^{\Lambda}_{\hilbone_1,\hilbone_2}  \colon & \fA_\Lambda(\hilbone_1) \otimes \fA_\Lambda(\hilbone_2)  =  
    \qty(\bigotimes_{v \in \Lambda} \cB(\hilbone_1)) \otimes \qty(\bigotimes_{v \in \Lambda} \cB(\hilbone_2)) \\
    & \xrightarrow{ \cong} \bigotimes_{v \in \Lambda} \qty(\cB(\hilbone_1) \otimes \cB(\hilbone_2))
    \xrightarrow{\cong} \bigotimes_{v \in \Lambda} \cB(\hilbone_1 \otimes \hilbone_2)  = \fA_\Lambda(\hilbone_1 \otimes \hilbone_2) 
  \end{align*}
  where the first isomorphism is the obvious one.
  For $\Lambda_1 \subset \Lambda_2 \Subset \Gamma$, these isomorphisms form a commutative square
\[
\begin{tikzcd}[column sep = large, row sep = large]
  \fA_{\Lambda_2}(\hilbone_1) \otimes \fA_{\Lambda_2}(\hilbone_2)\rar["\eta^{\Lambda_2}_{\hilbone_1,\hilbone_2}"]& \fA_{\Lambda_2}(\hilbone_1 \otimes \hilbone_2)  \\
  \fA_{\Lambda_1}(\hilbone_1) \otimes \fA_{\Lambda_1}(\hilbone_2)\arrow[u,hook]\rar["\eta^{\Lambda_1}_{\hilbone_1,\hilbone_2}"]&
\fA_{\Lambda_1}(\hilbone_1 \otimes \hilbone_2) \arrow[u,hook]
\end{tikzcd}
\]
By nuclearity of the $C^*$-algebras $\fA_\Gamma (\hilbone_i)$ and since the colimit of a countable strict
inductive system of nuclear $C^*$-algebras is nuclear \cite[Thm.~5]{Takesaki1964}, forming the colimits
on both sides gives the desired isomorphism $\eta^{\fA_\Gamma}_{\hilbone_1,\hilbone_2}$.
\end{proof}

Next we want to show that associating to every finite dimensional Hilbert
space $\hilbone$ the quasi-local algebra $\fA_\Gamma (\hilbone)$ over the lattice
$\Gamma$ can be turned into a functor. The domain of that functor is the
category $\fHilb$ of non-zero finite dimensional Hilbert spaces and linear isometries. 
We therefore need to define what we mean by $\fA_\Gamma (f)$ where
$f:\hilbone \to\hilbtwo$ is a linear isometry between Hilbert spaces.
To this end we put for $\Lambda \Subset \Gamma$ and $P \in \mathfrak{B}(\hilbone)$
\[
\hilbone_\Lambda = \bigotimes_{v \in \Lambda} \hilbone, \qquad \hilbtwo_\Lambda = \bigotimes_{v \in \Lambda} \hilbtwo, \qquad \isoone_\Lambda = \bigotimes_{v \in \Lambda} \isoone, \quad \tn{and} \quad P_\Lambda = \bigotimes_{v \in \Lambda} P \ .
\] 
We implicitly use the identification $\fA_\Lambda(\hilbone) \cong \cB(\hilbone_\Lambda)$. Furthermore, we let $\Ad (A)$ denote the adjoint action 
\[\Ad(A)(B) = ABA^*\] whenever $A$ is a bounded linear operator between Hilbert spaces or when $A$ is an element of a $C^*$-algebra. 

\begin{rem}
Note that $\isoone^*\isoone = \1$ and $\isoone\isoone^* = P$, where $P \in \cB(\hilbtwo)$ is the orthogonal projection onto the image of $\isoone$. 
Since $\isoone_\Lambda^*\isoone_\Lambda = \1$ and $\isoone_\Lambda \isoone_\Lambda^* = P_\Lambda$,  we have 
\[\Ad(\isoone_\Lambda^*) \circ \Ad(\isoone_\Lambda) = \id_{\fA_\Lambda(\hilbone)} \quad \quad \text{and}\quad \quad \Ad(\isoone_\Lambda) \circ \Ad(\isoone_\Lambda^*) = \Ad(P_\Lambda).\]
\end{rem}

The following results establish the functoriality of $\fA_\Gamma$ on $\fHilb$.
\begin{lemma}\label{lem:isometry_adjoint_colimit}
  There exists a unique bounded linear map
  \[ \fA_\Gamma (f) \colon \fA_\Gamma(\hilbtwo) \rightarrow \fA_\Gamma(\hilbone) \]
  such that, for all $\Lambda \Subset \Gamma$, the diagram 
\begin{equation}\label{eq:isometry_adjoint_colimit}
\begin{tikzcd}
\fA_\Gamma(\hilbtwo) \arrow[r,"\fA_\Gamma(f) "] & \fA_\Gamma(\hilbone)\\
\fA_{\Lambda}(\hilbtwo) \arrow[u,hook,"\iota_\Lambda^\hilbtwo"]\arrow[r,"\Ad(\isoone_{\Lambda}^*)"']& \fA_{\Lambda}(\hilbone)\arrow[u,hook,"\iota_\Lambda^\hilbone"']
\end{tikzcd}
\end{equation}
commutes.
Furthermore, $\fA_\Gamma(f)$ is completely positive and unital.
\end{lemma}

\begin{proof}
Observe that $\Ad(\isoone_\Lambda^*)$ is linear and positive, as well as unital by the fact that $\isoone_\Lambda^*\isoone_\Lambda = \1$. Again from the fact that $\isoone^*\isoone = \1$, it follows that the diagram
\[
\begin{tikzcd}
\fA_{\Lambda_2}(\hilbtwo) \arrow[r,"\Ad(\isoone_{\Lambda_2}^*)"] & \fA_{\Lambda_2}(\hilbone)\\
\fA_{\Lambda_1}(\hilbtwo) \arrow[u,hook,"\iota_{\Lambda_2\Lambda_1}^\hilbtwo"]\arrow[r,"\Ad(\isoone_{\Lambda_1}^*)"']& \fA_{\Lambda_1}(\hilbone)\arrow[u,hook,"\iota_{\Lambda_2\Lambda_1}^\hilbone"']
\end{tikzcd}
\]
commutes for all finite subsets $\Lambda_1 \subset \Lambda_2 \Subset \Gamma$.
Therefore, by construction of $\fA_{\Gamma,\tn{loc}}(\hilbtwo)$, there exists a unique function 
\[ \fA_\Gamma(f)\colon \fA_{\Gamma,\tn{loc}}(\hilbtwo) \to \fA_\Gamma(\hilbone)\] 
such that \eqref{eq:isometry_adjoint_colimit} commutes when $\fA_\Gamma(\hilbtwo)$ is replaced by $\fA_{\Gamma,\tn{loc}}(\hilbtwo)$. 

Linearity and boundedness of $ \fA_\Gamma(f)$ follow  from linearity and boundedness of the maps $\Ad(\isoone_{\Lambda}^*)$ and from the structure of $\fA_\tn{loc}(\hilbtwo)$. Therefore $ \fA_\Gamma(f)$ has a unique extension to a bounded linear map 
\[ \fA_\Gamma(f)\colon \fA_\Gamma (\hilbtwo) \to \fA_\Gamma (\hilbone).\] 
From the fact that any $\Ad(\isoone_\Lambda^*)$ is unital we see that $ \fA_\Gamma(f)$ is unital. From the fact that each $\Ad(\isoone_\Lambda^*)$ is completeley positive by \cite[Lem.~1.2.2]{StormerBook} we see that $ \fA_\Gamma(f)$ is completely positive on $\fA_\tn{loc}(\hilbtwo)$.
One concludes that $ \fA_\Gamma(f)$ is completely positive on $\fA_\Gamma (\hilbtwo)$ as follows.

If $A$ is a positive element of the matrix algebra
$M_n(\fA_\Gamma (\hilbtwo)) \cong \fA_\Gamma (\hilbtwo) \otimes M_n (\C)$, then there exist
a $B \in M_n(\fA_\Gamma (\hilbtwo))$ such that $A =B^*B$ and a sequence $(B_k)_{k\in \N}$ in 
$M_n(\fA_\tn{loc} (\hilbtwo))$ converging to $B$. By positivity of $ \fA_\Gamma(f) \otimes \id_{M_n(\C)}$
we know that each $(\fA_\Gamma(f) \otimes \id_{M_n(\C)})(B_k^*B_k)$ is positive.
On the other hand, the sequence $\big((\fA_\Gamma(f) \otimes \id_{M_n(\C)})(B_k^*B_k)\big)_{k\in\N}$ converges
to $\big(\fA_\Gamma(f) \otimes \id_{M_n(\C)}(A)\big)_{k\in\N}$. Since the limit of
a sequence of positive elements is again positive, $(\fA_\Gamma(f) \otimes \id_{M_n(\C)})(A)$
is a positive element of $M_n(\fA_\Gamma (\hilbone))$ and the claim is proved.
\end{proof}

\begin{prop} 
There is a contravariant functor
\[
  \fA_\Gamma\colon \fHilb \to C^*\text{-}\category{Alg}_{1,+}.
\]
On objects and morphisms it is defined by 
 \begin{align*}
  \hilbone &\mapsto \fA_\Gamma (\hilbone)  \\
  \big( f\colon \hilbone \to \hilbtwo \big)& \mapsto
  \big( \fA_\Gamma (f) \colon \fA_\Gamma (\hilbtwo) \rightarrow \fA_\Gamma (\hilbone) \big).
  \end{align*}
Furthermore, there is a canonical isomorphism 
$\fA_\Gamma (\bbC) \cong \bbC$. 
\end{prop}
\begin{proof}
If $\hilbone=\hilbtwo$ and $\isoone=\id_\hilbone$, then it is clear from uniqueness in Lemma \ref{lem:isometry_adjoint_colimit} that $ \fA_\Gamma (f) = \id_{\fA(\hilbone)}$. 
If we have two isometries $f$ and $g$ such that $g\circ f$ is defined, 
 then it is again clear from uniqueness in Lemma \ref{lem:isometry_adjoint_colimit} that 
\[ \fA_\Gamma (\isotwo \circ \isoone) =\fA_\Gamma ( \isoone) \circ \fA_\Gamma ( \isotwo ).\]
The isomorphism $\fA_\Gamma (\bbC) \cong \bbC$ is the unique unital $*$-homomorphism between these
$C^*$-algebras.
\end{proof}

\begin{rem}
  If we equip the category $\cpCAlg$ of nuclear unital $C^*$-algebras and completely positive maps with the $C^*$-tensor product it in fact
  becomes a symmetric monoidal category. 
  
  To verify this
  recall first that both the minimal and the maximal $C^*$-tensor products are associative
  by \cite[11.1 \& 11.3]{KadisonRingroseII} and the universal property of the maximal $C^*$-tensor product.
  This implies in particular  that the $C^*$-tensor product of two nuclear $C^*$-algebras is again nuclear;
  see also \cite[IV.3.1.1.]{BlackadarOperatorAlgebras}.
  Moreover, the algebraic tensor product of two completely positive linear maps between nuclear $C^*$-algebras extends
  uniquely to a completely positive map between the $C^*$-tensor product algebras by e.g.\
  \cite[Prop.~3.5.3]{BrownOzawa}.  
  Therefore, the $C^*$-tensor product $\otimes$  is a bifunctor on the category
  $\cpCAlg$.
  By associativity of the minimal (or maximal) $C^*$-tensor product one has a natural associator map
  $\alpha_{\fA,\fB,\fC} : \fA\otimes (\fB \otimes \fC) \cong (\fA\otimes \fB) \otimes \fC$
  for nuclear $C^*$-algebras $\fA,\fB,\fC$.
  The monoidal unit is obviously given by $\C$. The coherence conditions for the algebraic tensor product naturally
  extend to the $C^*$-tensor product for nuclear $C^*$-algebras, and  $\cpCAlg$ endowed with $\otimes$
  becomes a monoidal category as claimed. 
  
  It is in fact symmetric since the swap map
  on the algebraic tensor product of $\fA$ and $\fB$ is continuous in the maximal (or minimal)
  $C^*$-tensor product norm, hence extends to an isomorphism $\fA \otimes \fB \cong \fB \otimes \fA$. 
  
  We denote the symmetric monoidal category of nuclear unital $C^*$-algebras and completely positive maps by $\cpCAlgo$. 
\end{rem}

Recall that $\fHilbo$ denotes the category $\fHilb$ of non-zero finite dimensional
  Hilbert spaces and linear isometric embeddings, equipped with the $\C$-linear tensor product as its symmetric monoidal structure.
  The unit is again $\C$.
\begin{proposition}\label{prop:strongmonoidalfunctoriality}
  The functor
  \[
   \fA_\Gamma\colon \fHilbo \to C^*\text{-}\category{Alg}_{1,+}^\otimes
 \]
 together with the natural isomorphism
 \[
  \eta^{\fA_\Gamma} \colon \fA_\Gamma (-) \otimes \fA_\Gamma (-) \xrightarrow{ \ \cong \ } \fA_\Gamma (- \otimes -) \ .
\]
of \cref{prop:canisoalg} is a strong monoidal functor.
\end{proposition}
\begin{proof}
  We have to verify that the isomorphisms $\eta^{\fA_\Gamma}_{\hilbone_1,\hilbone_2}$ 
  given in \cref{prop:canisoalg} for finite dimensional Hilbert
  spaces $\hilbone_i$, $i=1,2$, are  actually the components of a natural transformation.
  In other words we have to prove that for any pair of linear isometries 
  \[ f_i\colon \hilbone_i\to\hilbtwo_i,\quad i=1,2 \ , \]
  between non-zero finite dimensional Hilbert spaces the diagram
  \begin{equation}
    \label{dia:nattransstackingquantumlatticefunc}
  \begin{tikzcd}[row sep=large, column sep=large]
    \fA_{\Gamma}(\hilbone_1) \otimes \fA_{\Gamma}(\hilbone_2)\ar[r,"\eta^{\fA_\Gamma}_{\hilbone_1,\hilbone_2}"]&
    \fA_{\Gamma}(\hilbone_1 \otimes \hilbone_2)  \\
    \fA_{\Gamma}(\hilbtwo_1) \otimes \fA_{\Gamma}(\hilbtwo_2)\arrow[u,"f_1^*\otimes f_2^*"]\ar[r,"\eta^{\fA_\Gamma}_{\hilbtwo_1,\hilbtwo_2}"]&
\fA_{\Gamma}(\hilbtwo_1 \otimes \hilbtwo_2) \arrow[u,"(f_1\otimes f_2)^*",swap] 
  \end{tikzcd}
  \end{equation}
 commutes, where for brevity of notation we have written $f^*$ instead of $\fA_\Gamma(f)$ for any linear
 isometry $f:\hilbone \to \hilbtwo$. To this end we further abbreviate and write $f^*$ for the maps
 $\Ad (f_\Lambda^*) :\fA_\Lambda (\hilbtwo) \to \fA_\Lambda (\hilbone) $ with $\Lambda \Subset \Gamma$ 
 which were introduced in  \cref{lem:isometry_adjoint_colimit}. 
 By construction of the completely positive maps $f_i^* = \Ad (f_{i,\Lambda}^*)$, $i=1,2,$ and
 $(f_1\otimes f_2)^* = \Ad \big( (f_1\otimes f_2)_\Lambda^*\big)$,
 the cube below commutes for  $\Lambda_1 \subset \Lambda_2 \Subset \Gamma$.
\[
  \begin{tikzcd}[row sep={60,between origins}, column sep={80,between origins}]
    & \fA_{\Lambda_2}(\hilbone_1) \otimes \fA_{\Lambda_2}(\hilbone_2) \ar[rr,"\eta^{\Lambda_2}_{\hilbone_1,\hilbone_2}"] \ar[from=dd,hook] &
    & \fA_{\Lambda_2}(\hilbone_1 \otimes \hilbone_2)   \\
    \fA_{\Lambda_2}(\hilbtwo_1) \otimes \fA_{\Lambda_2}(\hilbtwo_2)\ar[ru,"f_1^*\otimes f_2^*"]\ar[crossing over,rr,"\qquad\eta^{\Lambda_2}_{\hilbtwo_1,\hilbtwo_2}"]& & 
    \fA_{\Lambda_2}(\hilbtwo_1 \otimes \hilbtwo_2)\ar[ru,"(f_1\otimes f_2)^*",swap]  \\
    & \fA_{\Lambda_1}(\hilbone_1) \otimes \fA_{\Lambda_1}(\hilbone_2)\ar[rr,"\qquad\eta^{\Lambda_1}_{\hilbone_1,\hilbone_2}"] &
    & \fA_{\Lambda_1}(\hilbone_1 \otimes \hilbone_2)\ar[uu,hook]  \\
    \fA_{\Lambda_1}(\hilbtwo_1) \otimes \fA_{\Lambda_1}(\hilbtwo_2)\ar[uu,hook]\ar[rr,"\eta^{\Lambda_1}_{\hilbtwo_1,\hilbtwo_2}"]\ar[ru,"f_1^*\otimes f_2^*"]&
    & \fA_{\Lambda_1}(\hilbtwo_1 \otimes \hilbtwo_2)\arrow[crossing over,uu,hook] \ar[ru,"(f_1\otimes f_2)^*",swap]
\end{tikzcd}
\]
Passing to the four vertical colimits results in the four corners of Diagram
(\ref{dia:nattransstackingquantumlatticefunc}). By the universal property of the colimit,
commutativity of all cubes with $\Lambda_1 \subset \Lambda_2 \Subset \Gamma$ and by constrution
of the $\eta^{\fA_\Gamma}_{\hilbone_1,\hilbone_2}$ in \cref{lem:isometry_adjoint_colimit},
the corners may be filled with the maps between them making (\ref{dia:nattransstackingquantumlatticefunc})
commute.  Naturality of $\eta^{\fA_\Gamma}$ follows. 
\end{proof}

\subsection{States and pure states}

Mathematically, the states of a quantum mechanical system described by a $C^*$-algebra $\fA$
consist of the positive normalized linear functionals $\omega \colon\fA\to\C$.  
Given a quantum lattice system over the lattice $\Gamma$ and with site Hilbert space $\hilbone$,
we consider the sets 
\[\sS_\Gamma (\hilbone) := \sS(\fA_\Gamma(\hilbone)) \quad \quad \text{and} \quad \quad
\sP_\Gamma (\hilbone) := \sP(\fA_\Gamma(\hilbone))\] 
of states and pure states on
$\fA_\Gamma(\hilbone)$, respectively, equipped with the
weak$^*$ topology. 

There is an abundance of literature
on state spaces of $C^*$-algebras, see e.g.\ \cite{KadSR,Dixmier,AlfsenHancheOlsenShultz,BratteliRobinsonOAQSMI,EilCCCA,AlfsenShultzStateSpacesOpAlgs,AlfsenShultzGeometryStateSpacesOpAlgs,kishimoto_ozawa_sakai_2003,ContinuousKadison} and further references therein. 
We discussed in \cite[Thms.\ 1.20 \& 1.21]{ContinuousKadison} properties of $\sP(\fA_\Gamma(\hilbone))$.
For instance, it is a path connected and locally path connected topological space \cite{EilCCCA}. It is also a Polish space, meaning that it is separable and completely metrizable \cite{PedersenCAlgAutomorphisms}.
Since a first countable Hausdorff space is compactly generated, 
the pure state space $\sP(\fA_\Gamma(\hilbone))$ is a compactly generated topological space as in \cref{conv:top}. It is thus an object in our category $\CGTop$. 

\begin{rem}[{\cite{TakedaInductiveLimits} and \cite[Thm.~5.2]{ContinuousKadison}}]
  We will use the following useful fact about states on directed colimits of $C^*$-algebras.
  
 Consider a directed system of unital $C^*$-algebras $(\fA_i)_{i \in I}$ and injective unital $*$-homomorphisms $\iota_{ji}\colon\fA_i \rightarrow \fA_j$ defined whenever $i \leq j$. Let $\fA$ and the $*$-homomorphisms $\iota_i \colon \fA_i \rightarrow \fA$ form a directed colimit of this system. If $\omega \in \sS(\fA)$, then $\omega_i \defeq \omega \circ \iota_i \in \sS(\fA_i)$ for all $i \in I$, and $\omega_j \circ \iota_{ji} = \omega_i$ whenever $i \leq j$. Conversely, given states $\omega_i \in \sS(\fA_i)$ for all $i$ satisfying $\omega_j \circ \iota_{ji} = \omega_i$ whenever $i \leq j$, there exists a unique state $\omega \in \sS(\fA)$ such that $\omega_i = \omega \circ \iota_i$ for all $i$. Furthermore, if all of the states $\omega_i$ are pure, then so is $\omega$. 
  
 We caution, however, that purity of $\omega$ does not imply purity of the states $\omega_i$.
 By the Banach-Alaoglu theorem,  the state space
 of a unital $C^*$-algebra is compact in the weak$^*$ topology.
 Hence  the state space  $\sS(\fA)$  of the colimit $\fA$ of a directed system of unital $C^*$-algebras $(\fA_i)_{i \in I}$ and injective unital $*$-homomorphisms $\iota_{ji}\colon\fA_i \rightarrow \fA_j$ can be identified with the  limit of the inverse system of
  topological spaces   $(\sS(\fA_i))_{i \in I}$ and continuous maps $\iota_{ji}^*\colon\sS(\fA_j) \rightarrow \sS(\fA_i)$.
  Note that the maps $\iota_{ji}^*$ are surjective by e.g.\ \cite[Prop.\ 3.1.6]{PedersenCAlgAutomorphisms}.
\end{rem}

As a consequence of this remark we obtain the next auxiliary result.
\begin{proposition}
Let $\hilbone$ be a finite-dimensional Hilbert space $\hilbone$ and   
\begin{equation}
\label{eq:deffuncssobjects}
  \sS_\Lambda (\hilbone) :=\sS (\fA_\Lambda (\hilbone))
\end{equation}
for all $\Lambda \Subset \Gamma$. The state space $\sS_\Gamma (\hilbone) $  of the quasi-local algebra $\fA_\Gamma(\hilbone)$ is the directed colimit of the
  state spaces $(\sS_\Lambda (\hilbone))_{\Lambda \Subset \Gamma}$ and the natural continuous surjections
  \[ \big(\iota_{\Lambda_2\Lambda_1}^\hilbone\big)^*\colon\sS_{\Lambda_2} (\hilbone) \rightarrow \sS_{\Lambda_1}(\hilbone) \quad \quad \Lambda_1 \subset \Lambda_2 \Subset \Gamma \  . \]
\end{proposition}

We want to extend $\sS_\Gamma$ to a functor from the category $\fHilb$ to compact topological spaces.
So suppose $\hilbone$ and $\hilbtwo$ are non-zero finite-dimensional Hilbert spaces and $\isoone\colon \hilbone \to \hilbtwo$ is a linear isometry  (which need not be surjective).
We now define
\begin{equation}
\label{eq:deffuncssmorphisms}
  \sS_\Gamma  (f) \colon\sS_\Gamma (\hilbone) \rightarrow \sS_\Gamma (\hilbtwo ), \quad  \sS_\Gamma  (f) (\psi) = \psi  \circ  \fA_\Gamma (f). 
\end{equation}
Since $ \fA_\Gamma (f)$ is positive and unital by Lemma \ref{lem:isometry_adjoint_colimit}, this is well-defined. From the characteristic mapping property of the weak$^*$ topology we see that $ \sS_\Gamma  (f)$ is continuous when both the domain and codomain are given the weak$^*$ topology.

\begin{prop}\label{lem:alpha*_all_states}
The map $ \sS_\Gamma  (f)$ is injective and its image is given by
\begin{equation}\label{eq:alpha*_image_all_states}
 \Im \sS_\Gamma  (f)  = \qty{\omega \in \sS(\fA_\Gamma(\hilbtwo)) : \omega(\iota_\Lambda^\hilbtwo(P_\Lambda)) = 1 \tn{ for all $\Lambda \Subset \Gamma$}} \ ,
\end{equation}
where as before $P_\Lambda =\isoone_\Lambda \isoone_\Lambda^*$ is projection onto the image of $\isoone_\Lambda$.
Thus, $ \sS_\Gamma  (f)$ is a closed embedding.
\end{prop}

\begin{proof}
  If $\psi_1, \psi_2 \in \sS_\Gamma (\hilbone)$ and $\psi_1 \circ  \fA_\Gamma (f) = \psi_2 \circ  \fA_\Gamma (f)$,
  then restricting to any $\Lambda \Subset \Gamma$ yields 
\[
\psi_1 \circ \iota_\Lambda^\hilbone \circ \Ad(\isoone_\Lambda^*) = \psi_1 \circ  \fA_\Gamma (f) \circ \iota_\Lambda^\hilbtwo = \psi_2 \circ  \fA_\Gamma (f) \circ \iota_\Lambda^\hilbtwo = \psi_2 \circ \iota_\Lambda^\hilbone \circ \Ad(\isoone_\Lambda^*) \ .
\]
Composing with $\Ad(\isoone_\Lambda)$ on both sides yields $\psi_1 \circ \iota_\Lambda^\hilbone = \psi_2 \circ \iota_\Lambda^\hilbone$. Since this holds for all $\Lambda$ we conclude that $\psi_1 = \psi_2$, so $ \sS_\Gamma  (f)$ is injective. Since $\sS_\Gamma (\hilbone)$ and $\sS_\Gamma (\hilbtwo)$ are compact Hausdorff, we know that $ \sS_\Gamma  (f)$ is an embedding. If we can show \eqref{eq:alpha*_image_all_states}, then it will follow that $ \sS_\Gamma  (f)$ is a closed embedding, since the right-hand side of \eqref{eq:alpha*_image_all_states} is manifestly weak$^*$ closed.

We now prove \eqref{eq:alpha*_image_all_states}. Since $\isoone_\Lambda \isoone_\Lambda^* = P_\Lambda$ and $\isoone_\Lambda^*\isoone_\Lambda = \1$, we have
\[
 \fA_\Gamma  (f)(\iota_\Lambda^\hilbtwo(P_\Lambda)) = \iota_\Lambda^\hilbone(\Ad(\isoone_\Lambda^*)(\isoone_\Lambda \isoone_\Lambda^*)) = \1.
\]
Thus, for any $\psi \in \sS_\Gamma (\hilbone)$, we have $\sS_\Gamma  (f)(\psi)(\iota_\Lambda^\hilbtwo(P_\Lambda)) = \psi(\1) = 1$.

Now suppose $\omega \in \sS_\Gamma (\hilbtwo)$ and $\omega(\iota_\Lambda^\hilbtwo(P_\Lambda)) = 1$ for all $\Lambda \Subset \Gamma$. Let $\omega_\Lambda = \omega \circ \iota^\hilbtwo_\Lambda$. The fact that $\omega_\Lambda(P_\Lambda) = 1$ implies both that $\omega_\Lambda = \omega_\Lambda \circ \Ad(P_\Lambda)$ and that $\omega_\Lambda \circ \Ad(\isoone_\Lambda) \in \sS(\fA_\Lambda(\hilbone))$.  Given $\Lambda_1 \subset \Lambda_2 \Subset \Gamma$, we observe that
\begin{align*}
\omega_{\Lambda_2} \circ \Ad(\isoone_{\Lambda_2}) \circ \iota_{\Lambda_2 \Lambda_1}^\hilbone &= \omega_{\Lambda_2} \circ \Ad(\isoone_{\Lambda_2}) \circ \iota_{\Lambda_2 \Lambda_1}^\hilbone \circ \Ad(\isoone_{\Lambda_1}^*) \circ \Ad(\isoone_{\Lambda_1})\\
&= \omega_{\Lambda_2} \circ \Ad(\isoone_{\Lambda_2}) \circ \Ad(\isoone_{\Lambda_2}^*) \circ \iota_{\Lambda_2\Lambda_1}^\hilbtwo \circ \Ad(\isoone_{\Lambda_1})\\
&= \omega_{\Lambda_2} \circ \iota_{\Lambda_2\Lambda_1}^\hilbtwo \circ \Ad(\isoone_{\Lambda_1}) \\
&= \omega_{\Lambda_1} \circ \Ad(\isoone_{\Lambda_1}).
\end{align*}
Thus, there exists a unique $\psi \in \sS_\Gamma (\hilbone)$ such that $\psi \circ \iota^\hilbone_\Lambda = \omega_\Lambda \circ \Ad(\isoone_\Lambda)$ for all $\Lambda \Subset \Gamma$. 

We claim that $ \sS_\Gamma  (f)(\psi) = \omega$. Indeed, observe that for any $\Lambda \Subset \Gamma$, 
\[
\psi \circ  \fA_\Gamma (f) \circ \iota_\Lambda^\hilbtwo = \omega_\Lambda \circ \Ad(\isoone_\Lambda) \circ \Ad(\isoone_\Lambda^*) =\omega_\Lambda = \omega \circ \iota^\hilbtwo_\Lambda
\]
This proves the claim, and completes the proof of \eqref{eq:alpha*_image_all_states}.
\end{proof}

Since $\fA_\Gamma$ is a contravariant functor, the construction of $\sS_\Gamma$ by Equations
\eqref{eq:deffuncssobjects} and \eqref{eq:deffuncssmorphisms} entails that $\sS_\Gamma$ can be understood
as a functor from the category  $\fHilb$ to the category $\CGTop$ of compactly generated topological spaces.
But there is more structure. The categories $\fHilb$ and $\CGTop$ are both \emph{topologically enriched}.
Such an enrichment is a choice of compactly generated topology on the morphism sets which makes the composition
pairings continuous. 
The topological enrichment on $\fHilb$ is obtained by noting that the set of linear isometric maps, which we
denote by $\I(\cH,\cK)$, is naturally a subspace of the space of bounded linear maps.
The set $\Map (X,Y)$ of continuous maps between two compactly generated spaces $X$ and $Y$ is endowed with
the $k$-ification of the compact-open topology as explained in \cref{conv:top}.
Recall that for every topological space $Z$ the $k$-\emph{ification} $kZ$ can be understood as the coarsest
topology on the set $Z$ such that $kZ$ is compactly generated and such that the identitiy map
$kZ\to Z$ is continuous.  Note that one has $kZ = Z$ for metrizable $Z$.
This implies that if $X$ is compact and $Y$ is metrizable, then the compact-open topology on  $\Map (X,Y)$  is metrizable
hence already compactly generated. So $k$-ification does not change the topology in this case. 
We will silently make use of these observations. 
By endowing morphism sets $\Map (X,Y)$ with the $k$-ification of the compact-open topology,
$\CGTop$ becomes a topologically enriched category as well. 

Together, these observations imply the following.
\begin{thm}\label{thm:statespacefuntopoenriched}
The assignment which maps a finite dimensional Hilbert space  $\hilbone$ to $\sS_\Gamma(\cH)$ and a linear isometry $f$ to  $\sS_\Gamma(f)$ is a topologically enriched covariant functor
\[ \sS_\Gamma \colon \fHilb \to \CGTop \]
with the property that, for each $\hilbone$, $\sS_\Gamma(\hilbone)$ is a compact space and for each $f$, $\sS_\Gamma(f)$ is a closed embedding.

Furthermore, there is a natural transformation
\[\eta^{\sS_\Gamma}
  \colon   \sS_\Gamma(-) \times \sS_\Gamma(-) \to \sS_\Gamma(- \otimes -)
\]
which sends $(\omega_1,\omega_2)\in  \sS_\Gamma(\hilbone_1) \times \sS_\Gamma(\hilbone_2)$ to
$(\omega_1\otimes \omega_2) \circ (\eta_{\hilbone_1, \hilbone_2}^{\fA_\Gamma})^{-1}$.
\end{thm}
\begin{proof}
  The only parts that requires further justification are the topological enrichments and the construction and
  naturality of $\eta^{\sS_\Gamma}$. Let us start with the construction of $\eta^{\sS_\Gamma}$. Using
  the natural isomorphism $\eta^{\fA_\Gamma}_{\hilbone_1,\hilbone_2}$ from Propositions \ref{prop:canisoalg} and
  \ref{prop:strongmonoidalfunctoriality} we put
  \[
    \eta_{\hilbone_1,\hilbone_2}^{\sS_\Gamma}(\omega_1, \omega_2) =
    (\omega_1 \otimes \omega_2) \circ (\eta_{\hilbone_1, \hilbone_2}^{\fA_\Gamma})^{-1}\ .
  \]
  For $A_i \in \fA_\Gamma(\hilbone_i)$, $i=1,2$, we then obtain 
  \[
    \eta^{\sS_\Gamma}_{\hilbone_1,\hilbone_2} (\omega_1,\omega_2)
    \left( \eta_{\hilbone_1, \hilbone_2}^{\fA_\Gamma} (A_1\otimes A_2) \right) = 
    (\omega_1 \otimes \omega_2)(A_1\otimes A_2) = \omega_1 (A_1) \cdot \omega_2(A_2) \ .
  \]
  This equation entails continuity of $\eta^{\sS_\Gamma}_{\hilbone_1,\hilbone_2} (\omega_1,\omega_2)$
  because the span of elements of the form $\eta^{\fA_\Gamma}_{\hilbone_1, \hilbone_2}(A_1 \otimes A_2)$ is dense in
  $\fA_\Gamma(\hilbone_1 \otimes \hilbone_2)$ and the equation shows
  that evaluating on these elements is jointly weak$^*$ continuous in $\omega_1$ and $\omega_2$.
  Naturality of $\eta^{\sS_\Gamma}$ follows from the naturality of $\eta^{\fA_\Gamma}$. 

It remains to prove that the functor $\sS_\Gamma$ is enriched. We need to show that the map
\begin{equation}\label{eq:statespacefuncmorphisms}
  \I(\hilbone,\hilbtwo) \to \Map (\sS_\Gamma(\hilbone) , \sS_\Gamma(\hilbtwo)), \quad f \mapsto \sS_\Gamma(f)
\end{equation}
is continuous. Since $\sS_\Gamma(\hilbone)$ is compact and
$\sS_\Gamma(\hilbtwo)$ is metrizable by separability of $\fA_\Gamma (\hilbtwo)$, the compact-open topology on
the set $\Map (\sS_\Gamma(\hilbone) , \sS_\Gamma(\hilbtwo))$ is metrizable and
coincides with the topology of uniform convergence. The space $\I(\hilbone,\hilbtwo)$ is a closed
subspace of the Banach space $\cB(\hilbone,\hilbtwo)$ of all linear maps from $\hilbone$ to $\hilbtwo$,
hence $\I(\hilbone,\hilbtwo)$ is a complete metric space. The strategy now is to construct an
appropriate metric $d$ on $\sS_\Gamma(\hilbtwo)$ inducing the weak$^*$ topology and such that
the map \eqref{eq:statespacefuncmorphisms} becomes Lipschitz continuous when
$\Map (\sS_\Gamma(\hilbone) , \sS_\Gamma(\hilbtwo))$ is endowed with the supremum metric
\[
  d_{\sS_\Gamma(\hilbone)}\colon \Map (\sS_\Gamma(\hilbone) , \sS_\Gamma(\hilbtwo)) \to \R, \:
  (F,G) \mapsto \sup_{\omega \in \sS_\Gamma(\hilbone)} d\big( F(\omega) ,  G (\omega)  \big) \ .
\]
By \cite[V.5.1]{DunfordSchwartz1}, a metric $d$ recovering the weak$^*$ topology can be obtained by choosing a
countable dense family $(A_k)_{k\in \N}$ in the unit ball of $\fA_\Gamma (\hilbtwo)$, a sequence of positive numbers
$\lambda_k \geq 1$ and defining
\[
  d (\psi,\omega ) = \sum_{k\in \N} \frac{1}{2^{k+1}\lambda_k}  \left| (\omega-\psi)(A_k) \right|  \quad
  \text{for all } \psi,\omega \in \sS_\Gamma(\hilbtwo) \ .
\]
We can assume without loss of generality that the elements $A_k$ are all contained in the local algebra $\fA_{\Gamma,\textup{loc}}(\hilbtwo)$.
For each $k$ then choose a finite $\Lambda_k \Subset \Gamma$ such that $A_k \in \fA_{\Lambda_k} (\hilbtwo)$
and put $\lambda_k := |\Lambda_k|$. 
Now compute for $f,g \in  \I(\hilbone,\hilbtwo) $ and  $\omega\in \sS_\Gamma(\hilbone)$:
\begin{equation}\label{eq:uniformconvergence}
\begin{split}
  d & (\sS_\Gamma (f) (\omega) , \sS_\Gamma (g) (\omega) )  =
  \sum_{k\in \N} \frac{1}{2^{k+1}|\Lambda_k|}  \left|  \big(\sS_\Gamma (f) - \sS_\Gamma (g)\big) (\omega) (A_k)   \right| \\
  & = \sum_{k\in \N} \frac{1}{2^{k+1}|\Lambda_k|}  \left|  \omega \circ \big(\fA_\Gamma (f) - \fA_\Gamma (g)\big) (A_k)   \right|   \\
  & = \sum_{k\in \N} \frac{1}{2^{k+1}|\Lambda_k|}  \left|  \omega \circ \big(\Ad(f_{\Lambda_k}^*) - \Ad ( g_{\Lambda_k}^* \big) (A_k)   \right|  \\
  & = \sum_{k\in \N} \frac{1}{2^{k+1}|\Lambda_k|}  \left|  \omega \big( f_{\Lambda_k}^*A_kf_{\Lambda_k} - g_{\Lambda_k}^*A_kg_{\Lambda_k} \big) \right|  \\
  & \leq   \sum_{k\in \N} \frac{1}{2^{k+1}|\Lambda_k|}\Big(
    \left\| \big( f_{\Lambda_k}^* -  g_{\Lambda_k}^* \big) A_kf_{\Lambda_k} \right\|
    + \left\| g_{\Lambda_k}^* A_k \big( f_{\Lambda_k} - g_{\Lambda_k} \big) \right\| \Big) \\
    & \leq   \sum_{k\in \N} \frac{1}{2^k |\Lambda_k|}
    \left\| f_{\Lambda_k} -  g_{\Lambda_k} \right\| \leq \sum_{k\in \N} \frac{1}{2^{k}} \| f - g \| \leq 2 \, \| f - g \| \ . 
\end{split}
\end{equation}
Hence  $ d_{\sS_\Gamma(\hilbone)} \big( \sS_\Gamma (f) , \sS_\Gamma (g) \big) \leq 2 \, \| f - g \|$
and the map sending $f\in \I(\hilbone,\hilbtwo)$ to $\sS_\Gamma(f) $ is Lipschitz continuous, so in particular continuous. 
\end{proof}

\begin{remark}
  The natural isomorphism
  $\eta^\fA\colon \fA_\Gamma ( - ) \otimes \fA_\Gamma ( - ) \to \fA_\Gamma ( - \otimes - )$
  provides a canonical identification of the internally and externally
  stacked quantum lattice systems, so from a physical and categorical standpoint they are the same.
  Henceforth, we suppress from now on to distinguish them notationally. Under this agreement, the natural
  transformation $\eta^{\sS_\Gamma}$ maps $(\omega_1,\omega_2)$ to $\omega_1\otimes \omega_2$. 
\end{remark}

We now restrict our attention to the pure states.

\begin{prop}\label{thm:alpha*_image_pure_states}
  For every linear isometry $f\colon \hilbone \to \hilbtwo$ between finite dimensional Hilbert spaces the map
  $ \sS_\Gamma  (f)$ maps pure states to pure states. The restriction
  $ \sP_\Gamma (f) := \sS_\Gamma  (f)|_{\sP_\Gamma (\hilbone)}:\sP_\Gamma (\hilbone) \rightarrow \sP_\Gamma (\hilbtwo)$ has image
\begin{equation}\label{eq:alpha*_image_pure_states}
  \Im \sP_\Gamma  (f) =
  \qty{\omega \in \sP_\Gamma (\hilbtwo) :
  \omega(\iota_\Lambda^\hilbtwo(P_\Lambda)) = 1 \tn{ for all $\Lambda \Subset \Gamma$}}
\end{equation}
and is a closed embedding. 
\end{prop}

\begin{proof}
  Suppose $\psi \in \sP_\Gamma (\hilbone)$ and suppose $\phi$ is a positive linear functional on $\fA_\Gamma(\hilbtwo)$ such that $\phi \leq \psi \circ  \fA_\Gamma  (f)$. Assuming $\phi \neq 0$, let $\widehat \phi = \phi/\norm{\phi}$, so that $\widehat \phi$ is a state. Since $\psi( \fA_\Gamma  (f)(\iota^\hilbtwo_\Lambda(\1 - P_\Lambda))) = 0$ for all $\Lambda \Subset \Gamma$, we know $\phi(\iota_\Lambda^\hilbtwo(\1 - P_\Lambda)) = 0$ for all $\Lambda \Subset \Gamma$, hence $\widehat\phi(\iota^\hilbtwo_\Lambda(P_\Lambda)) = 1$ for all $\Lambda$. By Lemma \ref{lem:alpha*_all_states} there exists
  $\chi \in \sS_\Gamma (\hilbone)$ such that $\widehat \phi =  \sS_\Gamma  (f)(\chi)$. 

We now observe that for any $\Lambda \Subset \Gamma$, 
\begin{align*}
\norm{\phi} \cdot \chi \circ \iota_\Lambda^\hilbone &=\norm{\phi} \cdot \chi \circ \iota^\hilbone_\Lambda \circ \Ad(\isoone_\Lambda^*) \circ \Ad(\isoone_\Lambda)\\
&= \norm{\phi} \cdot \chi \circ  \fA_\Gamma  (f) \circ \iota^\hilbtwo_\Lambda \circ \Ad(\isoone_\Lambda)\\
&= \phi \circ \iota^\hilbtwo_\Lambda \circ \Ad(\isoone_\Lambda)\\
&\leq \psi \circ  \fA_\Gamma  (f) \circ \iota^\hilbtwo_\Lambda \circ \Ad(\isoone_\Lambda) = \psi \circ \iota_\Lambda^\hilbone.
\end{align*}
It follows that $\norm{\phi} \cdot \chi \leq \psi$. Since $\psi$ is pure, we know there exists $t \in [0,1]$ such that $\norm{\phi}\cdot\chi = t \cdot \psi$. Applying $\1$ yields $t = \norm{\phi}$. Thus we obtain $\chi = \psi$, hence $\phi = \norm{\phi}  \sS_\Gamma  (f) (\chi) = \norm{\phi}  \sS_\Gamma  (f) (\psi)$. This proves that $ \sS_\Gamma  (f) (\psi)$ is pure.

Now suppose $\omega \in \sP_\Gamma (\hilbtwo)$ and $\omega(\iota^\hilbtwo_\Lambda(P_\Lambda)) = 1$ for all $\Lambda \Subset \Gamma$. There exists $\psi \in \sS_\Gamma (\hilbone)$ such that $\omega =  \sS_\Gamma  (f) (\psi)$. Suppose $\chi$ is a positive linear functional on $\fA(\hilbone)$ and $\chi \leq \psi$. Then $\chi \circ  \fA_\Gamma  (f) \leq \psi \circ  \fA_\Gamma  (f) = \omega$, so there exists $t \in [0,1]$ such that $\chi \circ  \fA_\Gamma  (f) = t \omega$. Applying this functional to $\1$ yields $t = \norm{\chi}$. Assuming $\norm{\chi} \neq 0$, let $\widehat{\chi} = \chi/\norm{\chi}$. Then on any $\Lambda \Subset \Gamma$ we have
\[
\psi \circ \iota_\Lambda^\hilbone \circ \Ad(\isoone_\Lambda^*) = \omega \circ \iota^\hilbtwo_\Lambda = \widehat \chi \circ  \fA_\Gamma  (f) \circ \iota^\hilbtwo_\Lambda = \widehat \chi \circ \iota^\hilbone_\Lambda \circ \Ad(\isoone_\Lambda^*).
\]
Applying $\Ad(\isoone_\Lambda)$ to both sides yields $\psi \circ \iota^\hilbone_\Lambda = \widehat \chi \circ \iota^\hilbone_\Lambda$. Since $\Lambda$ was arbitrary, we conclude that $\psi = \widehat \chi$, hence $\chi = \norm{\chi} \psi$. This proves that $\psi \in \sP_\Gamma (\hilbone)$, hence $\omega \in  \sS_\Gamma  (f)(\sP_\Gamma (\hilbone))$. This proves \eqref{eq:alpha*_image_pure_states}.

Since $ \sS_\Gamma  (f)$ is an embedding on the full state spaces, its restriction to the pure state spaces is an embedding. Its image is manifestly weak$^*$ closed in $\sP_\Gamma (\hilbtwo)$, so $ \sP_\Gamma  (f)$ is a closed embedding.
\end{proof}

We finally put these results together and state the main result of this section, the construction of the \emph{universal quantum state type}.
\begin{theorem}\label{thm:uniquantstatetype}
  There is a topologically enriched covariant functor from the category of finite dimensional Hilbert spaces and
  isometric linear embeddings to the category of compactly generated topological spaces
\begin{align*}
  \sP_\Gamma\colon \fHilb \to \CGTop
  \end{align*}
  which, on objects and morphisms is given by
\begin{align*}
 \hilbone &\mapsto \sP_\Gamma (\hilbone) , \\
   \big( f\colon \hilbone \to \hilbtwo \big) &\mapsto
  \big( \sP_\Gamma (f) \colon \sP_\Gamma (\hilbtwo) \rightarrow \sP_\Gamma (\hilbone) \big).
\end{align*}
It has the property that, for any  isometric linear embedding $f$, $\sP_\Gamma(f)$ is a closed embedding.

Furthermore, there is a natural transformation
\[\eta^{\sP_\Gamma} \colon   \sP_\Gamma( - ) \times \sP_\Gamma( - ) \to \sP_\Gamma( - \otimes - )  \]
which sends $(\omega_1,\omega_2)\in \sP_\Gamma(\hilbone_1) \times \sP_\Gamma(\hilbone_2)$ to
$\omega_1\otimes \omega_2$, so that $\sP_\Gamma$ is a lax monoidal functor
\[  \sP_\Gamma\colon \fHilbo \to \CGTop \]
where the monoidal product on $\CGTop$ is the cartesian product.
\end{theorem}

\begin{proof}
  The functoriality of $\sP_\Gamma$ follows immediately by definition of $\sP_\Gamma$ as the restriction of $\sS_\Gamma$
  to pure state spaces and by functoriality of $\sS_\Gamma$.
  The claim on closed embeddings is \cref{thm:alpha*_image_pure_states}.

  Analogously one constructs $\eta^{\sP_\Gamma}$ as the restriction of the natural transformation $\eta^{\sS_\Gamma}$ to pure
  state spaces. Since
  \begin{equation}
    \label{eq:defnaturaltransfprodtensorprod}
    \eta^{\sS_\Gamma} (\omega_1,\omega_2) = \omega_1\otimes \omega_2 \quad
    \text{for all } (\omega_1,\omega_2)\in \sP_\Gamma(\hilbone_1) \times \sP_\Gamma(\hilbone_2) 
  \end{equation}
  and since the tensor product of two pure   states is pure,   $\eta^{\sP_\Gamma}$  is well-defined and
  a natural transformation indeed. Moreover, Eqn.~\eqref{eq:defnaturaltransfprodtensorprod} entails that
  $\sP_\Gamma$ is a lax monoidal functor. 
 
  It remains to discuss the topological enrichment. Since the weak$^*$ topology on $\sP_\Gamma(\hilbtwo)$ is
  induced by a uniformity, more precisely by the weak$^*$ uniformity,
  the compact-open topology on $\Map ( \sP_\Gamma(\hilbone),\sP_\Gamma(\hilbtwo))$ coincides with the topology of compact 
  convergence or in other words with the topology of uniform convergence on compact subsets.
  The topology of uniform convergence is metrizable, so in particular compactly generated, and finer then the
  topology of uniform convergence on compacta.
  By the universal property of the $k$-ification, the topology of uniform convergence 
  therefore is finer than the topology we assume  $\Map ( \sP_\Gamma(\hilbone),\sP_\Gamma(\hilbtwo))$ to be endowed with,
  namely the $k$-ification of the compact-open topology. 
  Now if $(f_n)_{n\in \N}$ is a sequence in $\I(\hilbone,\hilbtwo)$ converging to $f : \hilbone\to\hilbtwo $, 
  then Eqn.~\eqref{eq:uniformconvergence} in the proof of \cref{thm:statespacefuntopoenriched}
  shows that the image sequence $\big(\sP_\Gamma (f_n)\big)_{n\in \N}$ converges uniformly on $\sP_\Gamma (\hilbone)$
  to $\sP_\Gamma (f)$. Therefore, the map
  \[ \I(\hilbone,\hilbtwo) \to \Map ( \sP_\Gamma(\hilbone),\sP_\Gamma(\hilbtwo)), \quad f \mapsto \sP_\Gamma (f) \]
  is continuous, and $\sP_\Gamma$ is a topologically enriched functor.  
\end{proof}

\subsection{Quantum state types}
With the results of the preceding sections in hand, we are ready to define the more general notion of a quantum state type.

The following definition will play an important role in the rest of the paper. It provides minimal assumptions necessary to deduce some formal homotopical properties of spaces of quantum states.

\begin{defn}\label{defn:quantumstatetype}
A \emph{quantum state type} on a lattice $\Gamma$ is a functor 
\[\QS_\Gamma \colon \fHilbo \to  \CGTop \]
such that the following hold:
\begin{enumerate}[(a)]
\item For each object $\cH$, $\QS_\Gamma(\cH)$ is a subspace of $\sP_\Gamma(\cH)$ containing all completely factorized states.
\item  For each morphism $f\colon \hilbone \to \hilbtwo$, $\QS_\Gamma(f)$  is the restriction of $\sP_\Gamma(f)$ and is a closed embedding.
\item The natural transformation $\eta^{\sP_\Gamma} $ restricts to a natural transformation
\[
\eta^{\QS_\Gamma}  \colon  \QS_\Gamma(\cH_1)\times \QS_\Gamma(\cH_2) \to  \QS_\Gamma(\cH_1\otimes \cH_2).
  \]
  \end{enumerate}
  The quantum state type $\sP_\Gamma$ is called the \emph{universal quantum state type}.
\end{defn}

We will show that a quantum state type
 $\QS_\Gamma$ naturally extends to a larger class of complex inner product spaces. 
We start with the following definition, which is the source for our extension of $\QS_\Gamma$.
\begin{defn}\label{defn:sJ}
  Let $\sJ$ be the category whose objects are non-zero complex vector spaces of countable algebraic dimension over $\C$,
  equipped with a Hermitian inner product. The infinite dimensional objects  are topologized as the topological union of
  their finite dimensional subspaces. 
  The morphisms are linear isometries $\I(\cV,\cW)$ topologized as subsets of the mapping space in $\CGTop$.
\end{defn}

In the following, we will examine the topological vector space structure of the objects of $\sJ$ in some
more detail. To this end we need a few additional notions from functional analysis and topological tensor
products which we briefly recall. For further information on locally convex vector spaces we
refer to \cite{KoetheTopVecSp1,Treves} and for details on topological tensor products we recommend the seminal work by Grothendieck \cite{GrothendieckTensorProducts} or the books \cite{Pietsch,Treves}. 

Recall first that by a \emph{Fr\'echet space} one understands a
complete and metrizable locally convex topological vector space. A \emph{Limit of Fr\'echet spaces} or
briefly an \emph{LF space}  is a locally convex topological vector space
$E$ which can be written as the colimit in the category of locally convex topological vector spaces
of a countable inductive system $(E_n , i_{n,m})$ of Fr\'echet spaces. 
If each of the maps $i_{n,m} :E_n\to E_m$ is an embedding of topological vector spaces, then
the colimit $E =\colim_{n\in \N} E_n$  is called a \emph{strict LF space}.
Note that in the older literature the terminology ``direct limit'' instead of ``colimit'' is used.
In case all the topological vector spaces in the inductive system $(E_n , i_{n,m})$ are Banach spaces,
one often calls the  colimit $E =\colim_{n\in \N} E_n$ an \emph{LB space}. 
See \cite[\S 19.~5.]{KoetheTopVecSp1} for more details on LB and LF spaces. 

To describe the monoidal structure of the category $\sJ$ we need an appropriate choice of a topological
tensor product. Unlike in the finite dimensional case or the case of Hilbert spaces  there are in general
several admissible locally convex topologies on the algebraic tensor product $E \otimes F$
of two locally convex vector spaces $E$ and $F$. 
The finest locally convex topology $\iota$ on $E \otimes F$  such that the map $E \times F \to E \otimes F $
is separately continuous has been called the \emph{inductive tensor product topology}
by Grothendieck \cite[Def.\ 3 in I \S 3.I]{GrothendieckTensorProducts}. One writes $E\otimes_\iota F$
for the tensor product endowed with the topology $\iota$. 
The finest locally convex topology $\pi$ on $E \otimes F$  such that the map $E \times F \to E \otimes F $
is (jointly) continuous is the \emph{projective topology}. The corresponding projective topological
tensor product is denoted by $E\otimes_\pi F$. 
The last topological tensor product which is of relevance here is the injective tensor product.
It is somewhat more complicated to be defined. 
Given two convex zero neighborhoods $U \subset E$ and $V\subset F$ one obtains a semi-norm
$\varepsilon_{U,V}$ on $E \otimes F$  by defining it for $t = \sum_{i=0}^n v_i \otimes w_i \in E\otimes F$
with  $v_i,\ldots,v_n \in E$, $w_0,\ldots,n \in F$ by 
\[
  \varepsilon_{U,V} (t) := 
  \sup \left\{ \left|\sum_{i=0}^n \lambda (v_i) \times \mu (w_i) \right| \, \Big| \: \lambda \in U^\circ  \: \& \: \mu \in V^\circ \right\} \ ,
\]
where 
$U^\circ$ and  $V^\circ$ denote the polar sets of $U$ and $V$, respectively.
Note that $ \varepsilon_{U,V} (t)$ is independent of the particular representation of $t$ as a sum of simple tensors.  
The semi-norms $\varepsilon_{U,V}$ generate a locally convex topology  $\varepsilon$ on $E\otimes F$
which is compatible with the tensor product in the sense that the canonical map
$E \times F \to E \otimes_\varepsilon F$ is (jointly) continuous. In particular,
$\varepsilon$ is weaker than the projective tensor product topology.
The topology $\varepsilon$ is called the \emph{injective tensor product topology},
the resulting topological tensor product $E \otimes_\varepsilon F$ the  \emph{injective tensor product} of
$E$ and $F$. Note that the three thus defined tensor product topologies can be completed to give the
locally convex topological vector spaces $E\hat{\otimes}_\iota F$,  $E\hat{\otimes}_\pi F$, and $E \hat{\otimes}_\varepsilon F$
called the \emph{completed inductive, projective and injective tensor product}, respectively.

By definition according to Grothendieck \cite[Def.\ 1]{GrothendieckTensorProductsSMF}, a \emph{nuclear space}  is a locally convex topological vector space $E$ such that for each
locally convex space $F$ the natural continuous linear map $E\otimes_\pi F \to E\otimes_\varepsilon F$ is
even a topological isomorphism. Note that if $E$ and $F$ are Fr\'echet spaces, then the natural map
 $E\otimes_\iota F \to E\otimes_\pi F$ is an isomorphism of topological vector spaces. 

Now we can formulate and prove our result on the  monoidal structure of the category $\sJ$ and the
topological structure of its objects.

\begin{proposition}
  The objects of $\sJ$ are nuclear strict LB spaces, so in particular locally convex. For two objects $\cV,\cW$ of $\sJ$,
  the colimit topology on the tensor product $\cV \otimes \cW$ with respect to the inductive system of
  finite dimensional subspaces coincides with Grothendieck's inductive tensor product topology $\iota$.
  Moreover, if  $(\cV)_{i\in I}$ and $(\cW_j)_{j\in J}$ are two countable strict inductive systems in  $\sJ$
  with colimits  $\cV =\colim_i\cV_i$ and $\cW = \colim_j \cW_j$, respectively, then 
  the canonical maps induce isomorphisms 
  \begin{equation}
    \label{eq:tensorproductscolimitscommute}
    \cV \otimes \cW \cong \colim_j  \cV \otimes \cW_j \cong \colim_i \cV_i \otimes \cW \cong  \colim_{(i,j)} \cV_i \otimes \cW_j.
  \end{equation}
That is, the tensor product in $\sJ$ commutes with colimits of strict inductive systems. 
  Finally, the category $\sJ$ is a symmetric monoidal category with product given by the tensor product and unit the vector
  space $\C$.
\end{proposition}

\begin{proof}
  Let $\cV$ be an object of  $\sJ$ and $(\cV_i)_{i\in \N}$ a sequence of finite dimensional subspaces $\cV_i \subset \cV$ such that
  $\cV= \bigcup_{i\in \N} \cV_i$. A subset $U\subset \cV$ is then open if and only if the intersection of
  $U$ with any $\cV_i$ is open in $\cV_i$. By local compactness of the $\cV_i$ one concludes that
  for any open  $U\subset \cV$ and any point $v \in U$ there is an open convex $V\subset U$ such that $v \in V$. 
  Hence $\cV$ is the locally convex colimit of  $(\cV_i)_{i\in \N}$, so in particular an LB space. By nuclearity of the $\cV_i$,
  the colimit $\cV$ then is nuclear as well.
  
  By \cite[Prop 14]{GrothendieckTensorProducts}, \eqref{eq:tensorproductscolimitscommute} holds true when the tensor product
  $\cV \otimes \cW$ of two colimits $\cV =\colim_i\cV_i$ and $\cW = \colim \cW_j$ is endowed with the inductive tensor product topology.
  Moreover, \eqref{eq:tensorproductscolimitscommute} also shows that the inductive tensor product topology on 
  $\cV \otimes \cW$ coincides with the locally convex colimit topology of the inductive system  $(\cV_i \otimes \cW_j)_{(i,j)\in I\times J}$ when
  $(\cV_i)_{i\in I} \subset \cV$ and $(\cW_j)_{j\in J}\subset\cW$ are countable families of
  finite dimensional subspaces whose union span $\cV$ and $\cW$, respectively.
  In other words, the inductive tensor product topology on $\cV \otimes \cW$ coincides with topology constructed
  according to \cref{defn:sJ}.
  Since the projective tensor product is associative and since the inductive and the projective tensor product topologies coincide on
  any object $\cV$ of  $\sJ$ by \cref{strongnuclearity} below,  the category $\sJ$ is in fact symmetric monoidal. Obviously, $\C$ is the unit
  in this monoidal category.
\end{proof}

\begin{rem}\label{strongnuclearity}
  Let $\cV$ and $\cW$ be two complex vector spaces of countable infinite dimension endowed with the
  colimit topologies according to \cref{defn:sJ}. Then the inductive, injective and projective tensor product
  topologies on the algebraic tensor product $\cV \otimes \cW$ all coincide. 
  
  Equality of the injective and
  projective tensor product topologies follows from nuclearity of  $\cV$ and $\cW$ 
  \cite[Def.\ 1]{GrothendieckTensorProductsSMF}. The tricky part is equality of the inductive and projective
  tensor product topologies which in general need not hold true for non-metrizable LF spaces.
  More precisely, one has to show that the canonical map $\cV \times \cW \to \cV \otimes_\iota \cW$
  from the cartesian product of  $\cV$ and $\cW$ to the inductive tensor product of  $\cV$ and $\cW$
  is not only separately but even jointly continuous. 
  
  Because  $\cV$ and $\cW$ are both colimits
  of countable strict inductive systems of finite dimensional vector spaces  $(\cV_i)_{i\in I}$ and
  $(\cW_j)_{j\in J}$, respectively, one can apply \cite[Thm.\ 4.1]{HiraiEtAl} to conclude that
  the product topology on  $\cV \times \cW $ coincides with the colimit topology of the
  strict inductive system $(\cV_i\times \cW_j)_{(i,j)\in I\times J}$. Moreover, by local compactness of the
  $\cV_i$ and $\cW_j$, the colimit topologies on $\cV$, $\cW$ and  $\cV \times \cW $ are all locally convex. 
  Hence  $\cV \times \cW \to \cV \otimes_\iota \cW$ is jointly continuous since the restrictions to the finite
  dimensional vector spaces $\cV_i\times \cW_j$ are. 
\end{rem}

\begin{rem}
  The compact-open topology on $\I(\cV,\cW) $ coincides with the topology of uniform convergence on compacta
  since $\cW$ inherits a uniform structure from its locally convex topology.
  To check that a sequence in $\I(\cV,\cW) $ converges, it suffices to
  verify that it converges uniformly on the closed unit  balls
  $\closure{B}_i \subset \cH_i$, where $(\cH_i)_{i\in \N}$
  is a strict inductive system of finite dimensional subspaces $\cH_i\subset \cV$ whose union is $\cV$.
  Moreover, if $\cW =\colim_j \hilbtwo_j$, where
  $(\hilbtwo_j)_{j\in \N}$ is a sequence of finite dimensional subspaces whose union is $\cW$,
  then
  \begin{equation}
    \label{eq:colimcommutesembeddings}
    \colim_j \I(\cH_i, \cK_j)  \xrightarrow{\cong}\I(\cH_i, \colim_j \cK_j) =   \I(\cH_i,\cW) \ .
  \end{equation}
  This follows from the observation that $\I(\cH_i,\cW)$ is a subspace of the space $\cL(\cH_i,\cW)$ of linear maps
  $\cH_i\to\cW$ endowed with the compact-open topology and the fact that
  the inductive tensor product commutes with colimits:
  \[
    \colim_j \cL (\cH_i,\cK_j) \cong \colim_j \cH_i' \otimes \cK_j \cong \cH_i' \otimes \colim_j\cK_j \cong \cL (\cH_i,\cW)\ , \]
  where $\cH_i'$ is the dual vector space of $\cH_i$.
\end{rem}
 
\begin{rem}
The category  $\fHilbo$ embeds in $\sJ$ as the full-subcategory whose objects are finite dimensional. Any object $\cV$ of $\sJ$ is a colimit
\[ \cV = \colim_{i} \hilbone_i\]
for $ (\hilbone_i)_{i\in I}$ a countable family of finite dimensional subspaces whose union is  $\cV$ . 
\end{rem}

Before extending quantum state types to functors defined on the category $\sJ$, we state
the following auxiliary lemma. These kinds of result are typically left implicit or stated without proof by homotopy theorists working exclusively in $\CGTop$, but we have decided to break with tradition and include a proof for clarity, especially due to the subtle nature of pure state spaces in the weak$^*$ topology.

\begin{lem}\label{lem:CGTOPFACTS}
Let $X$ be a compact Hausdorff space and
let
\[\xymatrix{Y_0 \ar[r]^-{f_0} & Y_1  \ar[r]^-{f_1}   &  Y_2  \ar[r]^-{f_2}  & \cdots  }\]
be a diagram in the category of topological spaces, where the $f_i$ are closed embeddings and the  $Y_i$ are weak Hausdorff spaces. Let 
\[\jmath_i \colon Y_i \to \colim_{i,f_i} Y_i =:Y\]
be the canonical maps to the colimit $Y$ of the diagram in topological spaces. Suppose that the $Y_i$ are weak Hausdorff. Then
\begin{enumerate}[(a)]
\item\label{firstclaim} the colimit $Y$ is weak Hausdorff.
\item\label{secondclaim} if
$ g\colon X \to  Y$
is continuous, then there exists $i$ such that $g(X) \subset \jmath_i(Y_{i})$.
\item\label{thirdclaim} if the $Y_i$ are compactly generated, then so is $Y$.
\item\label{fourthclaim} if the $Y_i$ are compactly generated, then
  the canonical map
  \[  \Phi \colon \colim_{i,\Map(X,f_i)} \Map(X, Y_i) \to \Map(X, \colim_{i,f_i} Y_i) \]
  is a homeomorphism.
\end{enumerate}
\end{lem}
\begin{proof}
Let $\jmath_i(Y_i) \subset Y$. Then 
\[ Y \cong \bigcup_{i\geq 0} \jmath_i(Y_i) \]
with the union topology, defined so that a subset $A\subset Y$ is closed if and only if $\jmath_i(Y_i)  \cap A$ is closed in $\jmath_i(Y_i) $ for all $i\geq 0$. Furthermore, $\jmath_i \colon Y_i \to \jmath_i(Y_i) $ is a homeomorphism onto its image. So, we suppress the $\jmath_i$ and simply prove the claims for 
\[ Y \cong \bigcup_{i\geq 0} Y_i \]
with the union topology, under the assumption that 
\[Y_0 \subset Y_1 \subset Y_2 \subset \cdots\]
are the inclusions of closed subspaces.

We first show \eqref{firstclaim}, that if the $Y_i$ are weak Hausdorff, then $Y$ is weak Hausdorff. Let $X$ be a compact Hausdorff space and consider $g(X)$. We want to show that this is closed. This holds if and only if $g(X) \cap Y_i$ is closed for every $i$. But $Y_i$ is closed in $Y$ since the union is along closed embeddings, hence $g^{-1}(Y_i)$ is closed in $X$. But since $X$ is compact, $g^{-1}(Y_i)$ is compact. So,
\[g( g^{-1}(Y_i) ) =g(X) \cap Y_i\]
is closed since $Y_i$ is weak Hausdorff. Hence, $g(X)$ is closed since its intersection with each $Y_i$ is closed.

Now, we turn to \eqref{secondclaim}.  Since $X$ is compact Hausdorff, by the definition of a weak Hausdorff space, $g(X)$ is closed in $Y$ since it is the continuous image of a compact Hausdorff space.
Suppose there is no $i$ such that $g(X)\subset Y_i$. Then for every $i \geq0$, there exists $x_i \in X$ such that $g(x_i) \not\in Y_i$. 
Consider the sets
\[K_m = \{g(x_m), g(x_{m+1}), \ldots\}.\]
Note that 
\[K_0 \supset K_1 \supset K_2 \supset K_3 \supset \cdots\]
so this is a nested family. Furthermore, for any $i$
\[K_m \cap Y_i  = g(X_m)\]
for a finite subset $X_m\subset X$.
Since $X$ is Hausdorff, $X_m$ is compact, so it's image under $g$ must be closed in the weak Hausdorff space $Y_i$. Therefore, $K_m$ is closed for all $m\geq 0$. We thus have a nested sequence 
\[g(X)\supset K_0 \supset K_1 \supset K_2 \supset K_3 \supset \cdots\]
of closed subsets of a compact set.
The sequence has the finite intersection property, but the intersection $\bigcap_{m\geq 0}K_m$ is empty. This is a contradiction.  So, $g(X)\subset Y_i$ for some $i$.

Next we show \eqref{thirdclaim}. Let $A \subset Y$ be a compactly closed subset of $Y$. That is, $A$ has the property that $g^{-1}(A)$ is closed in $X$ for any continuous map $g \colon X \to Y$ for $X$ a compact Hausdorff space. We need to show that $A \cap Y_i$ is closed for any $i$. 
Let $ g\colon X \to Y_i$ be any continuous map from a compact Hausdorff space $X$. Then we can compose $g$ with the inclusion to get a continuous map $g \colon X \to Y$.
But,
\[g^{-1}(A \cap Y_i) = g^{-1}(A) \cap g^{-1}(Y_i).\]
Both these sets are closed, hence $g^{-1}(A \cap Y_i) $ is closed in $X$. Therefore, $A \cap Y_i$ is compactly closed. Since $Y_i$ is compactly generated, $A \cap Y_i$ is closed. Hence, $A$ is closed.

It remains to prove \eqref{fourthclaim}. Since the $f_i$ are closed embeddings, so are the maps $ {f_i}_* =\Map(X,f_i)$. It follows that
  $ \colim_i \Map(X, Y_i) $ is in $\CGTop$. The canonical map $\Phi$ is induced by the map of diagrams
\[\xymatrix{
\cdots \ar[r] &  \Map(X, Y_i) \ar[d]^-{\Map(X,\iota_i)} \ar[r] &  \Map(X, Y_{i+1}) \ar[r] \ar[d]^-{\Map(X,\iota_{i+1})} & \cdots  \\
\cdots \ar[r] & \Map(X,\colim Y_i) \ar[r]^-{=} &   \Map(X,\colim Y_i) \ar[r]^-{=} &  \cdots 
}\]
and so is continuous. We show that it is also a closed map. Note that a subset $A\subset Z$ of a compactly generated space $Z$ is closed if and only if $G^{-1}(A)$ is closed for every continuous map $G\colon K \to Z$ where $K$ is compact Hausdorff. Let $A$ be closed in $ \colim_{i,\Map(X,f_i)} \Map(X, Y_i) $ and consider
$\Phi (A) \subset \Map(X, \colim_{i,f_i} Y_i)$. Let $G \colon K \to  \Map(X, \colim_{i,f_i} Y_i)$ be a map from a
compact Hausdorff space $K$. Since
\[ \Map(K, \Map(X, \colim_{i,f_i} Y_i) )\cong \Map(X\times K, \colim_{i,f_i} Y_i) \]
and $X\times K$ is compact, the adjoint $g\colon X\times K\to \colim_{i,f_i} Y_i$ of $G$ has image in $Y_{i_0}$
for some $i_0$ by \eqref{secondclaim}. That means  
\[g(x,k) = G(k)(x) \in Y_{i_0} \quad \text{for all } (x,k) \in X \times K \ .\]
Therefore, we get a factorization
\[ \xymatrix{
K \ar[r]^-{G} \ar[d]_{G_0} & \Map(X, \colim_{i,f_i} Y_i) \\
\Map(X, Y_{i_0}) \ar[ur]_-{\Map(X,\iota_{i_0})} & 
}\]
Hence, 
\[G^{-1}(\Phi(A)) =G^{-1}_0 \Map(X,\iota_{i_0})^{-1}(A) \] 
Since $A$ is closed in the colimit if and only if each $\Map(X,\iota_{i})^{-1}(A) \subset \Map(X, Y_{i})$ is closed, then $\Map(X,\iota_{i_0})^{-1}(A)) \subset \Map(X, Y_{i_0})$ is closed. Since  $ \Map(X, Y_{i_0})$
 is in $\CGTop$, this implies that $G^{-1}_0(\Map(X,\iota_{i_0})^{-1}(A))$ is closed in $K$. The claim follows.
\end{proof}

We extend the quantum state type to $\sJ$ as follows. 
First, we extend the functor $\QS_\Gamma$ to $\sJ$ by defining it on an infinite dimensional 
 $\cV \in \sJ $ as
\[\QS_\Gamma(\cV) := \colim_i \QS_\Gamma( \hilbone_i) \ ,\]
where the colimit is taken over all finite dimensional subspaces
$\cH_i\subset \cV$. 
Denote by
\[ \jmath_i^\cV \colon \QS_\Gamma( \hilbone_i) \to \QS_\Gamma(\cV) \]
the inclusions in the colimit.
For $f \in \I(\cV,\cW)$ with $\cV = \colim_i  \hilbone_i$, $\cW=\colim_j  \hilbtwo_j$ and  $\hilbone_i \subset \cV$, $\hilbtwo_j \subset \cW$ finite dimensional, we can choose for every
$i$ an index $j(i)$ so that $f (\hilbone_i)\subset \hilbtwo_{j(i)}$. Then $f$ restricts to a linear isometry 
\[
  f_i \colon \hilbone_i \to \hilbtwo_{j(i)} \ .
\]
We now let $\QS_\Gamma(f) \colon \QS_\Gamma (\cV) \to \QS_\Gamma (\cW)$
be the unique continuous map such that
\[
  \QS_\Gamma(f) \circ \jmath^\cV_i = \jmath^\cW_{j(i)} \circ \QS_\Gamma(f_i)
  \quad \text{for all } i\ .
\]  
Next, we extend
\[\eta^{\QS_\Gamma} \colon \QS_\Gamma(\cV) \times \QS_\Gamma(\cW) \to \QS_\Gamma(\cV\otimes \cW) \] 
as follows. Note that
\[\cV\otimes \cW \cong \colim_{i,j} \hilbone_i\otimes  \hilbtwo_j.\]
Therefore, up to a canonical homeomorphism,
\[\colim_{i,j}\QS_\Gamma( \hilbone_i \otimes  \hilbtwo_j) \xrightarrow{\cong} \QS_\Gamma( \cV\otimes \cW)  \]
Since $\eta^{\QS_\Gamma}$ is a natural transformation, we also have a map
\[\xymatrix{ \colim_{i,j} \QS_\Gamma( \hilbone_i) \times  \QS_\Gamma(\hilbtwo_j) \ar[rr]^-{\colim_{i,j}\eta^{\QS_\Gamma}} & &  \colim_{i,j} \QS_\Gamma( \hilbone_i \otimes \hilbtwo_j) }\]
Since the natural map 
\[ \colim_{i,j} \QS_\Gamma( \hilbone_i) \times  \QS_\Gamma(\hilbtwo_j)  \xrightarrow{ \ \cong \ } \colim_{i} \QS_\Gamma( \hilbone_i) \times  \colim_{j}   \QS_\Gamma(\hilbtwo_j) \]
is a homeomorphism, composing its inverse with the previous maps gives the extension of $\eta^{\QS_\Gamma}$.

\begin{rem}
Here, it is important that we are working in compactly generated topological spaces. For any compactly generated $X$, the functor $-\times X$ is left adjoint to $\Map(X,-)$ and so commutes with colimits. These unions are directed colimits along maps with closed images. In the category of compactly generated spaces, such a colimit is again a compactly generated space as we saw in \cref{lem:CGTOPFACTS}. We can then use use 
\begin{align*}\QS_\Gamma(\cV) \times \QS_\Gamma(\cW) &\cong \colim_i( \QS_\Gamma(\hilbone_i) \times \QS_\Gamma(\cW)) \\
&\cong \colim_i(\colim_j(\QS_\Gamma( \hilbone_i)\times \QS_\Gamma(\hilbtwo_j)).
\end{align*}
\end{rem}

We claim that the extended functor $\QS_\Gamma \colon \sJ \to \CGTop$ we just constructed
is topologically enriched. To this end, we need to show that for all objects $\cV, \cW$ in $\sJ$ the map
  \begin{equation}
    \label{eq:QSfunctopenriched}
    \Phi_{\cV,\cW} \colon \I (\cV,\cW) \times \QS_\Gamma (\cV) \to \QS_\Gamma (\cW) ,
    \quad (f,\omega)  \mapsto \QS_\Gamma (f) (\omega)
  \end{equation}
  is continuous. We show this by a number of reductions.

\begin{lem}
Write $\cV = \colim_i \cH_i$ with $\cH_i$ finite dimensional. 
The diagram 
\[\xymatrix@C=4pc{ 
\I(\cV,\cW) \times \QS_\Gamma(\cH_i)  \ar[d]_{\id\times\QS_\Gamma(\jmath^{\cV}_i)}   \ar[r]^-{( - \circ \jmath_i^{\cV}) \times \id} & \I(\cH_i,\cW) \times\QS_\Gamma(\cH_i)  \ar[d]^-{\Phi_{\cH_i,\cW} } \\
 \I(\cV,\cW) \times\QS_\Gamma(\cV) \ar[r]_-{\Phi_{\cV,\cW}} &\QS_\Gamma(\cW)
}\]
then commutes for every $i$.
Furthermore, $\Phi_{\cV,\cW}$ is continuous if $\Phi_{\cH_i,\cW}$ is so for all $i$.
\end{lem}
\begin{proof}
  That the diagram commutes is obvious.  First we show that  $\Phi_{\cV,\cW}$ is continuous if and only if
  $\Phi_{\cV,\cW} \circ \big(\id\times\QS_\Gamma(\jmath^\cV_i)\big)$ is continuous for all $i$.
  If $\Phi_{\cV,\cW}$ is continuous, then since $\QS_\Gamma(\jmath^\cV_i)$ is continuous, so is
  $\Phi_{\cV,\cW} \circ \big( \id\times\QS_\Gamma(\jmath^\cV_i)\big)$. Furthermore,
\begin{align*}
\colim_i \left( \I(\cV,\cW) \times \QS_\Gamma(\cH_i)\right) &\xrightarrow{\cong}  \I(\cV,\cW) \times \colim_i\QS_\Gamma(\cH_i) \\
&=  \I(\cV,\cW) \times\QS_\Gamma(\cV)   .
\end{align*}
But a map out of $\colim_i \left( \I(\cV,\cW) \times \QS_\Gamma(\cH_i)\right) $ is continuous if and only if it is continuous when restricted to each of the $\I(\cV,\cW) \times\QS_\Gamma(\cH_i) $.
Therefore, $\Phi_{\cV,\cW}$ is continuous if and only if
$\Phi_{\cV,\cW} \circ \big( \id\times\QS_\Gamma(\jmath^\cV_i)\big)$ is continuous.
But 
\[ ( - \circ \jmath_i^\cV)  \colon \I(\cV,\cW)  \to \I(\cH_i,\cW)   \]
is continuous, so provided that $\Phi_{\cH_i,\cW}$  is continuous, then so is
\[ \Phi_{\cH_i,\cW}\circ \big( (- \circ \jmath_i^\cV ) \times \id\big) =
  \Phi_{\cV,\cW} \circ \big( \id\times\QS_\Gamma(\jmath^\cV_i)\big) \ .  \]
This proves the claims.
\end{proof}

So, we have reduced the problem to proving that 
\[\Phi_{\cH_i,\cW} \colon \I(\cH_i,\cW) \times \QS_\Gamma(\cH_i)  \to \QS_\Gamma(\cW) \]
is continuous.

  \begin{lem}
  Write $\cW = \colim_j \cK_j$ with $\cK_j$ finite dimensional. 
  Then for every finite dimensional $\cH$ the map $\Phi_{\cH,\cW} $ is continuous provided that,
  for each $j$, $\Phi_{\cH,\cK_j}$ is continuous. 
\end{lem}
\begin{proof}
  By \eqref{eq:colimcommutesembeddings}, we have a natural isomorphism 
  \[ \colim_j \I(\cH, \cK_j)  \xrightarrow{\cong}\I(\cH, \colim_j \cK_j) = \I(\cH ,\cW).\]
  Therefore, we have isomorphisms
  \[\colim_j (\I(\cH , \cK_j)  \times \QS_\Gamma (\cH) ) \xrightarrow{\cong}
    (\colim_j \I(\cH, \cK_j) ) \times \QS_\Gamma (\cH) \xrightarrow{\cong} \I(\cH, \cW) \times  \QS_\Gamma (\cH)\ . \]
  But a map out of the colimit is continuous if and only if each of its restrictions is continuous,
  so we need to know that the maps
  \[\I(\cH, \cK_j)  \times \QS_\Gamma(\cH) \to \QS_\Gamma (\cW)\]
  given by
  \[ (f, \omega) \mapsto \QS_\Gamma (\jmath^\cW_j \circ f)(\omega) =
     \QS_\Gamma (\jmath^\cW_j)\circ \QS_\Gamma (f) (\omega) =\QS_\Gamma(\jmath^\cW_j) \circ \Phi_{\cH,\cK_j} (f,\omega)    \]
   are continuous. Since the maps $\QS_\Gamma(\jmath^\cW_j)$ and $\Phi_{\cH,\cK_j}$
   are continuous, the claim follows.
\end{proof}

Altogether, we have thus proved the following main result of this section. 

\begin{prop}\label{prop:extend}
  Given any quantum state type $\QS_\Gamma$, the above construction gives a
  topologically enriched extension
   \[\QS_\Gamma \colon \sJ \to \CGTop\] 
  and an extended natural transformation $\eta^{\QS_\Gamma}$ 
\[\eta^{\QS_\Gamma} \colon \times \circ \QS_\Gamma   \to \QS_\Gamma\circ \otimes \]
between functors
\[ \times \circ \QS_\Gamma, \QS_\Gamma\circ \otimes  \colon (\sJ)^2 \to \CGTop.\]
\end{prop}


\section{Formal properties of quantum state types}\label{sec:homotopical}

In this section, we deduce various formal properties of a quantum state type $(\QS_\Gamma, \etagQ)$. First, we define our space of ``probes'' as follows:
\begin{defn}\label{defn:parameterspace}
A \emph{parameter space} is a compact Hausdorff space  $X$ which is homotopy equivalent to a CW complex. A pair consisting of a parameter space $X$ and a closed embedding $x_0 \colon \pt \to X$ whose image is a neighborhood deformation retract is called a \emph{based parameter space}. 
\end{defn}
Our goal will be to relate
\[[X,\QS_\Gamma(\cV)]\]
for certain infinite dimensional objects $\cV \in \sJ$ with phases of quantum systems of type $\QS_\Gamma$ parametrized by $X$.

\subsection{Cellularizing quantum state types}\label{sec:CWreplacements}
We want to study quantum state types from a homotopical perspectives. To this end, it useful to put ourselves in a slightly more friendly setting for homotopy theory by performing a CW-replacement. The reader who is happy to just accept that we can implicitly replace a quantum state type so that $\QS_\Gamma(\cH)$ are CW-complexes and $\QS_\Gamma(f)$ are inclusions of subcomplexes can skip this section.

\bigskip

Let $\CWTop$ be the full subcategory of $\CGTop$ whose objects are compactly generated spaces that are homotopy equivalent to CW-complexes. 
We consider the functor
\[\rep\colon \CGTop \to \CWTop \subset \CGTop \]
obtained as the composite
\[\rep(X) = |\Sing_\bullet(X)|\]
where 
\[\xymatrix{   |-| : \sSets   \ar@<1ex>[r] & \CGTop:  \Sing_\bullet  \ar@<1ex>[l] }\]
for $\sSets$ the category of simplicial sets, $\Sing_\bullet$ the singular simplicial complex functor, and $|-|$ the geometric realization functor. See, for example, \cite{Maysset} for a thorough treatment on simplicial sets. We have also provided some minimal background in \cref{sec:simplicial}.

The pair of functors $(|-|, \Sing_\bullet)$ form an adjunction with $|-|$ the left adjoint, so that
\[ \Hom_{\CGTop}( |X_\bullet|, Y) \cong \Hom_{\sSets}(X_\bullet, \Sing_\bullet(Y)) .\]
Furthermore, the co-unit of the adjunction
\[\varepsilon_X \colon \rep X \to X\]
is a \emph{weak equivalence} in the following sense.

\begin{defn}\label{defn:weakequivalence}
A map $f\colon Y \to Z$ is called a \emph{weak equivalence} if it induces a bijection on path components, and an isomorphism on all homotopy groups with respect to any choice of base point in $Y$. We say that two spaces are weakly equivalent if there is a zig-zag of weak equivalences connecting them. 
\end{defn}

\begin{rem}\label{rem:facts}
By the Whitehead Theorem, if $X$ has the homotopy type of a CW-complex, and $Y \to Z$ is a weak equivalence, then $f$ induces an isomorphism
\[[X, f] \colon [X,Y ] \xrightarrow{ \ \cong \ } [X,Z].\]
\end{rem}

The functor $\Sing_\bullet$ preserves all limits since it is a right adjoint and $|-|$ preserves finite limits (this is a non-formal property of that functor).  Therefore, for any $Y$ and $Z$, there is a natural isomorphism
\[\gamma_{X,Y} \colon \rep (Y\times Z)   \xrightarrow{\cong}   \rep Y \times \rep Z .\]

On the other hand, the functor $|-|$ preserves all colimits (since it is a left adjoint).
 The functor $\Sing_\bullet$ is defined using $\Map(\Delta_n, -)$ where $\Delta_n$ is the standard $n$-simplex. Since this is a compact Hausdorff and connected space, it follows that $\Sing_\bullet$ preserves coproducts, as well as colimits over diagrams in $\CGTop$ which have a countable cofinal subset, and whose maps are closed embeddings. See \cref{lem:CGTOPFACTS}. Again, these are not formal properties, but rather follow from the definition of $\Sing_\bullet$. 
 
In particular, for any collection of compactly generated spaces $Y_i$, there is a natural isomorphism
\[ \coprod_{i} \rep Y_i   \xrightarrow{\cong}  \rep\left(\coprod_{i} Y_i \right) ,\]
and for any diagram $\colim_{i,f_i} Y_i$  in $\CGTop$  indexed on the non-negative integers for which the maps $f_i$ are closed embeddings, a natural isomorphism
\[ \colim_{i,\rep f_i}  \rep Y_i  \xrightarrow{\cong}   \rep \left(\colim_{i,f_i} Y_i\right) . \]

Therefore, we have the following comparison result.
\begin{lem}
If $X$ is a parameter space, and $\colim_{i,f_i} Y_i$ in $\CGTop$ is a colimit over a diagram indexed on the non-negative integers for which the maps $f_i$ are closed embeddings,
 then there are isomorphisms
\begin{align*} 
\colim_{i,\rep f_i} [X,    \rep Y_i]  &\xrightarrow{\cong}  [X,   \colim_{i, \rep f_i}  \rep Y_i]  \\
&\xrightarrow{\cong}  [X, \rep \colim_{i, f_i} Y_i]  \\
& \xrightarrow{\cong} [X, \colim_{i,f_i} Y_i]  \xleftarrow{\cong} \colim_{i,f_i} [X,  Y_i] 
\end{align*}
which are natural in $X$. 
\end{lem}

We can use the functor $\rep$ to replace any quantum state type $\QS_\Gamma$ with a CW-valued topologically enriched functor
\[\rep \circ \QS_\Gamma \colon \sJ \to \CWTop\]
and natural transformations
\[\rep(\eta^{\QS_\Gamma}_{\cH, \cK}) \circ \gamma_{ \QS_\Gamma(\cH),  \QS_\Gamma(\cK)}^{-1}  \colon  \rep\QS_\Gamma(\cH) \times \rep \QS_\Gamma(\cK) \to \rep\QS_\Gamma(\cH\otimes \cK)  .\]
The replacement comes equipped with weak equivalences
\[\rep\QS_\Gamma(\cV) \xrightarrow{\ \simeq \ } \QS_\Gamma(\cV)\]
for all $\cV \in \sJ$, and satisfies the property that, if 
\[\cV = \colim_{i,f_i} \cH_i \in \sJ\]
where the $\cH_i$ are finite dimensional and the $f_i$ are linear isometries, then
\begin{align*} 
\colim_{i,\rep\QS_\Gamma( f_i)} [X,    \rep \QS_\Gamma(\cH_i)]  &\xrightarrow{\cong}  [X,  \colim_{i,\rep\QS_\Gamma( f_i)}  \rep \QS_\Gamma(\cH_i)]  \\
&\xrightarrow{\cong}  [X, \rep \QS_\Gamma(\cV)]  \\
& \xrightarrow{\cong} [X, \QS_\Gamma(\cV)]  \xleftarrow{\cong} \colim_{i,\QS_\Gamma(f_i)} [X,  \QS_\Gamma(\cH_i)] .
\end{align*}

\begin{rem}
If $\QS_\Gamma$ is such that $\QS_\Gamma(\cH)$ is a CW-complex for all $\cH\in \fHilb$ and $\QS_\Gamma(f) \colon \QS_\Gamma(\cH) \to \QS_\Gamma(\cK)$ is the inclusion of a sub-complex for all morphism $f \colon \cH \to \cK$ in $\fHilb$, the replacement is not necessary and in fact, in that case,
\[\varepsilon_{ \QS_\Gamma(\cH)} \colon \rep \QS_\Gamma(\cH) \to  \QS_\Gamma(\cH) \]
is a homotopy equivalence by the Whitehead Theorem. 
We do not know if this holds for the universal quantum state type $\sP_\Gamma$ in the case when $\Gamma$ is infinite.
\end{rem}

\begin{convention}
For the remainder of \cref{sec:homotopical}, we implicitly replace all quantum state types with $\rep \QS_\Gamma$, but we suppress $\rep$ from the notation.
\end{convention}

\subsection{The classifying spaces for quantum systems and phases}In this section, we introduce our classifying spaces.

We choose an object $\cH$ of $\fHilbo$ of dimension at least two and a pure state $\psi$ of $\cH$. 
We (temporarily) choose a unit vector $u \in \psi$. We let 
\[\ff_i \colon \cH^{\otimes i} \to \cH^{\otimes (i+1)}\] 
be defined by
\[\ff_i(x) =x \otimes u   .\]
This is a linear isometric embedding, and so we get maps
\[\QS_{\Gamma}(\ff_i)\colon  \QS_{\Gamma} (\cH^{\otimes i}) \to \QS_{\Gamma} (\cH^{\otimes (i+1)}).\]
 If we also denote by $\psi$ the factorized state uniquely determined by 
\[ \psi(A) = \prod_{v\in \Lambda} \langle \psi, A_v\psi\rangle \quad \text{for} \quad A=\bigotimes_{v\in \Lambda}A_v \in \fA_\Lambda(\cH)\]
then 
\[\QS_{\Gamma}(\ff_i)(\omega) = \omega\otimes \psi\]
and so does not depend on the choice of vector $u$ we used to define $\ff_i$.

\begin{defn}\label{defn:classifyingspaces}
For $\cH^{\otimes i}$, $\ff_i$ as above and $\QS_\Gamma$ a quantum state type,
 let
\[\QS_{\Gamma}^{\sqcup}(\cH):= \coprod_{i\geq 0} \QS_\Gamma(\cH^{\otimes i})\]
and 
\[\mu_i \colon \QS_\Gamma(\cH^{\otimes i})  \to \QS_{\Gamma}^{\sqcup}(\cH)\] be the inclusions.
We base the space $\QS_{\Gamma}^{\sqcup}(\cH)$ at $\mu_0$.\footnote{Since we implicitly replaced $\QS_\Gamma$ by $\rep \QS_\Gamma$, this base point   is the inclusion of a zero-cell and so is non-degenerate. This holds for all our choices of based points in this section.}
Let
\begin{align*}  \QS_\Gamma(\cH^{\otimes \infty}) &:=  \QS_\Gamma(\colim_{i,\ff_i} \cH^{\otimes i})  \cong \colim_{i,\QS_\Gamma( \ff_i)}   \QS_\Gamma(\cH^{\otimes i}) 
\end{align*}
and
\[\jmath_i \colon \QS_\Gamma(\cH^{\otimes i})  \to \QS_\Gamma(\cH^{\otimes \infty})  \]
be the inclusion of $\QS_\Gamma(\cH^{\otimes i})  $ in the colimit. We base $\QS_{\Gamma}(\cH^{\otimes \infty})$ at $\jmath_0$.
Finally,  let
\[\jmath = \coprod \jmath_i \colon \QS_\Gamma^{\sqcup}(\cH)  \to \QS_\Gamma(\cH^{\otimes \infty}).\]
\end{defn}

\begin{rem}
In the colimit defining $  \QS_\Gamma(\cH^{\otimes \infty})$, the base point is identified with each
\[\psi_i := \bigotimes_{v\in \Gamma}\psi^{\otimes i}  \in  \QS_\Gamma(\cH^{\otimes i}).\]
The spaces $\QS_\Gamma^{\sqcup}(\cH) $ and $\QS_\Gamma(\cH^{\otimes \infty}) $ depend on the choice of vacuum $\psi$ but we have decided not to include that dependence in our notation.
\end{rem}

Both spaces $\QS_\Gamma^\sqcup(\cH) $ and $\QS_\Gamma(\cH^{\otimes \infty}) $  are going to have additional structure. That is, we will show that they are homotopy commutative and associative $H$-spaces. That is, they have a product which is unital, associative and commutative up to homotopy.
Furthermore, the map $\jmath$ will be a map of $H$-spaces in the sense that it will preserve the product and units up to homotopy. Background on these concepts is given in \cref{sec:topmonhspace}. 

Before jumping into the work we need to show these claims, we make the definition that motivates our study of these spaces.

\begin{defn}\label{defn:classifying}
For a parameter space $X$, we define:
\begin{enumerate}[(a)]
\item A \emph{parametrized system} of quantum state type $\QS_\Gamma$ modeled on $(\cH, \psi)$ is a continuous function
\[ \omega \colon X \to \QS_\Gamma^{\sqcup}(\cH).\]
The space of parametrized systems is $\Map(X, \QS_\Gamma^{\sqcup}(\cH))$ equipped with the $k$-ification of the compact-open topology.
\item The parametrized system is called \emph{trivial} if it is constant with value the completely factorized state $\psi_i$ for some $i\geq 0$.  We abuse notation and call this constant function $\psi_i$.
\item A \emph{quantum phase parametrized by $X$} of state type $\QS_\Gamma$ modeled on $(\cH, \psi)$ is the homotopy class of a map
\[ \varphi \colon X \to \QS_\Gamma(\cH^{\otimes \infty}).\]
The set of parametrized quantum phases is $[X,  \QS_\Gamma(\cH^{\otimes \infty})]$, the set of homotopy classes maps.
\end{enumerate}
\end{defn}

\begin{rem}
In \cref{defn:classifying}, $X$ is not assumed to be based and the maps and homotopies are not based.
\end{rem}

\subsection{Stacking structure}

The goal for this section is to define the $H$-space structures for  $ \QS_\Gamma^{\sqcup}(\cH)$ and $\QS_\Gamma(\cH^{\otimes \infty})$ coming from stacking (or tensor product in this case) and to discuss the relationship between the two spaces.   Background on topological monoids and $H$-spaces is given in \cref{sec:topmonhspace}.

We start with  $ \QS_\Gamma^{\sqcup}(\cH)$.
\begin{prop}\label{prop:QStopmonoid}
The natural transformation $\etagQ$ endows $\QS_{\Gamma}^{\sqcup}(\cH)$ with the structure of a strictly associative topological monoid through the pairing
\[ \QS_{\Gamma}^{\sqcup}(\cH)\times \QS_{\Gamma}^{\sqcup}(\cH) \cong  \coprod_{i,j\geq 0} \QS_\Gamma(\cH^{\otimes i})\times \QS_\Gamma(\cH^{\otimes j}) \to \QS_{\Gamma}^{\sqcup}(\cH)\]
 given by
\[ (\omega_i,\omega_j)\mapsto \etagQ(\omega_i,\omega_j) \in  \QS_{\Gamma}(\cH^{\otimes (i+j)}).\]
The unit is the state
\[\psi_0 \in \QS_\Gamma(\C).\] 
The topological monoid $\QS_{\Gamma}^{\sqcup}(\cH)$ is homotopy commutative, but not strictly commutative. 
\end{prop}
\begin{proof}
Everything except the last statement is straightforward. For homotopy commutativity, we note that for any $\cK \in \fHilbo$, the unitary group on $\cK$ is a path connected space. So, we can choose a path 
\[ \alpha_{i,j} \colon \interval \to \I(\cH^{\otimes (i+j)}, \cH^{\otimes (i+j)}) \] 
from the identity to the map $\sigma_{i,j}$ which switches the factors in $\cH^{\otimes (i+j)} \cong \cH^{\otimes i}\otimes \cH^{\otimes j}$. Then 
\begin{align*} 
 h_{i,j}\colon \interval \times \QS_\Gamma(\cH^{\otimes i}) \times   \QS_\Gamma(\cH^{\otimes j}) &\to \QS_\Gamma(\cH^{\otimes (i+j)})   \\
 (t,x,y)& \mapsto  \QS_\Gamma(\alpha_{i,j}(t))(x\otimes y).      
 \end{align*}
gives a path from the two compositions in the diagram
\[\xymatrix{ \QS_\Gamma(\cH^{\otimes i}) \times   \QS_\Gamma(\cH^{\otimes j}) \ar[rr]^-\sigma \ar[dr]_-\otimes& & \QS_\Gamma(\cH^{\otimes j}) \times   \QS_\Gamma(\cH^{\otimes i}) \ar[dl]^-\otimes \\
&  \QS_\Gamma(\cH^{\otimes (i+j)}). &
}\]
Defining the top horizontal arrow in
\[\xymatrix{
\interval \times \QS_\Gamma^{\sqcup}(\cH) \times \QS_\Gamma^{\sqcup}(\cH)
\ar[d]_-{\cong} \ar[rr]&    &\QS_\Gamma^{\sqcup}(\cH) \\
 \coprod_{i,j} \interval \times \QS_\Gamma(\cH^{\otimes i}) \times   \QS_\Gamma(\cH^{\otimes j})\ar[rr]^-{\coprod h_{i,j}} & & \coprod_{k} \QS_\Gamma(\cH^{\otimes k})\ar[u]_-{\cong} }\]
 so that the diagram commutes gives the required homotopy. 
\end{proof}

\bigskip

To discuss the higher structure of $\QS_\Gamma(\cH^{\otimes \infty})$, we will need to introduce more theory.  
The reader who is not familiar with the notion of  $\sE_\infty$-spaces and of operads may be satisfied with knowing that the result implies that $\QS_\Gamma(\cH^{\otimes \infty})$ has a multiplication that is unital, associative and commutative up to homotopy. That is, $\QS_\Gamma(\cH^{\otimes \infty})$ is a homotopy commutative and associative $H$-space.  In particular, this implies that $\pi_0\QS_\Gamma(\cH^{\otimes \infty})$ is a monoid. Furthermore, if $\pi_0\QS_\Gamma(\cH^{\otimes \infty})$ happens to also be a group, then $\QS_\Gamma(\cH^{\otimes \infty})$ has inverses up to homotopy and so is an $H$-group. Such quantum state types are called \emph{invertible}.  With this in mind, such a reader can skip ahead to \cref{sec:groupcompletion}. Alternatively, we have provided some background on $\sE_\infty$-spaces in \cref{sec:opandeinfty}.

In the remainder of this section, we will show the following result.
\begin{thm}\label{thm:jmathHspace}
The space  $ \QS_\Gamma(\cH^{\otimes \infty})$ is an $\sE_\infty$-space. As such, it is a homotopy commutative and associative $H$-space. If $\QS_\Gamma$ is invertible, then $ \QS_\Gamma(\cH^{\otimes \infty})$ is an $H$-group.

For any quantum state type modeled on $(\cH,\psi)$, the map $\jmath \colon \QS_\Gamma^{\sqcup}(\cH) \to \QS_\Gamma(\cH^{\otimes \infty})$  is a morphism of $H$-spaces.
\end{thm}

\bigskip

The key for proving \cref{thm:jmathHspace} will be to show the following much stronger result.
\begin{theorem}\label{thm:QSEinfty}
For any infinite dimensional object $\cV$ of $\sJ$ and linear isometry $u \colon \C \to \cV$, the space 
$\QS_\Gamma(\cV)$ based at 
\[\QS_\Gamma(u) \colon \QS_\Gamma(\C)\to \QS_\Gamma(\cV) \]
is an $\sE_\infty$-space. 
In particular, any choice of $\varphi \in \I(\cV^{\otimes 2}, \cV)$ with the property that $\varphi(u\otimes u)=u$ specifies the structure of homotopy commutative and associative $H$-space on   $\QS_\Gamma(\cV)$.

If $\pi_0\QS_\Gamma(\cV)$  is a group, then this $H$-space is an $H$-group (i.e., it has homotopy inverses), and there is an equivalence of $H$-groups
\[\QS_\Gamma(\cV) \simeq  \QS_\Gamma(\cV)_0 \times \pi_0\QS_\Gamma(\cV)  \]
where $ \pi_0\QS_\Gamma(\cV)$ is given the discrete topology and $\QS_\Gamma(\cV)_0$ is the path component of $\psi_0 = \QS_\Gamma(u)(\pt)$.

For any choice of $(u_1 \in \cV_1)$ and $(u_2\in \cV_2)$ where $\cV_1$ and $\cV_2$ are infinite dimensional gives isomorphic $\sE_\infty$-spaces $\QS_\Gamma(\cV_1)$ and $\QS_\Gamma(\cV_2)$.
\end{theorem}

\begin{defn}\label{defn:QSinvert}
A  quantum state type $\QS_\Gamma$ is \emph{invertible} if, given any choice of infinite dimensional object $\cV$ and unit vector  $u\in \cV$, for the $H$-space structure of \cref{thm:QSEinfty}, $\pi_0\QS_\Gamma(\cV)$ is a group.  Equivalently, $\QS_\Gamma$ is  invertible if $\pi_0\QS_\Gamma(\cV)$ is a group \emph{for some} choice of $(\cV,u)$.
\end{defn}

To prove \cref{thm:QSEinfty}, we will introduce an operad which is a multiplicative analogue of the linear isometry operad, so could be called the \emph{multiplicative linear isometry operad}. The necessary background on operads can be found in \cite[\S 1]{MayGILS} although, here, we do not require the zero space of an operad to be a single point. We use the terminology introduced in \cite{2020arXiv200310934M} and call operads whose zero space is a point \emph{reduced}.

\begin{defn}
Let $\cV$ be an infinite dimensional object in $\sJ$. Let
 \[\sK(\cV)(j) = \I(\cV^{\otimes j}, \cV)\] 
 and $\1 \in \sK(\cV)(1)$ be the identity. For $\sum  j_i =j$, let 
\[\gamma \colon \sK(\cV)(k) \times \sK(\cV)(j_1)\times \cdots \times\sK(\cV)(j_k) \to\sK(\cV)(j)     \]
be given by
\[\gamma(g ; f_1,\ldots, f_k) = g\circ (f_1\otimes \cdots \otimes f_k).\]
For $\Sigma_j$ the symmetric group on $j$ letters, define a right action by letting $\sigma \in \Sigma_j$ act on $\sK(V)(j)$ by
\[f\sigma := f\circ \sigma\]
for $f\in \I(\cV^{\otimes j}, \cV)$, 
where $\sigma$ acts on $\cV^{\otimes j}$ 
through the linear action determined by
\[\sigma(u_1\otimes \cdots \otimes u_j) = u_{\sigma^{-1}(1)}\otimes \cdots \otimes u_{\sigma^{-1}(j)}.\]

For any choice of $u\in \cV$, let
\[\bsK(\cV)(j) =  \I_u(\cV^{\otimes j}, \cV)=\{\varphi \in   \I(\cV^{\otimes j}, \cV) : \varphi(u^{\otimes j})=u\},  \]
topologized as a subspace of $ \I(\cV^{\otimes j}, \cV) $, with structure maps the restrictions of the $\gamma$s.
\end{defn}

\begin{prop}\label{lem:sKop}
Let $\cV$ be an infinite dimensional object in $\sJ$. Then $\sK(\cV)$ is an unreduced $\Sigma$-free operad in $\CGTop$ with the property that the spaces $\sK(\cV)(j)$ are contractible for all $j\geq 0$.  Furthermore, $\bsK(\cV)$ is a reduced $\sE_\infty$-operad.
\end{prop}

\begin{proof}
Checking that these are operads is quite formal as the structure maps of $\sK(\cV)$ are closely related to those of the endomorphism operads (\cite[Def. 1.2]{MayGILS}), so we do not give further details on that point.
For the $\Sigma$-free property, note that if $f = f\sigma$ where $f\in \sK(\cV)(j)$ and $\sigma \in \Sigma_j$, then for any simple tensor, 
\[f(u_1\otimes \cdots \otimes u_j) = f(u_{\sigma^{-1}(1)}\otimes \cdots \otimes u_{\sigma^{-1}(j)}).\]
But $f$ is injective, hence this implies that any simple tensor is fixed by $\sigma$. This happens only if $\sigma$ is the identity (since the dimension of $\cV$ is greater than one).
For the last claim, since  $\cV$ is infinite dimensional, then $\I(\cW,\cV)$ is contractible for any $\cW$. See \cite[Lemma 1.3]{MayEinfty}.

For $\bsK(\cV)$ we simply note that $\bsK(\cV)(j)$ is a subspace of $\sK(\cV)(j)$ and the structure maps are simply restrictions. Thus $\bsK(\cV)$ is an operad and it is reduced since
the condition $\varphi(1) = u$ completely determines a linear map $\varphi:\C\to\cV$. Furthermore, since the restriction map 
\[\bsK(\cV) = \I_u(\cV^{\otimes j},\cV) \to \I(\cV^{\otimes j} - \C u^{\otimes j},\cV - \C u)\] 
to the space of linear embeddings on the orthogonal complements is a homeomorphism, and the target is contractible, we have that $\bsK(\cV)$ is indeed an $\sE_\infty$-operad. 
\end{proof}

\begin{defn}\label{defn:Ospacestructure}
Let $(\QS_\Gamma,\etagQ)$ be a quantum state type and $\cV\in \sJ$ be infinite dimensional. Let 
\[ \theta_j \colon \sK(\cV)(j) \times  \QS_\Gamma(\cV)^j \to \QS_\Gamma(\cV)\]
be given by
\[ \theta_j(f; \omega_1, \ldots, \omega_j) = \QS_\Gamma(f)(\etagQ(\omega_1,\ldots,\omega_j)) = \QS_\Gamma(f)(\omega_1\otimes \cdots \otimes \omega_j)\]
for $f\in  \sK(\cV)(j)=\I(\cV^{\otimes j}, \cV)$ and $\omega_1, \ldots, \omega_j \in \QS_\Gamma(\cV)$. 

For a choice of $u\in \cV$, give  $\QS_\Gamma(\cV)$ the base point
\[\QS_\Gamma(u) \colon \QS_\Gamma(\C) \to \QS_\Gamma(\cV)\]
where we write $u\colon \C \to\cV$ for the linear map which sends $1$ to $u$. Let
\[ \theta_j^u \colon \bsK(\cV)(j) \times  \QS_\Gamma(\cV)^j \to \QS_\Gamma(\cV)\]
be the restrictions of $\theta_j$ along the inclusion $ \bsK(\cV)(j)  \subset  \sK(\cV)(j) $.
\end{defn}

The proof of the following result is as in \cite[\S 1]{MayEinfty} for the linear isometry operad so we do not repeat the arguments here.
\begin{prop}\label{prop:Kspaces}
The space $\QS_\Gamma(\cV)$ together with the structure maps of \cref{defn:Ospacestructure} is an $\sK(\cV)$-space.   For $u \in \cV$, 
$\QS_\Gamma(u) \colon \QS_\Gamma(\C) \to \QS_\Gamma(\cV)$ is a based  $\bsK(\cV)$-space.
\end{prop}

\begin{proof}[Proof of \cref{thm:QSEinfty}.]
The first claim follows from \cref{prop:Kspaces} and the fact that $\bsK(\cV)$ are $\sE_\infty$-operads.

We give a bit more information on how the structure maps of \cref{defn:Ospacestructure} allow us to define the $H$-space structure. This is very standard so we do not give all the details.

Let $\cV$ be an infinite dimensional object in $\sJ$ and $u\colon \C \to \cV$ our choice of linear isometric embedding.  Any element 
\[\varphi \in \bsK(\cV)(2) = \I(\cV\otimes \cV, \cV)\] determines a binary operation
\[ \theta := \theta_2(\varphi;-) \colon \QS_\Gamma(\cV)^2 \to \QS_\Gamma(\cV)\]
which is the multiplication for the $H$-space structure. Concretely,
\[\theta(\omega, \tau) = \QS_\Gamma(\varphi)(\omega\otimes \tau), \]
so $\varphi$ is used to \emph{internalize} the tensor product. The key is that, because the space $\I_u(\cV\otimes \cV, \cV)$ is contractible, all choices give homotopic products, and the homotopies between those homotopies are homotopic, and so on. So,  the choice was ``unique'' up to all higher homotopies,

The unit for the $H$-space structure is the base point $\psi_0:=\QS_\Gamma(u)(1)$ of $\QS_\Gamma(\cV)$.
The  properties of an $ \bsK(\cV)$-space, and the fact that $\bsK(\cV)(j)$ are all contractible, then imply that $\QS_\Gamma(\cV)$ is an $H$-space. 
See, for example, \cite[Lemma 1.9]{MayGILS} for a detailed proof of homotopy associativity.

If $\QS_\Gamma(\cV)$ is group-like, i.e., $\pi_0\QS_\Gamma(\cV)$ is a group, then since $\QS_\Gamma(\cV)$ is a homotopy commutative and associative $H$-space, it is in fact an $H$-group.\footnote{Here, we use the fact that we have implicitly replaced $\QS_\Gamma(\cV)$ by the CW-complex $\rep \QS_\Gamma(\cV)$.} In fact, all of the path components of $\QS_{\Gamma}(\cV)$ are homotopy equivalent to the path component $\QS_\Gamma(\cV)_0$ of $\psi_0$. This is an $H$-group in the sense that there is a map 
\[i \colon \QS_\Gamma(\cV) \to \QS_\Gamma(\cV)\]
so that
\[ \theta(-, i(-)) \simeq  \id \simeq \theta( i(-),-). \]
Furthermore, there is an equivalence of $H$-spaces
\[\QS_\Gamma(\cV) \simeq  \QS_\Gamma(\cV)_0 \times \pi_0\QS_\Gamma(\cV)  \]
where $ \pi_0\QS_\Gamma(\cV)$ is given the discrete topology.
 These things are shown in \cite[Lemmas 9.2.2 \& 9.2.3]{MayPonto}.
 
Next, consider two choices $(u_1 \in \cV_1)$ and $(u_2\in \cV_2)$ where $\cV_1$ and $\cV_2$ are infinite dimensional. Let $f \in \I(\cV_1, \cV_2)$ be an isomorphism so that $f(u_1)=u_2$. Then $\QS_\Gamma(f)$ is a homeomorphism  of based spaces 
\[\QS_\Gamma(f) \colon (\QS_\Gamma(\cV_1), \QS_\Gamma(u_1)) \xrightarrow{\ \cong \ } \QS_\Gamma(\cV_2), \QS_\Gamma(u_2)) \] 
 meaning the diagram
\[ \xymatrix@C=3pc{ \QS_\Gamma(\C) \ar@{=}[r]\ar[d]_{\QS_\Gamma(u_1)} &\QS_\Gamma(\C) \ar[d]^{\QS_\Gamma(u_2)} \\
\QS_\Gamma(\cV_1) \ar[r]_-{\QS_\Gamma(f)}^-{\cong} & \QS_\Gamma(\cV_2)
}\] commutes. 
Furthermore, for any $j$, we get a commuting diagram
\[\xymatrix@C=3pc{
 \mathscr{L}_{u_1}^{\otimes}(\cV_1)(j) \ar[r]^{\cong} \ar[d]_-{\subset} &  \mathscr{L}_{u_2}^{\otimes}(\cV_2)(j) \ar[d]^-{\subset} \\
 \sK(\cV_1)(j) \ar[r]^{\cong} &  \sK(\cV_2)(j)  }
\]
where the horizontal isomorphisms are given by
\[\varphi \mapsto  f\varphi (f^{-1})^{\otimes j}.\] 
One can check that these isomorphisms commute with the strucutre maps so that we have an isomorphism of operads. From this, it is straightforward to check that we have an isomorphism of $\sE_\infty$-spaces.
 \end{proof}

It remains to prove \cref{thm:jmathHspace}.

\begin{proof}[Proof of \cref{thm:jmathHspace}]
Letting $\cV$ be $\cH^{\otimes \infty}$, then \cref{thm:jmathHspace} implies the first statement of \cref{thm:jmathHspace}. It remains to prove that $\jmath$ is a map of $H$-spaces.

The $H$-space structure is independent of the choice of 
\[\varphi \in \I(\cH^{\otimes\infty}\otimes\cH^{\otimes\infty}, \cH^{\otimes \infty} ),\] so for convenience, we choose $\varphi$ defined as follows. Write $v_1\otimes v_2\otimes \cdots$ with $v_i = v$ for $i> m$  to mean the equivalence class of a simple $v_1\otimes \cdots \otimes v_m$ in $\cH^{\otimes \infty}$.
Let
\[\varphi((v_1\otimes v_2 \otimes \cdots )  \otimes (w_1 \otimes w_2\otimes \cdots ))= v_1\otimes w_1\otimes v_2\otimes w_2\otimes \cdots \]
We need to prove that
\[\xymatrix@C=5pc{\QS_\Gamma^{\sqcup}(\cH)  \times \QS_\Gamma^{\sqcup}(\cH)  \ar[r]^-{\otimes}\ar[d]_{\jmath \times \jmath}  & \QS_\Gamma^{\sqcup}(\cH)  \ar[d]^-{\jmath}\\
 \QS_\Gamma(\cH^{\otimes \infty})   \times  \QS_\Gamma(\cH^{\otimes \infty})   \ar[r]_-{\QS_\Gamma(\varphi)\circ \etagQ} &   \QS_\Gamma(\cH^{\otimes \infty})
}\]
commutes up to homotopy. It suffices to 
find, for each $m,n\geq 0$, a homotopy between the two composites of
\[\xymatrix@C=5pc{\QS_\Gamma(\cH^{\otimes m})  \times \QS_\Gamma(\cH^{\otimes n})  \ar[r]^-{\eta^{\QS_\Gamma}}\ar[d]_{\jmath_m \times \jmath_n}  & \QS_\Gamma(\cH^{\otimes(m+n)}) \ar[d]^-{\jmath_{m+n}}\\
 \QS_\Gamma(\cH^{\otimes \infty})   \times  \QS_\Gamma(\cH^{\otimes \infty})   \ar[r]_-{\QS_\Gamma(\varphi) \circ \eta^{\QS_\Gamma}} &   \QS_\Gamma(\cH^{\otimes \infty})  .
}\]
First note that $\jmath_m = \QS_\Gamma(i_m)$, where
\[i_m:\cH^{\otimes m}\to\cH^{\otimes\infty}\] is the inclusion map.
The diagram above can be expanded into
\[
\xymatrix@C=3pc{
\QS_\Gamma(\cH^{\otimes m})\times\QS_\Gamma(\cH^{\otimes n})\ar[r]^-{\eta^{\QS_\Gamma} }\ar[d]_-{\jmath_m\times\jmath_n} &
\QS_\Gamma(\cH^{\otimes m}\otimes\cH^{\otimes n})\ar[d]_-{\QS_\Gamma(i_m\otimes i_n)} \ar[r]^-{\jmath_{m+n}} & \QS_\Gamma(\cH^{\otimes\infty})  \\
\QS_\Gamma(\cH^{\otimes\infty})\times\QS_\Gamma(\cH^{\otimes\infty})\ar[r]_-{\eta^{\QS_\Gamma}} & \QS_\Gamma(\cH^{\otimes\infty}\otimes\cH^{\otimes\infty})\ar[ur]_-{\QS_\Gamma(\varphi)} 
}
\]
\noindent where, by naturality of $\eta^{\QS_\Gamma}$, the lefthand square commutes.
It therefore, suffices to find a homotopy
\[\gamma_{m,n}:\QS_\Gamma(\varphi)\circ\QS_\Gamma(i_m\otimes i_n) \simeq\jmath_{m+n}.\]
To do this, choose a path 
\[\alpha_{m,n} \colon \interval \to\I(\cH^{\otimes\infty}\otimes\cH^{\otimes\infty}, \cH^{\otimes \infty} )\] 
from $\varphi$ to
\begin{align*}
\varphi_{m,n}((v_1\otimes v_2 \otimes \cdots ) & \otimes (w_1 \otimes w_2\otimes \cdots )) = \\
 & v_1\otimes \cdots\otimes  v_m \otimes w_1\otimes \cdots \otimes w_n \otimes v_{m+1}\otimes w_{n+1} \otimes \cdots .
 \end{align*}
 Since $\QS_\Gamma \colon \sJ \to \CGTop$ is a topologically enriched functor, the map
\[\QS_\Gamma \colon \I((\cH^{\otimes \infty})^{\otimes 2}, \cH^{\otimes \infty} ) \to \Map(\QS_\Gamma ((\cH^{\otimes \infty})^{\otimes 2}),\QS_\Gamma( \cH^{\otimes \infty} )) \]
which sends $\varphi$ to $\QS_\Gamma (\varphi)$ is continuous. Hence, so is the restriction along $\alpha_{m,n}$. This is the map
\[ \QS_\Gamma \circ \alpha_{m,n} \colon \interval \to \Map(\QS_\Gamma (\cH^{\otimes \infty}\otimes \cH^{\otimes \infty} ),\QS_\Gamma( \cH^{\otimes \infty} ))\]
which sends $t$ to $\QS_\Gamma(\alpha_{m,n}(t))$.  Since $\CGTop$ has an exponential law, the adjoint  
\[\beta_{m,n} \colon  \interval \times \QS_\Gamma (\cH^{\otimes \infty}\otimes \cH^{\otimes \infty} ) \to \QS_\Gamma( \cH^{\otimes \infty} )  \]
given by
\[\beta_{m,n}(t,\omega) = \QS_\Gamma(\alpha_{m,n}(t))(\omega) \]
 is also continuous. The homotopy is then
\[\gamma_{m,n} = \beta_{m,n} \circ ( \id_\interval \times \QS_\Gamma(i_m\otimes i_n))\colon I\times\QS_\Gamma(\cH^{\otimes m}\otimes\cH^{\otimes n}) \to   \QS_\Gamma( \cH^{\otimes \infty} ) .\qedhere \]
\end{proof}

\subsection{Parametrized phases and group completion}\label{sec:groupcompletion}

In this section, we fix a quantum state type $\QS_\Gamma$ and a pair $(\cH,\psi)$ consisting of an object $\cH$ of $\fHilb$ and a pure state $\psi$ of $\cH$. We also let $X$ denote a parameter space in the sense of \cref{defn:parameterspace}.

Since $\QS_\Gamma(\cH)$ is a strictly associative and unital topological monoid  in $\CGTop$, then so is $\Map(X, \QS_\Gamma(\cH))$ with respect to the monoidal product 
\[(f\otimes g)(x) = f(x)\otimes g(x).\]
Similarly,  the space $\Map(X, \QS_\Gamma(\cH^{\otimes \infty}))$ is an $\sE_\infty$-space since $\QS_\Gamma(\cH^{\otimes \infty})$ has that structure. Furthermore, post-composition with $\jmath$ gives a map of $H$-spaces
\[ \Map(X,\jmath) \colon \Map(X, \QS_\Gamma(\cH)) \to \Map(X, \QS_\Gamma(\cH^{\otimes \infty})) .\]
In this section, we will study this map and its relationship with group completion.

First, note that applying $\pi_0$ to an $H$-space produces a monoid. Furthermore, for any space $Y$, 
\[  [X,Y] = \pi_0\Map(X,Y).\] Hence, we get a morphism of commutative monoids
\[ [X,\jmath] \colon [X,  \QS_\Gamma(\cH)] \to [X, \QS_\Gamma(\cH^{\otimes \infty})] .\]

\begin{rem}Applying $\pi_0$ to $\Map(X, \QS_\Gamma(\cH))$ is the step of taking deformation classes in the passage from quantum systems to phases, as described in \cref{steps} \eqref{def}. We next explain how  $[X,\jmath] $ is precisely the quotient by stacking stabilization, as described in \cref{steps} \eqref{stab}.   This justifies our terminology of calling $[X, \QS_\Gamma(\cH^{\otimes \infty})]$ the monoid of quantum phases parametrized by $X$.
\end{rem}

\begin{thm}\label{thm:jquotientiso}
For a parameter space $X$, the map $\jmath$ induces a map of monoids
\[\pi_0\Map(X,\jmath) \colon \Map(X,  \QS_\Gamma^{\sqcup}(\cH)) \to [X,  \QS_\Gamma(\cH^{\otimes \infty})], \]
which identifies $[X,  \QS_\Gamma(\cH^{\otimes \infty})]$ with the quotient of the space of quantum systems
\[ \Map(X,  \QS_\Gamma^{\sqcup}(\cH))\] 
by the relation $\omega_1 \approx \omega_2$ if and only if there are $i,j \geq 0$ such that
\[ \omega_1 \otimes \psi_i \simeq \omega_2 \otimes \psi_j.\]
\end{thm}
\begin{proof}
That $\approx$ is an equivalence relation that preserves the monoidal operation is straightforward to verify.

Since $X$ is compact the canonical map
\[  \colim_{i,\Map(X,\QS_\Gamma(\ff_i))} \Map(X, \QS_\Gamma(\cH^{\otimes i})) \to \Map(X, \QS_\Gamma(\cH^{\otimes \infty}))\]
is a homeomorphism. But $\pi_0$ also commutes with directed colimits along closed embeddings since both $\pt$ and $\interval$ are compact, so applying $\pi_0$ to both sides gives an isomorphism
\[  \colim_{i,[X,\QS_\Gamma(\ff_i)]} [X, \QS_\Gamma(\cH^{\otimes i})] \xrightarrow{ \ \cong \ } [X, \QS_\Gamma(\cH^{\otimes \infty})].\]

The morphism
\[[X,\jmath] \colon[X, \QS_\Gamma^{\sqcup}(\cH)] \to \colim_{i,\QS_\Gamma(\ff_i)} [X,\QS_\Gamma(\cH^{\otimes i}) ] \cong [X, \QS_\Gamma(\cH^{\otimes \infty})]\]
is clearly surjective. 
So, it remains to show that $\jmath(f)=\jmath(g)$ if and only if $f\approx g$. 

First observe that, since a parameter space $X$ is compact, it has a finite number of path components $X_\alpha$. For $f \colon X \to \QS_\Gamma^{\sqcup}(\cH)$, let $f_\alpha$ be the restriction of $f$ to $X_\alpha$. Then $f\approx g$ implies that $f_\alpha \approx g_\alpha$ for each $\alpha$. But the converse also holds. Indeed, if $f_\alpha \otimes \psi_{n_\alpha}  \simeq  g_\alpha \otimes\psi_{m_\alpha}$, then $f \otimes \psi_n \simeq g\otimes \psi_m$ for  $n = \sum n_\alpha$ and $m = \sum m_\alpha$.
 
Therefore, it suffices to justify the claim when $X$ is connected.
In this case,
\[[X,\QS_\Gamma^{\sqcup}(\cH)] \cong \coprod_{i\geq 0} [X,\QS_\Gamma(\cH^{\otimes i})].\]
Any element in the colimit is equal to $\jmath_k  f$ for some $f \colon X \to \QS_\Gamma(\cH^{\otimes k})$. Furthermore, for $g\colon X \to \QS_\Gamma(\cH^{\otimes \ell})$,  $\jmath_k  f  \simeq \jmath_\ell  g $ if and only if there are $n$ and $m$ so that $f\otimes \psi_n \simeq g\otimes \psi_m$. This is because
\[\QS_\Gamma(\ff_{i+j-1})\cdots \QS_\Gamma(\ff_{i+1})\QS_\Gamma(\ff_i)(f)= f\otimes \psi_j.\]
Hence, for $X$ connected, the claim holds.
\end{proof}

\bigskip

The remainder of this section focuses on relating  the passage from quantum systems to phases, to the process of Grothendieck group completion (a.k.a., $K$-theory). This material is not necessary for the rest of the paper, so the reader who is not interested in this relationship can skip to \cref{sec:loopspectra}. 

\bigskip
To relate our framework to group completion, we will need to introduce another space. 
\begin{defn}\label{defn:QSC}
Let $\QS_\Gamma$ be a quantum state type. Let $\cH$ be an object of $\fHilbo$ of dimension at least two and $\psi$ be a pure state of $\cH$.
Let
\[ \QSC (\cH) = \colim_{j, \coprod \QS_\Gamma(\ff_i)}  \QS_{\Gamma}^{\sqcup} (\cH),\]
where $ \coprod \QS_\Gamma(\ff_i)(\omega) = \omega\otimes \psi$.
That is, $ \QSC (\cH) $ is the colimit of the diagram
\[ \xymatrix@C=3pc{\QS_{\Gamma}^{\sqcup} (\cH)  \ar[r]^-{\coprod \QS_\Gamma(\ff_i)} \ar[r] &   \QS_{\Gamma}^{\sqcup} (\cH)   \ar[r]^-{\coprod \QS_\Gamma(\ff_i) }\ar[r]  &   \QS_{\Gamma}^{\sqcup} (\cH)   \ar[r]^-{\coprod \QS_\Gamma(\ff_i) }\ar[r]  & \cdots 
 }
\]
Let
\[\iota_k \colon  \QS_{\Gamma}^{\sqcup} (\cH)  \to  \QSC (\cH)  \]
be the inclusion in the colimit through the $k$th term of the diagram. We base $\QSC (\cH)$ at $\iota_0 \circ \mu_0$.

For any $k\in \Z$, let
\[\kappa_k \colon \QS_{\Gamma}(\cH^{\otimes \infty}) \to  \QSC (\cH) \]
be the unique map satisfying
\[  \kappa_k\circ \jmath_{k+\ell} = \iota_\ell \circ \mu_{k+\ell}.\]
\end{defn}

\begin{rem}
If  $  \QS_{\Gamma}^{\sqcup} (\cH)_\ell$ denotes the $\ell$th copy of $ \QS_{\Gamma}^{\sqcup} (\cH)$ in the colimit diagram defining $\QSC(\cH)$, then $\kappa_k$ is the map on colimits induced by the map of diagrams
\[ \xymatrix@C=3pc{ \QS_\Gamma(\cH^{\otimes (k+\ell)}) \ar[d]^-{\subset}_-{\mu_{k+\ell}} \ar[r]^-{\QS_\Gamma(\ff_{k+\ell})}  & \QS_{\Gamma} (\cH^{\otimes (k+\ell+1)})  \ar[d]^-{\subset}_-{\mu_{k+\ell+1}}\ar[r]^-{\QS_\Gamma(\ff_{k+\ell+1})}  &   \QS_{\Gamma} (\cH^{\otimes (k+\ell+2)}) \ar[d]^-{\subset}_-{\mu_{k+\ell+2}}  \ar[r]^-{\QS_\Gamma(\ff_{k+\ell+2})}  & \cdots  \\
  \QS_{\Gamma}^{\sqcup} (\cH)_\ell   \ar[r]^-{\coprod \QS_\Gamma(\ff_i)} & \QS_{\Gamma}^{\sqcup} (\cH)_{\ell+1}  \ar[r]^-{\coprod \QS_\Gamma(\ff_i)} &   \QS_{\Gamma}^{\sqcup} (\cH)_{\ell+2}   \ar[r]^-{\coprod \QS_\Gamma(\ff_i)}  & \cdots  
 }\]
where it is understood that we only begin the diagram when $k+\ell, k, \ell\geq 0$. 
\end{rem}

\begin{ex}\label{rem:invert}
The construction of $ \QSC(\cH)$ is a generalization of the following standard algebraic construction.
Let $M$ be a commutative monoid (in the algebraic rather than topological sense) and choose an element $m \in M$. Define
\[m^{-1}M :=
\colim \left( \xymatrix{ M \ar[r]^-{ m} & M \ar[r]^-{ m} & M \ar[r]^-{ m} & \cdots} \right) \]
and let
\[\iota_0 \colon M \to m^{-1}M\]
be the map coming from the inclusion of the first factor of $M$ in the colimit diagram. 
Then $\iota_0$ has the following universal property. For any morphism of monoids $\phi \colon M\to N$ such that $\phi(m)$ is invertible in $N$,  there is a factorization in the category of monoids:
\[\xymatrix{M \ar[dr]_-{\iota_0} \ar[rr]^-{\phi} &  &  N \\
 & m^{-1}M  \ar[ur]_-{m^{-1}\phi} & }.  \]
 For example, for $\N_0$  the non-negative integers under addition, the inclusion $\N_0 \to \Z$ induces an isomorphism
 \[1^{-1}\N_0 = \colim \left( \xymatrix{ \N_0 \ar[r]^-{+1} & \N_0 \ar[r]^-{+1} & \cdots}\right)  \cong \Z.\]
\end{ex}

\begin{lem}
For any parameter space $X$, 
\[[X,\iota_0]\colon [X,\QS_{\Gamma}^{\sqcup} (\cH) ] \to [X,\QSC (\cH) ] \]
 is the localization of the monoid $[X,\QS_{\Gamma}^{\sqcup} (\cH) ]$ at the constant map at $\psi$.
\end{lem}
\begin{proof}
Based on \cref{rem:invert}, it suffices to show that
\[  [X,\QSC (\cH) ]   \cong  \colim_{i,  \psi} [X, \QS_{\Gamma}^{\sqcup} (\cH) ] =  \psi^{-1} [X,\QS_{\Gamma}^{\sqcup} (\cH) ]. \]
As in  the proof of \cref{thm:jquotientiso}, this holds because $X$ is compact.
\end{proof}

\begin{rem}
Since the localization of a monoid is again a monoid, we get a monoid structure on $ [X,\QSC (\cH) ]$. Henceforth, this is the structure we mean when discussing algebraic properties of this set.
\end{rem}

\begin{defn}\label{defn:prodisomaps}
Let $\QS_\Gamma$ be a quantum state type. 
 Let
\[\psi^{-1}\jmath \colon \QSC(\cH) \to \QS_{\Gamma}(\cH^{\otimes \infty})   \]
be
the map induced by the diagram 
\[\xymatrix{
\QS_{\Gamma}(\cH)_0 \ar[r]^-{\psi} \ar[d]^-{\jmath} & \QS_{\Gamma}(\cH)_1  \ar[r]^-{\psi} \ar[d]^-{\jmath}& \QS_{\Gamma}(\cH)_2  \ar[r]^-{\psi}   \ar[d]^-{\jmath} & \cdots \\
 \QS_{\Gamma}(\cH^{\otimes \infty}) \ar[r]^-{=} &  \QS_{\Gamma}(\cH^{\otimes \infty})  \ar[r]^-{=}  &  \QS_{\Gamma}(\cH^{\otimes \infty})  \ar[r]^-{=}  & \cdots .
}\]
That is, $\psi^{-1}\jmath$ is the unique map so that for $k,\ell\geq 0$,
\[\psi^{-1}\jmath\circ\iota_\ell\circ\mu_k = \jmath_k.\] 

 Let $\N_0$ have the discrete topology. Let
\[\cc \colon \QS_\Gamma^{\sqcup}(\cH) \to \N_0\]
be the map of topological monoids which takes $\QS_\Gamma(\cH^{\otimes i})$ to $i$.
 Let
\[\psi^{-1}\cc \colon \QSC \to \cc(\psi)^{-1} \N_0 \cong \Z\]
be the induced map on localizations, where $\Z$ has the discrete topology.
This is the map induced by the diagram
\[\xymatrix{
  \QS_\Gamma^{\sqcup}(\cH)_0  \ar[r]^-{\psi} \ar[d]^-{\cc} &  \QS_\Gamma^{\sqcup}(\cH)_1  \ar[r]^-{\psi}  \ar[d]^-{\cc}  &  \QS_\Gamma^{\sqcup}(\cH)_2   \ar[r]  \ar[d]^-{\cc} & \cdots \\
 \N_0 \ar[r]^-{+1} & \N_0 \ar[r]^-{+1} & \N_0  \ar[r]^{+1} & \cdots .
}\]
That is, $\psi^{-1}\cc$ is the unique map so that for $k,\ell\geq 0$,
\[\psi^{-1}\cc\circ\iota_\ell\circ\mu_{k} = k - \ell.\]
\end{defn}

\begin{lem}\label{lem:isoprod}
Let $\QS_\Gamma$ be a quantum state type. Then
\[ \psi^{-1}\jmath \times \psi^{-1}\cc \colon \QSC (\cH) \to  \QS_\Gamma(\cH^{\otimes \infty}) \times \Z  \]
is a homeomorphism.
\end{lem}
\begin{proof}
Let $\kappa$ denote the map
\[ \QS_{\Gamma}(\cH^{\otimes \infty})  \times \Z \cong \coprod_\Z  \QS_{\Gamma}(\cH^{\otimes \infty})\xrightarrow{\coprod_{k\in\Z}\kappa_k}\QSC (\cH) .\]
Both $\kappa$ and $\psi^{-1}\jmath\times\psi^{-1}\cc$ are continuous by construction. We will show that they are inverses.

Take $(\jmath_j(\omega), k)\in\QS_\Gamma(\cH^{\otimes\infty})\times\Z$ for
some $\omega\in\QS_\Gamma(\cH^{\otimes j})$. We can and will choose $j$ so that
$i = j-k \geq 0$. Now
\[\kappa (\jmath_j(\omega),k) = \kappa_k\circ\jmath_j(\omega) =
\kappa_k\circ\jmath_{i+k}(\omega) = \iota_i\circ \mu_{i+k}(\omega),\]
and
\[(\psi^{-1}\jmath\times\psi^{-1}\cc) (\iota_i\circ\mu_{i+k}(\omega))
 = (\jmath_{i+k}(\omega),k). \]
 
Now take $\iota_i\circ\mu_j(\omega)\in\QSC (\cH)$, for some $\omega\in
\QS_\Gamma(\cH^{\otimes j})$. Then
\[\psi^{-1}\jmath\times\psi^{-1}\cc(\iota_i\circ\mu_j(\omega)) =
(\jmath_j(\omega), j-i),\] and
\[\kappa(\jmath_j(\omega),j-i) = \kappa_{j-i}\circ\jmath_j(\omega) =
\iota_i\circ\mu_j(\omega). \qedhere\]
\end{proof}

\begin{rem}\label{lem:kappakim}
The image of $\kappa_k \colon \QS_{\Gamma}(\cH^{\otimes \infty}) \to  \QSC (\cH) $ as in  \cref{defn:classifyingspaces},
 consists of the elements that can be expressed as  $\iota_j(\mu_{k+j}(\omega))$ for
$ \omega\in \QS_\Gamma(\cH^{\otimes (k+j)})$.
\end{rem}

\begin{prop}\label{prop:K0exact}
Let $X$ be a parameter space.
For any quantum state type $\QS_\Gamma$, there is an isomorphism of monoids
\[[X, \psi^{-1}\jmath] \  \times [X, \psi^{-1}\cc ] \ \colon[X, \QSC(\cH)] \xrightarrow{ \ \cong \ } [X, \QS_\Gamma(\cH^{\otimes \infty})] \times [X, \Z]  .\]
\end{prop}
\begin{proof}
This is a bijection by \cref{lem:isoprod}. We just need to justify why these are maps of monoids. But $\cc$ and $\jmath$ are maps of commutative and associative $H$-spaces, so this follows from the fact that 
\[[X,\psi^{-1}\cc] = \psi^{-1}[X,\cc], \quad \quad [X,\psi^{-1}\jmath] =\psi^{-1} [X,\jmath] . \qedhere\] 
\end{proof}

We are almost ready to state our group completion results. This result only applies to invertible quantum state types.

\begin{lem}\label{lem:pi1abelian}
Let $\QS_\Gamma$ be an invertible quantum state type. Let $\cV$ be an infinite dimensional object in $\sJ$. Let  $u\in \cV$, and give  $\QS_\Gamma(\cV)$ the base point $\QS_\Gamma(u)$.
\begin{enumerate}
\item  For any space $X$, the monoid $[X, \QS_\Gamma(\cV)]$ obtained from the $H$-space structure of $\QS_\Gamma(\cV)$ is an abelian group. 
\item The fundamental group of $\QS_\Gamma(\cV)$ at any choice of base point is abelian.  
\end{enumerate}
\end{lem}
\begin{proof}
The first statement is an easy consequence of the fact, for any choice of $u\in \cV$, an invertible quantum state type gives rise to homotopy commutative $H$-group  $\QS_\Gamma(\cV)$ with unit $\QS_\Gamma(u)$. The second follows from the fact that $\QS_\Gamma(\cV) \simeq \QS_\Gamma(\cV)_0 \times \pi_0\QS_\Gamma(\cV)$. Since $ \QS_\Gamma(\cV)_0 $ is an $H$-space, its fundamental group is abelian by the Eckmann--Hilton argument.
\end{proof}

Invertibility implies that for any $\omega \in \QS_\Gamma(\cH^{\otimes i})$, there exists a state $\tau \in \QS_\Gamma(\cH^{\otimes j})$ with a path from $\omega \otimes \tau$  to $\psi_{i+j}$. But we have a established the following stronger result, which we rephrase for clarity.
\begin{prop}\label{prop:pointwiseinvertible}
Let $X$ be a parameter space. If $\QS_\Gamma$ is an invertible quantum state type, then for any $f \in [X, \QS_\Gamma(\cH^{\otimes i})]$, there exists $g \in [X, \QS_\Gamma(\cH^{\otimes j})]$ such that $f\otimes g$ is homotopic to the constant map at $\psi_{i+j} \in [X, \QS_\Gamma(\cH^{\otimes i+j})]$. In particular, pointwise invertibility implies parametrized invertibility.
\end{prop}
\begin{proof}
By \cref{lem:pi1abelian},  $[X, \QS_\Gamma(\cH^{\otimes \infty})]$ is an abelian group. So $\jmath_i(f)$ has an inverse, say $\jmath_\ell(h)$. That is, 
\[\jmath_i(f)\otimes  \jmath_\ell(h) \simeq \jmath_0(\psi_0).\]
But $\jmath$ is a map of $H$-spaces, hence
\[\jmath_i(f)\otimes  \jmath_\ell(h)  \simeq \jmath_{i+\ell}(f\otimes g). \]
It follows that for some $k$, 
\[ f\otimes h \otimes \psi_k \simeq \psi_{i+\ell+k}.\]
So, we take $g=h \otimes \psi_k$ and $j=\ell+k$.
\end{proof}
\begin{rem}\label{rem:invertiblepointwisediscussion}
In \cite[\S III]{qpump}, we studied a gapped Hamiltonian $H$ parametrized by $X$. We pointed out that while an invertible parametrized system $H$ over $X$ has the property that for each point of $x$, $H(x)$ is an invertible system, the converse is not clear.  \cref{prop:pointwiseinvertible} gives an answer to this question for invertible quantum state types. As explained in the introduction, assuming the weak equivalence between a gapped Hamiltonian model and a state model for invertible gapped quantum systems,  \cref{prop:pointwiseinvertible}  should also settle this question for Hamiltonians. 
\end{rem}

We are finally ready to compare the passage from quantum systems to phases with the process of group completion. We first recall the definition of Grothendieck group completion.
\begin{defn}
The \emph{Grothendieck group completion} of a commutative monoid $M$ consists of an abelian group $\Kzero(M)$ together with a morphism of monoids
\[g \colon M \to \Kzero(M)\] 
(where we can always think of an abelian group as a monoid by forgetting the inverses) satisfying the following universal property. 
If $A$ is an abelian group
and $f \colon M \to A$ is a morphism of monoids,  then there is a unique group homomorphism $\Kzero(\phi) \colon \Kzero(M) \to A$ making the following diagram commute
\[\xymatrix{ M \ar[r]^-{f} \ar[d]_-g & A \\
\Kzero(M) \ar@{.>}[ur]_{\Kzero(f) }}\]
In other words, $\Kzero(-)$ is the left adjoint to the forgetful functor from abelian groups to commutative monoids.
\end{defn}

\begin{theorem}\label{thm:locisgroup}
Let $X$ be a parameter space and $\QS_\Gamma$ be an invertible quantum state type. Then
$ [X,\QSC (\cH) ]$ is an abelian group and 
\[[X,\iota_0] \colon     [X,\QS_{\Gamma}^{\sqcup} (\cH) ]  \to [X,\QSC (\cH) ] \]
is  the Grothendieck group completion. Therefore, the maps of \cref{prop:K0exact} induce an isomorphism
\[  K_0([X,  \QS_\Gamma^{\sqcup}(\cH)])  \cong [X,\QSC (\cH) ]    \xrightarrow{\cong} [X, \QS_\Gamma(\cH^{\otimes \infty}) ]\times [X,\Z]   .\]
\end{theorem}
\begin{proof}
That $ [X,\QS_{\Gamma}^{\sqcup} (\cH) ] $ is an abelian group follows from the isomorphism of monoids of \cref{prop:K0exact} and the fact that $[X, \QS_\Gamma(\cH^{\otimes \infty})]$ is a group when $\QS_\Gamma$ is invertible (\cref{lem:pi1abelian}). Since $\psi$ maps to an invertible element in the Grothendieck group $K_0 [X,\QS_{\Gamma}^{\sqcup} (\cH) ] $, the group completion factors as
\[\xymatrix{ [X,\QS_{\Gamma}^{\sqcup} (\cH) ] \ar[r]^-{g}\ar[d]_-{[X,\iota_0]} & K_0 [X,\QS_{\Gamma}^{\sqcup} (\cH) ]  \\
 [X,\QSC (\cH) ] \ar[ur]_-{\psi^{-1}g} & 
}\]
But, $g$ is the initial group homomorphism from $[X,\QS_{\Gamma}^{\sqcup} (\cH) ]$, hence $\psi^{-1}g$ is an isomorphism.
\end{proof}

We actually have the following a stronger result which concerns the \emph{topological} group completion of $\QS_\Gamma^\sqcup(\cH)$, and we finish this section by stating it. It will not be used later.

 \begin{theorem}\label{thm:quillengroupcompletion}
Let $\QS_\Gamma$ be an invertible quantum state type. The localization $\QSC(\cH)$ is weakly equivalent to Quillen's topological group completion $\Omega B  \QS_\Gamma^{\sqcup}(\cH)$ and $\QS_\Gamma(\cH^{\otimes \infty})$ is weakly equivalent to the identity path component of this loop space. 
\end{theorem}
\begin{proof}
This follows from \cite{McDuffSegal} using the fact $ \QS_\Gamma^{\sqcup}(\cH)$ is a homotopy commutative monoids, that $\pi_0\QSC(\cH)$ is already a group, and $\pi_1\QSC(\cH)$ is abelian for any choice of base point.
\end{proof}

\begin{rem}\label{rem:quillengroupcompletion}
We think that, for invertible quantum state types, \cref{thm:locisgroup} should be explained by a more topological statement. In fact, it seems likely that the canonical map
\[\Map(X, \QSC(\cH) )  \to \Omega B\Map(X, \QS_\Gamma^{\sqcup}(\cH))\]
is a weak equivalence, although we were not able to prove that. If $\QS_\Gamma^{\sqcup}(\cH)$ was an $\sE_\infty$-space, then $\Map(X, \QSC(\cH) )$ would be a grouplike $\sE_\infty$-space and so would already group complete. However,  $\QS_\Gamma^{\sqcup}(\cH)$ is not an $\sE_\infty$-space and we were not able to settle this question.
\end{rem}

\begin{rem}
There is a very nice history of the results on topological group completion given in \cite[Ch. 3]{AdamsLoop}. Early versions of topological group completion are due to Quillen in unpublished private communications among the key players, Barratt \cite{Barratt}, Barratt-Priddy \cite{BarrattPriddy}, May \cite{MayPermutative, MayClass}, McDuff--Segal \cite{McDuffSegal} and Segal \cite{SegalCoh}. More recent treatments include Randal-Williams \cite{RandalWilliams} and Nikolaus \cite{Nikolaus}. 
\end{rem}

\subsection{Loop-spectra and May's Recognition Principle}\label{sec:loopspectra}
Because of the importance of loop-spectra in this story, we dedicate this section to reviewing some key definitions. We also give more details about May's Recognition Principle.

\bigskip

We start with recalling the definition of a loop-spectrum.\footnote{In the homotopy theory literature, these are often called $\Omega$-spectra. However, the name ``loop-spectrum'' is the more popular term in condensed matter theory.}

\begin{defn}
A \emph{loop-spectrum} $\bY$ is a sequence of based spaces in $\bCGTop$
\[Y_0, Y_1, Y_2, \ldots \]
together with weak equivalences
\[ \omega_\dd \colon Y_\dd \xrightarrow{\ \simeq \ }\Omega Y_{\dd+1} .\]

A map of loop-spectra $\boldsymbol{f} \colon \boldsymbol{Y} \to \boldsymbol{Z}$ is a sequence of based continuous maps $f_\dd \colon Y_\dd \to Z_\dd$ so that $\omega_\dd\circ f_\dd = \Omega f_{\dd+1} \circ \omega_\dd$. Loop-spectra form a category which we denote by $\Omega\Sp$.
\end{defn}

\begin{defn}
For any integer $\dd$, the $\dd$th homotopy group of a loop-spectrum $\boldsymbol{Y}$ is defined to be
\[ \pi_\dd \boldsymbol{Y} := \begin{cases}
\pi_{\dd} Y_0 & \dd\geq 0 \\
\pi_{0}Y_{-\dd} & \dd<0.
\end{cases}\]
A loop-spectrum $\boldsymbol{Y}$ is \emph{connective} if $\pi_\dd \boldsymbol{Y}=0$ for all $\dd<0$.
\end{defn}

\begin{defn}
A map $\boldsymbol{f} \colon \boldsymbol{Y} \to \boldsymbol{Z}$ is a \emph{weak equivalence} of loop-spectra if $\pi_\dd\boldsymbol{f}$ is an isomorphism for all $\dd\in \Z$.
\end{defn}

For any spectrum $\bY$ and $n\geq 0$, let $\Sigma^n\bY$ be the spectrum 
\[(\Sigma^n\bY)_{\dd}= Y_{\dd+n} \]
with the same structure maps as $\bY_\dd$. Let 
$\Sigma^{-n}\bY$ be the spectrum
\[(\Sigma^{-n}\bY)_\dd= \Omega^nY_\dd,\]
with structure maps $\Omega^n\omega_\dd$. Then 
\[( \Sigma^{-n}\circ \Sigma^n)( \bY) = (\Sigma^n \circ \Sigma^{-n})(\bY)  = (\Omega^n Y_n, \Omega^{n}Y_{n+1}, \Omega^n Y_{n+2}, \ldots)\] 
and so the maps $  \omega_{\dd+n-1}\circ \cdots \circ \omega_{\dd}$ assemble to give a natural weak equivalence
\[ \bY  \xrightarrow{  \ \simeq \ } \Omega^n\Sigma^n \bY.\]
It follows that, for any $n\in \Z$, 
\[ \pi_{\dd} \bY \cong \pi_{\dd+n}\Sigma^n \bY.\]

For any spectrum $\bY$ and any integer $n$, we can construct a \emph{Postnikov stage} $\bY_{\tau {< n}}$ and a \emph{Whitehead cover} $\bY_{\tau {\geq n}}$. For a detailed review of these constructions see \cite[\S 4]{Rudyak}. But, roughly, these are loop-spectra equipped with maps
\[ \bY_{\tau {\geq n}} \to \bY \to \bY_{\tau {< n}} \]
so that 
\begin{itemize}
\item $ \bY_{\tau {\geq n}} \to \bY$ is induces an isomorphism on $\pi_\dd$ for $\dd\geq n$, and 
\item $\bY \to \bY_{\tau {< n}}$ induces an isomorphism on $\pi_\dd$ for $\dd< n$. 
\end{itemize}
The loop-spectrum $\bY_{\tau {\geq 0}} $ is called the \emph{connective cover} $\bY$. We also note that, for any $m\in \Z$, there is a weak equivalence
\[\Sigma^m(  \bY_{\tau \geq n}) \simeq (\Sigma^m \bY)_{\tau \geq n+m} ,\]
and a similar identity holds for Postnikov stages.

\begin{defn}
If $\boldsymbol{Y}$ is a loop-spectrum, we let
\[\Omega^{\infty} \boldsymbol{Y} = Y_0,\]
the zeroth space of the spectrum $\boldsymbol{Y}$. Given a map $\boldsymbol{f}$ of spectra, we let $\Omega^{\infty}\boldsymbol{f} =f_0$. 
This defines a functor
\[\Omega^{\infty} \colon \Omega\Sp \to \bCGTop.\]
We say that a space is an \emph{infinite loop space} if it is the image of $\Omega^{\infty}$.
\end{defn}

Finally, recall that a loop-spectrum determines a generalized cohomology theory. If $X$ is a space, which we always assume has the homotopy type of a CW complex, then
\[\boldsymbol{Y}^\dd(X) = [X, Y_\dd].  \]
If $X$ is a based space, we also let $\widetilde{\boldsymbol{Y}}^\dd(X) = [X, Y_\dd]_*$. Then the structure maps $\omega_\dd$ can be used to construct an isomorphism 
\[\widetilde{\boldsymbol{Y}}^\dd(X) \cong \widetilde{\boldsymbol{Y}}^{\dd+1}(\Sigma X).\] 
This is also reviewed in \cite[\S 2]{BeaudryCampbell}.

\bigskip
We are now ready to discuss May's Recognition Principle.

\begin{rem}
In general, an $\sE_\infty$-space $X$ is \emph{grouplike} if $\pi_0X$ is a group with respect to the induced $H$-space structure on $X$. So a quantum state type is invertible if the associated $\sE_\infty$-spaces $\QS_\Gamma(\cV)$ for $(\cV,u)$ is a grouplike $\sE_\infty$-space.
\end{rem}

\begin{thm}[{May's Recognition Principle \cite[\S 14]{MayGILS}}]\label{thm:recognition} 
The functor $\Omega^{\infty}$ refines to a functor
\[\Omega^{\infty} \colon \Omega \Sp \to \sE_\infty \text{-}\bCGTop,\]
where $\sE_\infty \text{-}\bCGTop$ is the category of $\sE_\infty$-spaces.
As such, it  has a left adjoint 
\[B^{\infty} \colon  \sE_\infty\text{-}\bCGTop \to   \Omega \Sp \] 
which takes an $\sE_\infty$-space $X$ and produces a loop-spectrum 
\[B^{\infty}X = (X, B^1X, B^2X, \cdots).\] 
The unit of the adjunction $X \to \Omega^{\infty}B^{\infty}X$ is weakly equivalent to the group completion $X \to \Omega BX$. This is a weak equivalence if and only if $X$ is grouplike.  
 The co-unit $B^{\infty}\Omega^{\infty}\boldsymbol{Y} \to \boldsymbol{Y}$ is weakly equivalent to a Whitehead cover $ \boldsymbol{Y}_{\tau \geq 0} \to \boldsymbol{Y}$, which is a weak equivalence if and only if $\boldsymbol{Y}$ is connective. 
\end{thm}

Combining this with what we have already established gives the following result.

\begin{theorem}\label{thm:QSinfiniteloop}
Let $\QS_\Gamma$ be a quantum state type, $\cV\in \sJ$ an infinite dimensional object, and choice of unit vector $u \in \cV$. There is a loop-spectrum $B^{\infty}\QS_\Gamma(\cV)$ whose zero space is weakly equivalent to $\QS_\Gamma(\cV)$ if and only if $\QS_\Gamma$ is an invertible state type (i.e., $\pi_0\QS_\Gamma(\cV)$ is a group). In particular, if $\QS_\Gamma$ is invertible, then $\QS_\Gamma(\cV)$ is an infinite loop-space.
\end{theorem}
\begin{rem}
Since we implicitly assume that we have replaced $\QS_\Gamma$ by its cellular approximation $\rep \QS_\Gamma$, for any $u\in \cV$, the replacement of $ \QS_\Gamma(u)$ is the inclusion of a zero simplex so is a non-degenerate base point.  Therefore, so long as we replace $\QS_\Gamma$, we can assume that we are working in $\bCGTop$.
\end{rem}


\subsection{Berry phase and free fermions}\label{sec:examples}

In the preceding sections we developed an abstract framework for modeling certain parametrized phases of matter. This framework was guided by some examples which we describe here.

\subsubsection{Gapped Quantum Systems in spacetime dimension $0\pone$}\label{sec:0d}
 It is well understood that the unique phase invariant of a quantum system in spacetime dimension $0\pone$ is the line bundle of ground states over $X$. We review this fact in the context of our framework.

Consider a $0\pone$ dimensional quantum system described by a finite dimensional Hilbert space $\cH$. Let $\bP(\cH)$ be its pure state space, and $\mathscr{H}(\cH)$ be the space of Hermitian operators on $\cH$ with a unique ground state. Then $\mathscr{H}(\cH)$ and $\bP(\cH)$ are homotopy equivalent. This means that any state in $\bP(\cH)$ is the unique ground state of some Hamiltonian, and we can model gapped parametrized systems over $X$ with a map 
\[ X \to \bP(\cH). \]
\begin{rem}
    This approach fails to model gapped parametrized systems in spacetime dimensions greater than $0\pone$ because there exist pure states which are not ground states of any gapped and local Hamiltonian. For example, it has been proven that  for local Hamiltonians, a spectral gap above the ground state implies exponential decay of correlations in the ground state \cite{Hastings_2004, Nachtergaele_2006, Hastings_2006}. It follows that a state $\omega \in \sP_\Gamma(\cH)$ (where $\Gamma = \mathbb{Z}^d$ for $d \geq 1$) which has a power-law decay of correlations is not the ground state of any local Hamiltonian.  
\end{rem}
Armed with a model for gapped parametrized systems over $X$ in spacetime dimension $0\pone$, the framework developed in the previous sections provides a prescription to pass to parametrized phases over $X$.

To connect with the results of the previous section, we can let $\Gamma$ be the zero lattice. Then
\[\bP \colon \fHilbo  \to \CGTop \]
maps $\cH$ to $\bP(\cH)$, which can be canonically identified with the space of pure states on the $C^*$-algebra $\fB(\cH)$, i.e., with $\sP(\fB(\cH))$.

Then
\[\eta^{\bP} \colon \bP(\cH)\times \bP(\cK) \to \bP(\cH \otimes \cK)\]
is the map which sends a pair of lines $\psi$ and $\phi$ to their tensor product $\psi\otimes\phi$.

If we fix a Hilbert space $\cH$ and a vacuum state $\psi \in \bP(\cH)$, then 
\[ \bP^{\sqcup}(\cH) = \coprod_{i\geq 0} \bP(\cH^{\otimes i})\]
and
\[\bP(\cH^{\otimes \infty}) = \colim \left(\xymatrix{\bP(\cH^{\otimes 0}) \ar[r]^-{\otimes \psi}  & \bP(\cH^{\otimes 1})  \ar[r]^-{\otimes \psi}  &  \bP(\cH^{\otimes 2})  \ar[r]^-{\otimes \psi}  & \cdots}
\right)\]
is an $\sE_\infty$-space.
In fact, this gives the standard  $\sE_\infty$-structure on infinite complex projective space coming from the tensor product of line bundles. Since $\bP(\cH^{\otimes \infty}) $ is connected, this is an invertible state type.  
In fact, we have reproduced a space with many names,
\[ \bP(\cH^{\otimes \infty})  \cong \C P^\infty \simeq K(\Z,2) \simeq BGL(1) \simeq BU(1).\]
Therefore, a $0+1$ dimensional phase of quantum systems parametrized by $X$ is an element of
\[[X, \bP(\cH^{\otimes \infty})] \cong \mathrm{Line}_{\C}(X) \cong H^2_{\mathrm{sing}}(X,\Z) \cong \breve{H}^1(X, U(1)).\]
Here, the second group is the isomorphism classes of complex line bundles, the third is the singular cohomology group, and the fourth is \u{C}ech cohomology.

\begin{rem}
One way to obtain the  $\sE_\infty$-structure  on $BU(1)$ is using the fact that $U(1)$ is an abelian group. Indeed, for any topological abelian group, we can iterated the bar construction and so
\[ (BU(1), B^2U(1), B^3U(1) \cdots)\]
is a de-looping of $BU(1)$. See, for example, Segal's $\Gamma$-space theory for generalizations for this construction.
It could be that our construction of the infinite loop space structure on $ \bP(\cH^{\otimes \infty}) $ using the operad $\sK(\cH^{\otimes \infty})$ appears somewhere in the literature, but we did not find it.
\end{rem}

\subsubsection{K-theory and free fermions}\label{sec:Ktheory}
Here, we review the classification of certain free-fermion phases following \cite{Kitaev2009} and explain how the
classification is related to our framework. Although this classification is not via quantum state types, the framework for free fermions is intimately related to the one we presented above, as will become apparent in this section.  Physically, the vector bundle discussed below gives the occupied single-particle states; the ground state of the system is obtained by filling these states with fermions.

The physical situation is the following. We wish to describe families of free fermion 
systems with particle number symmetry in spacetime dimension $0\pone$. 
Such systems are characterized by a Hamiltonian of the form
\[H_A = \sum_{1\leq j,k\leq n}A_{jk}a_j^\dagger a_k,\] where $A$ 
is a Hermitian matrix of size $n$ equal to the size of the system. We consider only gapped systems, i.e. systems for which there 
exists $\Delta>0$ with $\Delta < |\varepsilon|$ for each eigenvalue $\varepsilon$ of $A$.
The set of such matrices $A$ defines a subspace of $M(\C) = \colim_n M_n(\C)$
which we will call the space of \textit{admissible matrices}. The colimit
is taken with respect to the inclusions \[M_n(\C)\to M_{n+1}(\C)\] given
by \[A\mapsto \begin{pmatrix} A & 0 \\ 0 & 1\end{pmatrix}.\]   The goal 
is then to study the homotopy type of this space. 

First, let's make a reduction to a simpler space of the same homotopy type.   
\noindent The function 
\[f_t(x) = \frac{x}{|x|^t}\] can be applied 
using functional calculus to an admissible matrix $A$ to obtain a matrix 
$\widetilde{A}= f_1(A)$ with eigenvalues $\pm 1$. This defines a deformation
retraction from the space of admissible matrices onto the space $C$ of hermitian matrices with eigenvalues $\pm 1$.
Note that the characteristic polynomial map 
\begin{align*}
\colim_n M_n(\C)&\to \C[t]\\
A&\mapsto \mathrm{det}(A-tI)
\end{align*}
is continuous. The restriction of this map to $C$ is then a continuous map into
a discrete space giving a decomposition
\[C=\coprod_{k\in\N} C_k,\] where $C_k$ is the subspace of $C$ consisting of
matrices with $(-1)$-eigenspace of dimension $k$. 

We next analyze $C_k$ further.
First note that $C_k = \colim_n C_k(n)$, where $C_k(n)$ is the subspace of $C_k$
consisting of matrices of size $n$. We then have maps into grassmanians 
\[g_{k,n}:C_k(n)\to \mathrm{Gr}_k(\C^n)\]  sending a matrix to its $(-1)$-eigenspace.
We view these as maps into 
\[BGL(k) \simeq \mathrm{Gr}_k(\C^\infty) = \colim_n\mathrm{Gr}_k(\C^n)\] 
by postcomposition
with the inclusion into the colimit. The result is a map 
\[g_k:C_k\to BGL(k) \] which by a straightforward analysis of the fiber can be shown to be a weak equivalence.

For the purposes of phase-classification we can now identify 
\[ \mathscr{F}^{\sqcup} := \coprod_{k\geq 0}BGL(k)  \simeq C\] 
as a model of the  classifying
space of $0\pone$d free fermion systems with particle number symmetry in the following sense. 
A family of such systems parametrized by a space $X$ is the same as a map $A\colon X\to \mathscr{F}^{\sqcup}$. 
The phase of this family of systems depends only on the deformation (a.k.a. homotopy) class in $[X,\mathscr{F}^{\sqcup}]$, which by this identification can be viewed as an element in 
\[[X,\mathscr{F}^{\sqcup}]\cong \mathrm{Vect}_\C(X),\] 
i.e. as an isomorphism class of a vector bundle over $X$. 
Explicitly, $A$ defines a vector bundle over $X$ whose fiber over $x\in X$ is the $(-1)$-eigenspace of $A(x)$. The deformation invariant is then the isomorphism class of this vector bundle.

However, the phase of a family of systems is not determined completely by its deformation
class, but its deformation class after stacking with arbitrary trivial systems. We must, therefore, study the operation on $\mathscr{F}^{\sqcup}$
corresponding to stacking of systems. 

In the free fermion setup, stacking two Hamiltonians determined by admissible matrices $A_0$ and $A_1$ corresponds to taking the direct sum $A_0\oplus A_1$ of the two matrices.
Accordingly, the stacking operation is modeled on $M$ by the map 
\[\mathscr{F}^{\sqcup}\times \mathscr{F}^{\sqcup}\to \mathscr{F}^{\sqcup}\] 
assembled from the maps
\[BGL(n)\times BGL(m)\to BGL(n+m)\] 
coming from the group homomorphism 
\[GL(n)\times GL(m)\to GL(n+m)\] 
which is defined by the inclusion of $(X,Y)$ as a block diagonal matrix with blocks $X$ and $Y$.
This operation is unital with respect to the point $BGL(0)$. So, $\mathscr{F}^{\sqcup}$ is a topological monoid.

Finally, we can form a directed system
\[BGL(1) \to BGL(2) \to BGL(3) \to \cdots\]
using the inclusions 
\[BGL(n)  \to BGL(n+1)\]
given by
\begin{equation}\label{eq:blockwithone}
X\mapsto \begin{bmatrix} X & 0\\
0 & 1 
\end{bmatrix}
\end{equation}
Then
\[\sF^\infty  := \colim_n BGL(n) = BGL. \]
This is the classifying space for \emph{reduced} complex $K$-theory.
On the other hand, the self map
\[\psi \colon \mathscr{F}^{\sqcup} \to \mathscr{F}^{\sqcup}\]
which does \eqref{eq:blockwithone} on each point can be ``inverted'' by forming the colimit of
\[ \psi^{-1}\mathscr{F}^{\sqcup}:= \colim \left(\xymatrix{ \mathscr{F}^{\sqcup} \ar[r]^-{\psi} & \mathscr{F}^{\sqcup} \ar[r]^-{\psi} &  \mathscr{F}^{\sqcup} \ar[r]^-\psi & \cdots}\right).\]
The resulting space is
\[ \psi^{-1}\mathscr{F}^{\sqcup} \cong BGL \times \Z \]
the classifying space for topological complex $K$-theory.

\begin{rem}
The topological monoid $\mathscr{F}^{\sqcup}$ is, famously, the monoid that gives rise to $K$-theory using Quillen's topological group completion. It appears ubiquitously in the literature discussing these topics. See \cite[Ch. 2 \& 3]{AdamsLoop} for a very nice discussion and further references.
\end{rem}

The connection with our framework is already apparent, but there are even more similarities than we have pointed out so far. In the remainder of this section, we want to make this clear. This will be more technical, as again we will discuss operads and $\sE_\infty$-spaces.

What we do next was first introduced by Boardman--Vogt \cite{BoardmanVogt} in their introduction of ``homotopy everything'' spaces. It was framed into the context of operads by May\cite{MayEinfty}.

We need a category akin to $\sJ$.
\begin{defn}[\cite{BoardmanVogt}]\label{defn:liniso}
Let $\sI$ be the category whose objects are real inner product spaces of countable algebraic dimension over $\R$. The morphisms are real linear isometries $\I( \mathcal{V},\mathcal{W})$.  This is a topologically enriched category. The finite dimensional objects have the standard metric topology. The infinite dimensional objects are topologized so that they are the topological union of their finite dimensional subspaces. Then $\I( \mathcal{V},\mathcal{W})$ is given the compactly generated compact-open topology.
 The category $\sI$ is a symmetric monoidal category with monoidal operation given by the direct sum.
\end{defn}

\begin{defn}[\cite{MayEinfty}]
Let $\cV$ be an infinite dimensional object in $\sI$.
The  operad $\sL(\cV)$ is defined as
\[\sL(\cV)(j) =\I(\cV^{\oplus j}, \cV) \]
with the identity as $1 \in \sL(\cV)(1)$ and structure map
\[\gamma \colon \sL(\cV)(k) \times \sL(\cV)(j_1)\times \cdots \times \sL(\cV)(j_k) \to \sL(\cV)(j) \]
given by
\[\gamma(g ; f_1,\ldots, f_k) = g\circ (f_1\oplus \cdots \oplus f_k).\]
The action of $\Sigma_j$ is given by $(f\sigma)(x) = f(\sigma(x))$ where $\sigma$ permutes the summands in $\cV^{\oplus j}$.

The operad $\sL(\R^{\infty})$ is called the \emph{linear isometry operad}.
\end{defn}

The operads $\sL(\cV)$ are $\sE_\infty$-operads. This is shown in \cite[\S 1]{MayEinfty}. 
\begin{defn}[\cite{BoardmanVogt}]
An  $\sI$-space $(\sT,\eta)$ consists of a topologically enriched functor 
\[\sT \colon  \sI \to \CGTop,\] and
 a natural transformation $\eta \colon \sT(\cV) \times \sT(\cW) \to \sT(\cV\oplus \cW)$
such that, $\sT(0)$ is a single point, and the $\eta$s are unital, associative and commutative (in an appropriate categorical sense) and the value of $\sT$ on infinite dimensional objects is the colimit of the value of $\sT$ its finite dimensional subspaces.
\end{defn}

\begin{theorem}[{\cite{BoardmanVogt, MayEinfty}}]
If $(\cT,\eta)$ is an $\sI$-space, and $\cV$ is infinite dimensional, then $\cT(\cV)$ is an $\sL(\cV)$-space where
\[ \theta_j \colon \sL(\cV)(j) \times \cT(\cV)^j \to \cT(\cV)\]
is given by
\[ \theta_j(f; x_1, \ldots, x_j) = (Tf)(\eta(x_1,\ldots,x_j)). \]
In particular, $\cT(\cV)$ is an $\mathscr{E}_\infty$-space.
\end{theorem}

\begin{ex}
For each finite dimensional $\cV$, let
\[\sT(\cV) = BGL(\cV_\C),\]
the classifying space of the general linear group of its complexification $\cV_\C$. Let $\eta$ be the map induced by the homomorphism
\[ GL(\cV_\C)\times GL(\cW_\C) \to GL(\cV_\C \oplus \cW_\C) \]
given by the inclusion of $GL(\cV_\C) \oplus GL(\cW_\C)$ in $GL(\cV_\C \oplus \cW_\C)$.

If $\cV$ is infintie dimensional, let $\sT(\cV) = \colim_i \sT(\cV_i)$ where $\cV_i$ runs over the finite dimensional subspaces of $\cV$. We can also extend $\cT$ to morphisms, and extend $\eta$. 
We then have a $\sI$-space and
\[\sT(\R^{\infty})  = BGL,\]
the classifying for reduced $K$-theory. This equips $BGL$ with an infinite loop space structure.
\end{ex}

\begin{rem}
There are many ways to produce $BGL\times \Z$ as an infinite loop space. Producing this space seems like a right-of-passage for any theory. One is by showing that it is the classifying space of a permutative category \cite{MayPermutative}, another is using Segal's $\Gamma$-spaces \cite{SegalCoh}.
\end{rem}


\section{The fundamental group of the pure state space}\label{sec:fundamentalgroup}

Now that we have established the \emph{universal quantum state type} $\sP_\Gamma$ as a basis for a theory of quantum state types, we turn to the question of better understanding this universal example. For example, we do not know the weak homotopy type of $\sP_\Gamma(\cH)$ for an infinite lattice $\Gamma$ (and $\cH$ of dimension at least $2$). In this section, we will compute its fundamental group.

As a matter of fact, all we need to know to do this computation is that it is the pure state space of a UHF algebra. 
So in the following we step back from lattices and state types and
prove that the pure state space of a UHF algebra is simply connected in the
weak$^*$ topology. 

We begin by recalling the natural action of a $C^*$-algebra on its
state space and then show some crucial properties we later need. This is done in
\cref{sec:prelimfund}.
Using this action we introduce in \cref{subsec:A-homotopies} a new notion of  homotopy in the state space which
essentially is
a kind of homotopy \emph{lifted} to the $C^*$-algebra. We then examine the space of density matrices
of a matrix algebra in \cref{subsec:fddensitymatrices} and finally use the results obtained
there to prove our main claim that the space of pure states of a UHF algebra is simply connected.

\subsection{Action of $\mathfrak{A}$ on its state space}\label{sec:prelimfund}
In this section, we study an action of $\mathfrak{A}$ on its state space.
\begin{definition}
Let $\fA$ be a $C^*$-algebra, let $\omega \in \sS(\fA)$, and let 
\[
\fN_\omega = \qty{A \in \fA: \omega(A^*A) = 0}
\]
be the Gelfand ideal of $\omega$. Recall that if $A \notin \fN_\omega$, then we may define 
\begin{equation}\label{eq:action_on_state_def}
A \cdot \omega \colon \fA \to  \bbC, \quad (A \cdot \omega)(B) = \frac{\omega(A^*BA)}{\omega(A^*A)}.
\end{equation}
This is a state of $\fA$. In the GNS representation $(\cH_\omega, \pi_\omega, \Omega_\omega)$ of $\omega$, it is represented by the unit vector $\pi_\omega(A)\Omega_\omega/\sqrt{\omega(A^*A)}$. Thus, if $\omega$ is pure, then so is $A \cdot \omega$. Following the definition, we see that given $A, B \in \fA$, we have $AB\notin \fN_\omega$ if and only if $B \notin \fN_\omega$ and $A \notin \fN_{B \cdot \omega}$, in which case $(AB) \cdot \omega = A \cdot (B \cdot \omega)$. Of course, if $\fA$ is unital, then we have $\1 \cdot \omega = \omega$.
\end{definition}

Let us now prove  a few elementary facts of this action.

\begin{proposition}\label{prop:invariant_state}
Let $\fA$ be a $C^*$-algebra and let $\omega \in \sS(\fA)$. If $A \in \fA$ is nonzero and satisfies $\abs{\omega(A)} = \norm{A}$, then $A \cdot \omega = \omega$. 
\end{proposition}

\begin{proof}
Let $(\cH_\omega, \pi_\omega, \Omega_\omega)$ be the GNS representation of $\omega$. By the Cauchy-Schwarz inequality, we have
\[
\abs{\omega(A)} = \abs{\ev{\Omega_\omega, \pi_\omega(A)\Omega_\omega}} \leq \norm{\Omega_\omega}\norm{\pi_\omega(A)\Omega_\omega} \leq \norm{A}.
\]
Since $\abs{\omega(A)} = \norm{A}$, we see that the Cauchy-Schwarz inequality is saturated, so $\pi_\omega(A)\Omega_\omega = \lambda \Omega_\omega$ for some $\lambda \in \bbC$. Observe that $\lambda \neq 0$, otherwise we find that $\abs{\omega(A)} = 0$, which contradicts the hypotheses. Note that $\sqrt{\omega(A^*A)} = \abs{\lambda}$, so that $A \cdot \omega$ is represented by the vector $\lambda \abs{\lambda}^{-1}\Omega_\omega$. Vectors that differ by a phase represent the same state, so we conclude that $A \cdot \omega = \omega$.
\end{proof}

\begin{proposition}\label{prop:linear_combo_action}
Let $\fA$ be a $C^*$-algebra and let $\omega \in \sP(\fA)$. Suppose $A, B \in \fA \setminus \fN_\omega$ and $A \cdot \omega = B \cdot \omega$. If $\alpha, \beta \in \bbC$ and $\alpha A + \beta B \notin \fN_\omega$, then $(\alpha A + \beta B) \cdot \omega = A \cdot \omega = B \cdot \omega$. 
\end{proposition}

\begin{proof}
Let $(\cH_\omega, \pi_\omega, \Omega_\omega)$ be the GNS representation of $\omega$. Since $A \cdot \omega = B \cdot \omega$ and $\omega$ is pure, we know $\pi_\omega(A)\Omega_\omega$ and $\pi_\omega(B)\Omega_\omega$ are linearly dependent, i.e., there exists $\lambda \in \bbC \setminus \qty{0}$ such that $\pi_\omega(B)\Omega_\omega = \lambda \pi_\omega(A)\Omega_\omega$. Then $(\alpha A + \beta B) \cdot \omega$ is represented by
\[
\frac{\pi_\omega(\alpha A + \beta B)\Omega_\omega}{\norm{\pi_\omega(\alpha A + \beta B)\Omega_\omega}} = \frac{(\alpha + \lambda \beta)\cdot \pi_\omega(A)\Omega_\omega}{\norm{\pi_\omega(\alpha A + \beta B)\Omega_\omega}} .
\]
This differs from $\pi_\omega(A)\Omega_\omega/\sqrt{\omega(A^*A)}$ by a phase, so $(\alpha A + \beta B) \cdot \omega = A \cdot \omega$.
\end{proof}

Next, let us examine the continuity properties of the action of $\fA$ on $\sS(\fA)$.

\begin{proposition}\label{prop:evaluation_jointly_continuous}
Let $\fA$ be a $C^*$-algebra and equip $\sS(\fA)$ with the weak* topology. The map
\begin{equation}\label{eq:evaluation_jointly_continuous}
\fA \times \sS(\fA) \to  \bbC, \quad (A, \omega) \mapsto \omega(A)
\end{equation}
is continuous.
\end{proposition}

\begin{proof}
Fix $A_0 \in \fA$, $\omega_0 \in \sS(\fA)$, and $\varepsilon > 0$. Given $A \in \fA$ such that $\norm{A - A_0} < \varepsilon/2$ and $\omega \in \sS(\fA)$ such that $\abs{\omega(A_0) - \omega_0(A_0)} < \varepsilon/2$, we have
\begin{align*}
\abs{\omega(A) - \omega_0(A_0)} &\leq \abs{\omega(A) - \omega(A_0)} + \abs{\omega(A_0) - \omega_0(A_0)}  \\
&\leq \norm{A - A_0} + \abs{\omega(A_0) - \omega_0(A_0)} < \varepsilon.
\end{align*}
Thus, the map \eqref{eq:evaluation_jointly_continuous} is continuous at the arbitrary point $(A_0, \omega_0)$.
\end{proof}

\begin{proposition}\label{prop:action_on_state_continuous}
Let $\fA$ be a $C^*$-algebra and equip $\sS(\fA)$ with the weak* topology. The map
\begin{equation}\label{eq:action_on_state}
\qty{(A, \omega) \in \fA \times \sS(\fA)\colon  A \notin \fN_\omega} \to  \sS(\fA), \quad (A, \omega) \to  A \cdot \omega
\end{equation}
is continuous.
\end{proposition}

\begin{proof}
Fix an arbitrary $B \in \fA$. Continuity of \eqref{eq:action_on_state} will follow from continuity of $(A, \omega) \mapsto (A \cdot \omega)(B)$ since $B$ was arbitrary. But since $A \mapsto A^*A$ and $A \mapsto A^*BA$ are continuous maps from $\fA$ to itself, it is easy to see from \eqref{eq:action_on_state_def} and \cref{prop:evaluation_jointly_continuous} that $(A, \omega) \mapsto (A \cdot \omega)(B)$ is continuous.
\end{proof}

\subsection{$\mathfrak{A}$-homotopies and the fundamental group}\label{subsec:A-homotopies}
A key idea in our computation of the fundamental group is to study homotopies that can be factored through a special kind of lift from the state space to the $C^*$-algebra. We define what we mean here.

\begin{definition}\label{defn:Ahomotopy}
  Let $\fA$ be a unital $C^*$-algebra, equip $\sS(\fA)$ with the weak* topology, and let $X$ be a topological space.
  We then say that a continuous map $\psi\colon X \to  \sS(\fA)$ is $\fA$-\emph{homotopic} to a continuous map
  $\omega\colon X \to  \sS(\fA)$ if there exists a continuous map  $A\colon X \times \interval \to  \fA$
  such that  
\begin{enumerate}
	\item\label{ite:notinideal} $A_{x,s} \notin \fN_{\psi_x}$ for all $(x,s) \in X \times I$,
	\item $A_{x,0} = \1$ for all $x \in X$,
	\item\label{ite:homotopy} $A_{x,1} \cdot \psi_x = \omega_x$ for all $x \in X$.
\end{enumerate}

In this case $(x,s) \mapsto A_{x,s} \cdot \psi_x$ is a homotopy from $\psi$ to $\omega$ in the standard sense.
We call the map $A\colon X \times \interval \to  \fA$ fulfilling conditions \eqref{ite:notinideal} to \eqref{ite:homotopy}
an $\fA$-\emph{homotopy} from $\psi$ to $\omega$. It can be understood as a lift of the homotopy from $\psi$ to $\omega$.
Let us write $\psi \htpy{\fA} \omega$ if $\psi$ is $\fA$-homotopic to $\omega$.

When $X$ is the unit interval $\interval$, then we say $\psi$ is $\fA$-\emph{homotopic} to $\omega$ \emph{relative endpoints}
if there exists an $\fA$-homotopy $A\colon \interval \times \interval \to  \fA$ from $\psi$ to $\omega$ such that
\begin{enumerate}
\setcounter{enumi}{3}
\item\label{ite:startpoint}
  $A_{0,s} \cdot \psi_0 = \omega_0$ for all $s\in \interval$, and 
\item\label{ite:endpoint}
   $A_{1,s} \cdot \psi_1 = \omega_1$ for all $s\in \interval$. 
\end{enumerate}
An $\fA$-homotopy fulfilling these properies  will be called an $\fA$-\emph{homotopy relative endpoints}. 
Let us write $\psi \pathhtpy{\fA} \omega$ if $\psi$ is $\fA$-homotopic to $\omega$ relative endpoints.
In case both $\psi$ and $\omega$ are loops we sometimes say that  $\psi$ is $\fA$-\emph{homotopic} to $\omega$
\emph{relative the basepoint} whenever  $\psi \pathhtpy{\fA} \omega$. 
\end{definition}

\begin{proposition}
The relations $\htpy{\fA}$ and $\pathhtpy{\fA}$ are transitive.
\end{proposition}

\begin{proof}
Suppose $\chi \htpy{\fA} \psi$ and $\psi \htpy{\fA} \omega$. Let $A\colon X \times \interval \to  \fA$ and $B\colon X \times \interval \to  \fA$ be respective $\fA$-homotopies. Define $C\colon X \times \interval \to  \fA$ by 
\[
  C_{x,s} = \begin{cases} A_{x,2s} &\text{for } s \in [0,1/2] \ , \\
  B_{x,2s-1}A_{x,1}  &\text{for } s \in [1/2,1] \ .\end{cases}
\]
This is a well-defined continuous map since $B_{x,0} = \1$ for all $x \in X$. By construction, $C$ then is 
an $\fA$-homotopy from $\chi$ to $\omega$. If $X = \interval$ and both $A$ and $B$ are $\fA$-homotopies relative
endpoints, then it is  easy to see that $C$ is also an $\fA$-homotopy relative endpoints. 
\end{proof}

We note that being $\fA$-homotopic is \textit{not} a symmetric relation. For example, in what follows, we will act on
non-pure states with projections to obtain pure states, but we cannot undo this operation by acting with an element of
$\fA$ because acting on a pure state with an element of $\fA$ always returns a pure state.

In our applications, $\fA$ will be unital (with unit denoted $\1$), we will have $X = \interval$, and the function
$A\colon \interval \times \interval \to  \fA$ will be of the form 
\[A_{t,s} = s\tilde A_t + (1 - s)\1\] for some continuous map
$\tilde A\colon  \interval \to  \fA$. Given a state $\psi$ and $A \in \fA \setminus \fN_\psi$, it is therefore important to examine when the linear interpolation $sA + (1 - s)\1$ falls outside the Gelfand ideal of $\psi$ for all $s$. We examine a few special cases of interest.

\begin{proposition}\label{prop:projection_linear_interp}
Let $\fA$ be a unital $C^*$-algebra and let $\omega \in \sS(\fA)$. If $P \in \fA$ is a projection and $\omega(P) > 0$, then $sP + (1 - s)\1 \notin \fN_\omega$ for all $s \in \interval$.
\end{proposition}

\begin{proof}
For ease of notation, let $P_s = sP + (1-s)\1$. Then
\[
\omega\qty(P_s^*P_s) = s^2 \omega(P) + (1-s)^2 + 2s(1-s)\omega(P) > 0. \qedhere
\]
\end{proof}

\begin{proposition}\label{prop:unitary_linear_interp}
Let $\fA$ be a unital $C^*$-algebra and let $\omega \in \sS(\fA)$. If $U \in \Unitary(\fA)$ and $s \in \interval$, then $sU + (1-s)\1 \in \fN_\omega$ if and only if $\omega(U) = -1$ and $s = 1/2$.
\end{proposition}

\begin{proof}
For ease of notation, let $U_s = sU + (1 - s)\1$. Observe that
\begin{align*}
\omega(U_s^*U_s) &= s^2 + (1-s)^2 +2s (1-s)\Re \omega(U)\\
&\geq s^2 + (1-s)^2 - 2s(1-s) = (1 - 2s)^2
\end{align*}
If $\omega(U) = -1$ and $s = 1/2$, then the inequality is an equality and we get $\omega(U_s^*U_s) = (1 - 2s)^2 = 0$. If $s \neq 1/2$, then we  have $\omega(U_s^*U_s) \geq (1 - 2s)^2 > 0$. If $s = 1/2$ and $\omega(U) \neq -1$, then the inequality is strict, so $\omega(U_s^*U_s) > (1-2s)^2 = 0$.
\end{proof}

Suppose we have a path $\omega\colon \interval \to  \sS(\fA)$ of states and a path of unitaries $U\colon\interval \to  \Unitary(\fA)$, and we want to perform a homotopy of $\omega$ via a linear interpolation $(t,s) \mapsto sU_t + (1-s)\1$. We are in trouble if there exists $t \in \interval$ such that $\omega_t(U_t) = -1$. As the next lemma implies, we can avoid this difficulty by multiplying the path $U$ by a continuous phase $\lambda\colon \interval \to  S^1$.

\begin{lemma}\label{lem:phases_avoid_-1}
Let $\gamma\colon \interval \to  D^2$ be a path in the closed unit disk $D^2 \subset \bbC$. There exists a continuous map $\lambda\colon\interval \to  S^1 \subset \bbC$ such that for all $t \in \interval$ either $\lambda(t)\gamma(t) = 1$ or $\abs{\gamma(t)} < 1$. 
\end{lemma}
Intuitively, the lemma can be thought of as follows. Imagine $\gamma$ is the path of an ant walking on a round table and you are seated at the table. Any $\lambda = e^{i\theta} \in S^1$ corresponds to a rotation of the table by an angle $\theta$. The lemma says that as the ant is walking along the table, you can continuously rotate the table so that whenever the ant gets to the edge of the table, the ant will be directly in front of you.

\begin{proof}[Proof of \cref{lem:phases_avoid_-1}.]
The proof more or less follows the standard path-lifting argument for fiber bundles, although we are not looking at a fiber bundle here. 

Begin by covering $D^2$ with the open balls $B_1(0)$ and $B_1(z) \cap D^2$ for all $z \in S^1$. For any point $w_0 \in B_1(0)$ and any $\mu_0 \in S^1$, the constant function $\mu \colon  B_1(0) \to  S^1$ with value $\mu_0$ satisfies $\mu(w_0) = \mu_0$ and $\abs{w} < 1$ for all $w \in B_1(0)$. 

Now consider any $z \in S^1$, any point $w_0 \in B_1(z) \cap D^2$, and any $\mu_0 \in S^1$ such that either $\mu_0 w_0 = 1$ or $\abs{w_0} < 1$. If $\mu_0 w_0 = 1$, then the map 
\[
\mu\colon B_1(z) \cap D^2 \to  S^1, \quad \mu(w) = \frac{\abs{w}}{w}
\]
satisfies $\mu(w_0) = \mu_0$ and for all $w \in B_1(z) \cap D^2$ either $\mu(w)w  = 1$ or $\abs{w} < 1$. On the other hand, suppose $\abs{w_0} < 1$ and let $\theta \in [0,2\pi)$ such that $e^{i\theta} = \mu_0w_0/\abs{w_0}$. Then the function
\[
\mu\colon B_1(z) \cap D^2 \to  S^1, \quad \mu(w) = \mu_0 \cdot \frac{w_0}{\abs{w_0}} \cdot \frac{\abs{w}}{w} \cdot e^{-i\theta\qty(\abs{w} - \abs{w_0})/\qty(1 - \abs{w_0})}
\]
satisfies $\mu(w_0) = \mu_0$ and for all $w \in B_1(z) \cap D^2$, either $\mu(w) w = 1$ or $\abs{w} < 1$.

Now consider the path $\gamma$. Cover $\interval$ by the preimages $\gamma^{-1}(B_1(0))$ and $\gamma^{-1}(B_1(z) \cap D^2)$ for all $z \in S^1$. By the Lebesgue number lemma, there exists $N \in \bbN$ such that for all $k \in \qty{0,\ldots, N-1}$, the interval $[k/N, (k+1)/N]$ is contained in one of these preimages. 

Assume that for some $k \in \qty{0,\ldots, N-1}$ we have constructed a path $\lambda\colon [0,k/N] \to  S^1$ such that for all $t \in [0,k/N]$ either $\lambda(t) \gamma(t) = 1$ or $\abs{\gamma(t)} <1$. Choose one of the preimages above that contains $[k/N,(k+1)/N]$ and denote it $\gamma^{-1}(O)$. Setting $w_0 = \gamma(k/N)$ and $\mu_0 = \lambda(k/N)$, we may compose the restriction $\gamma \colon [k/N, (k+1)/N] \to  O$ with the appropriate map $\mu\colon O \to  S^1$ defined previously to obtain a map $\tilde \lambda\colon [k/N, (k+1)/N] \to  S^1$ such that $\tilde \lambda(k/N) = \mu_0 = \lambda(k/N)$ and for all $t \in [k/N, (k+1)/N]$ either $\tilde\lambda(t)\gamma(t) = 1$ or $\abs{\gamma(t)} < 1$. We may now glue $\lambda$ to $\tilde \lambda$ to get a continuous map $\lambda\colon [0,(k+1)/N] \to  S^1$ such that for all $t \in [0,(k+1)/N]$ either $\lambda(t)\gamma(t) = 1$ or $\abs{\gamma(t)} < 1$. Continuing in this manner constructs the desired path $\lambda \colon \interval \to  S^1$.
\end{proof}

\begin{cor}
Let $\fA$ be a unital $C^*$-algebra and let $\omega \colon \interval \to  \sS(\fA)$ be weak* continuous. Given a continuous map $U\colon \interval \to  \Unitary(\fA)$,  there exists a continuous $\lambda\colon \interval \to  S^1$ such that $s \lambda_tU_t + (1-s)\1 \notin \fN_{\omega_t}$ for all $t,s \in \interval$.
\end{cor}

\begin{proof}
This is immediate from \cref{prop:unitary_linear_interp} and \cref{lem:phases_avoid_-1} upon setting $\gamma(t) = \omega_t(U_t)$.
\end{proof}

\subsection{Finite dimensional density matrices and rectification of paths}\label{subsec:fddensitymatrices}
Let us now fix $n \in \bbN$ and consider the specific example
$\fA = M_n(\bbC)$. In several steps we will show in this section that
every loop in the state space $\sS( M_n(\bbC))$ is $M_n(\bbC)$-homotopic
relative the basepoint to the constant loop. In particular this implies that
the state space of a matrix algebra $M_n(\bbC)$ is simply connected. 
We start by identifying $\sS (M_n(\bbC))$ with the set
$\denmat{n}$ of $n \times n$ density matrices, topologized as a subspace of
$M_n(\bbC)$. Recall that if $\varrho \in \denmat{n}$, then
$A \mapsto \tr(\varrho A)$ is a state on $M_n(\bbC)$.

\begin{proposition}
The map $\denmat{n} \to  \sS(M_n(\bbC))$ that associates to a density matrix the corresponding state is a homeomorphism.
\end{proposition}

\begin{proof}
Suppose $\varrho_1, \varrho_2 \in \denmat{n}$ map to the same state. Then 
\[
\tr((\varrho_1 - \varrho_2)A) = 0
\] 
for all $A \in M_n(\bbC)$. Since $\varrho_1 - \varrho_2$ is self-adjoint, we can find an orthonormal basis of eigenvectors. Letting $A$ be the projection onto any eigenspace, the equation $\tr((\varrho_1 - \varrho_2)A) = 0$ implies that the corresponding eigenvalue is zero. Therefore $\varrho_1 - \varrho_2 = 0$. 

We now show surjectivity. Since $M_n(\bbC)$ has only one superselection sector, we know $\sP(M_n(\bbC)) \cong \bbP(\bbC^n)$, so $\sP(M_n(\bbC))$ is compact. We know $\sS(M_n(\bbC))$ is the closed convex hull of $\sP(M_n(\bbC))$, but since $\sP(M_n(\bbC))$ is a compact subset of a finite-dimensional vector space, the convex hull of $\sP(M_n(\bbC))$ is already closed, so $\sS(M_n(\bbC))$ is just the convex hull of $\sP(M_n(\bbC))$. Any pure state can be written as $A \mapsto \tr(PA)$ for some rank-one projection $P$, and any convex combination of rank-one projections is a density matrix. Thus, any state can be represented by a density matrix.

Since $\denmat{n}$ and $\sS(M_n(\bbC))$ are compact Hausdorff, all that remains to do is show continuity. Letting $\varrho_1, \varrho_2 \in \denmat{n}$ and letting $e_1,\ldots, e_n$ be an orthonormal basis of $\bbC^n$, we have
\[
\abs{\tr((\varrho_1 - \varrho_2)A)} \leq \sum_{i=1}^n \abs{\ev{e_i,(\varrho_1 - \varrho_2)A e_i}} \leq n \norm{\varrho_1 - \varrho_2}\norm{A}.
\]
Continuity follows.
\end{proof}

Note that if $\omega \in \sS(M_n(\bbC))$ is represented by the density matrix $\varrho$ and $A \in M_n(\bbC) \setminus \fN_\omega$, then $A \cdot \omega$ is represented by the density matrix $A\varrho A^*/\tr(A\varrho A^*)$. We will use this correspondence between states and density matrices freely.

Our first goal is to show that any loop in $\sS(M_n(\bbC))$ based at a pure state is $M_n(\bbC)$-homotopic to the constant map
relative the basepoint (in the sense of \cref{defn:Ahomotopy}). Since all pure states on $M_n(\bbC)$ are unitarily equivalent, it suffices to take our loop to be based at the state represented by the first standard basis vector. Denote this state by $\omega_0^n$. Explicitly, $\omega_0^n$ is the state
\[
  \omega_0^n \colon M_n(\bbC) \to \C,\quad A \mapsto \langle e_0 ,Ae_0\rangle ,
\]
where $e_0 =(1,0,\ldots,0) \in \C^n$. 
The following notation will also be convenient. Given $n \in \bbN$ and $k \in \qty{0,\ldots, n-1}$, let 
\[
P^n_k = \diag(1,\ldots, 1, 0, \ldots, 0) \in M_n(\bbC)
\]
where there are $k$ zeros on the diagonal.

\begin{lemma}\label{lem:omega(P^n_1)=0}
  If $n > 1$ and $\omega \in \sS(M_n(\bbC))$, then $\omega(P^n_1) = 0$ if and only if
  $\omega$ is represented by the last standard basis vector, that is, if and only if
  $\omega (A) =\langle e_n,A e_n \rangle$ for all $A\in M_n(\bbC)$, where
   $e_n =(0,\ldots,0,1) \in \C^n$.
\end{lemma}

\begin{proof}
If $\omega$ is represented by the last standard basis vector, then it is clear that $\omega(P^n_1) = 0$. Suppose $\omega(P^n_1) = 0$. Let $\varrho$ be the density matrix representing $\omega$. Then $\tr(P^n_1 \varrho P^n_1) = 0$, but $P^n_1 \varrho P^n_1$ is positive, so this implies that $P^n_1 \varrho P^n_1 = 0$. Note that $P^n_1 \varrho P^n_1$ is the matrix with the same entries as $\varrho$ in the upper left $(n-1) \times (n-1)$ submatrix, and with zeros everywhere else. Since this $(n-1) \times (n-1)$ submatrix is zero, it follows from the fact that $\tr(\varrho) = 1$ that the bottom right entry of $\varrho$ is $1$. Acting $\varrho$ on the last standard basis vector and using the fact that $\norm{\varrho} \leq 1$, we see that all other entries in the last column of $\varrho$ are zero. Since $\varrho$ is self-adjoint, all other entries in the last row of $\varrho$ are zero as well. Thus, the only nonzero entry of $\varrho$ is a one in the bottom right corner, and this is the density matrix corresponding to the last standard basis vector.
\end{proof}

Our next two lemmas follow the same train of thought as  \cref{lem:phases_avoid_-1}. In  \cref{lem:phases_avoid_-1}, the situation was analogous to an ant walking on a round table $D^2$ while an observer rotated the table by multiplying with elements of $S^1$. In \cref{lem:state_neighborhood} and \cref{lem:unitary_avoid_0}, our aim is to make a similar statement for paths in $\sS(M_n(\bbC))$, where we act on states with unitaries. 

For example, when $n = 2$ the state space $\sS(M_n(\bbC))$ is the Bloch ball and we can think of our unitaries acting on the Bloch ball by rotations. Analogously, if an ant is eating through an apple, I can rotate the apple continuously so that whenever the ant emerges, it is facing straight up.

\begin{lemma}\label{lem:state_neighborhood}
For every $\psi \in \sS(M_n(\bbC))$, there exists an open neighborhood 
\[\psi \in O \subset \sS(M_n(\bbC))\] 
with the following property: 

For every $\psi_0 \in O$ and $U_0 \in \Unitary(n)$ such that either $U_0 \cdot \psi_0 = \omega_0^n$ or $U_0 \cdot \psi_0$ is not pure, there exists a continuous map $U\colon O \to  \Unitary(n)$ such that
\begin{enumerate}[(a)]
\item $U_{\psi_0} = U_0$, and 
\item for every $\phi \in O$, either $U_\phi \cdot \phi = \omega_0^n$ or $U_\phi \cdot \phi$ is not pure.
\end{enumerate}
\end{lemma}

\begin{proof}
Suppose $\psi$ is not pure. Then we may take 
\[O = \sS(M_n(\bbC)) \setminus \sP(M_n(\bbC)).\] 
Indeed, observe that for any $\psi_0 \in O$ and $U_0 \in \Unitary(n)$, the constant map at $U_0$ satisfies the desired properties, since a unitary cannot take a non-pure state to a pure state.

Suppose $\psi$ is pure, let $\Psi$ be a representing vector for $\psi$, and let $P = \ketbra{\Psi}$. Let 
\[
O = \qty{\phi \in \sS(M_n(\bbC)): \phi(P) > 7/8}.
\]
Suppose $\phi \in O$ and let $\varrho_\phi \in \denmat{n}$ be the density matrix representing $\phi$. Write $\varrho_\phi$ using the spectral decomposition 
\[\varrho_\phi = \sum_{i=1}^n t_i \ketbra{\Phi_i},\] 
where $\Phi_1,\ldots, \Phi_n$ is an orthonormal basis of eigenvectors of $\varrho_\phi$. Observe that
\begin{align*}
\frac{7}{8} < \phi(P) = \tr(\varrho_\phi P) = \sum_{i=1}^n t_i \abs{\ev{\Phi_i, \Psi}}^2 \leq \qty(\max_i t_i) \sum_{i=1}^n \abs{\ev{\Phi_i, \Psi}}^2 = \max_i t_i.
\end{align*}
Let $k \in \qty{1,\ldots, n}$ be such that $t_k = \max_i t_i$. We have just shown that $t_k > 7/8$. Since $\sum_{i=1}^n t_i = 1$, we know that $t_k \leq 1$ and $t_i < 1/8$ for all $i \neq k$. 

Next, observe that
\begin{align*}
\frac{7}{8} < \phi(P) = t_k \abs{\ev{\Phi_k, \Psi}}^2 + \sum_{i \neq k} t_i \abs{\ev{\Phi_i, \Psi}}^2 \leq \abs{\ev{\Phi_k, \Psi}}^2 + \frac{1}{8}.
\end{align*}
Thus, $\abs{\ev{\Phi_k, \Psi}}^2 > 3/4$. 

Consider the map $O \ni \phi \mapsto \omega_\phi \in \sP(M_n(\bbC))$, where $\omega_\phi$ is represented by an eigenvector of $\varrho_\phi$ with maximum eigenvalue. This is a continuous map that leaves $O \cap \sP(M_n(\bbC))$ invariant. With $\Phi_k$ representing a maximum eigenvector of $\phi$ as above, we observe that
\[
\norm{\omega_\phi - \psi}^2 = 4 - 4\abs{\ev{\Phi_k, \Psi}}^2 < 4 - 4 \cdot \frac{3}{4} = 1.
\]
Thus,
\[
\norm{\omega_\phi - \omega_{\psi_0}} \leq \norm{\omega_\phi - \psi} + \norm{\psi - \omega_{\psi_0}} < 2. 
\]

There exists a continuous map $U\colon  B_2(\omega_{\psi_0}) \cap \sP(M_n(\bbC)) \to  \Unitary(n)$ such that $U_{\omega_{\psi_0}} = I$ and $U_\omega \cdot \omega = \omega_{\psi_0}$ for every $\omega \in B_2(\omega_{\psi_0}) \cap \sP(M_n(\bbC))$. If $\psi_0$ is pure, then $\psi_0 = \omega_{\psi_0}$, and the function
\[
\phi \mapsto U_0 U_{\omega_{\phi}}
\]
is the desired family of unitaries.

If $\psi_0$ is not pure, then choose a unitary $V$ such that $V \cdot \omega_{\psi_0} = \omega_0^n$. Write $U_0 = e^{iA}$ and $V = e^{iB}$ where $A,B \in M_n(\bbC)$ are self-adjoint. Then the function
\[
\phi \mapsto e^{iB\qty[S(\varrho_{\psi_0})-S(\varrho_\phi)]/S(\varrho_{\psi_0})} e^{iAS(\varrho_{\phi})/S(\varrho_{\psi_0})} U_{\omega_\phi}
\]
has the desired properties, where $S(\varrho)$ is the von Neumann entropy of $\varrho$.
\end{proof}

\begin{lemma}\label{lem:unitary_avoid_0}
Let $\omega\colon \interval \to  \sS(M_n(\bbC))$ be a loop based at $\omega_0^n$. There exists a continuous family of unitaries $U\colon  \interval \to  \Unitary(n)$ such that 
\[U_0 \cdot \omega_0^n = U_1 \cdot \omega_0^n = \omega_0^n\] 
and $(U_t \cdot \omega_t)(P^n_{1}) > 0$ for all $t \in \interval$.
\end{lemma}

\begin{proof}
For each $\psi \in \sS(M_n(\bbC))$, let $O_\psi$ be an open neighborhood of $\psi$ as provided by \cref{lem:state_neighborhood}. The neighborhoods $O_\psi$ cover $\sS(M_n(\bbC))$, hence the neighborhoods $\omega^{-1}(O_\psi)$ cover $\interval$. Thus, there exists $N \in \bbN$ such that for each $k \in \qty{0,\ldots, N-1}$, the set $[k/N, (k+1)/N]$ lies in some $\omega^{-1}(O_\psi)$. 

We now follow the standard path-lifting argument. Assume that for some $k \in \qty{0,\ldots, N-1}$ we have defined a continuous map $U\colon [0,k/N] \to  \Unitary(n)$ such that $U_0 = I$ and for all $t \in [0,k/N]$ either $U_t \cdot \omega_t = \omega_0^n$ or $U_t \cdot \omega_t$ is not pure. Choose $\psi \in \sS(M_n(\bbC))$ such that $[k/N, (k+1)/N] \subset \omega^{-1}(O_\psi)$. Composing $\omega|_{[k/N, (k+1)/N]}$ with such a continuous family of unitaries as provided by  \cref{lem:state_neighborhood}, we obtain a continuous map $V\colon [k/N, (k+1)/N] \to  \Unitary(n)$ such that $V_{k/N} = U_{k/N}$ and for every $t \in [k/N, (k+1)/N]$, either $V_t \cdot \omega_t = \omega_0^n$ or $V_t \cdot \omega_t$ is not pure. We may glue the maps $U$ and $V$ together to obtain a continuous map $W\colon [0,(k+1)/N] \to  \Unitary(n)$ such that $W_0 = I$ and for all $t$, either $W_t \cdot \omega_t = \omega_0^n$ or $W_t \cdot \omega_t$ is not pure. 

Proceeding in this fashion, we define a continuous map $U\colon \interval \to  \Unitary(n)$ such that $U_0 \cdot \omega_0^n = U_1 \cdot \omega_0^n = \omega_0^n$ and for all $t$ either $U_t \cdot \omega_t = \omega_0^n$ or $U_t \cdot \omega_t$ is not pure. The fact that $(U_t \cdot \omega_t)(P^n_1) > 0$ for all $t$ follows from \cref{lem:omega(P^n_1)=0}.
\end{proof}

\begin{lemma}\label{lem:P1}
  If $\omega\colon \interval \to  \sS(M_n(\bbC))$ is a loop based at $\omega_0^n$, then $\omega$ is $M_n(\bbC)$-homotopic
  relative the basepoint to a loop $\psi\colon \interval \to  \sS(M_n(\bbC))$ such that $\psi_t(P^n_1) = 1$
  for all $t \in \interval$.
\end{lemma}

\begin{proof}
Using \cref{lem:unitary_avoid_0} we find a continuous family of unitaries $U\colon \interval \to  \Unitary(n)$ such that $U_0 \cdot \omega_0^n = U_1 \cdot \omega_0^n = \omega_0^n$ and $(U_t \cdot \omega_t)(P_1^n) > 0$ for all $t \in \interval$. Using \cref{lem:phases_avoid_-1} we find a continuous path $\lambda\colon \interval \to  S^1$ such that $\lambda_t \omega_t(U_t) \neq -1$ for all $t \in \interval$. Then $(s,t) \mapsto s\lambda_t U_t + (1 - s)\1$ is continuous, equals $\1$ when $s = 0$, and satisfies $s\lambda_t U_t + (1 - s)\1 \notin \fN_{\omega_t}$ for all $t, s \in \interval$ by \cref{prop:unitary_linear_interp}. Furthermore, for all $s \in \interval$,  $s \lambda_0 U_0 + (1 - s)\1$ and $s\lambda_1 U_1 + (1- s)\1$ leave $\omega_0^n$ invariant by \cref{prop:linear_combo_action}. Therefore, $\omega$ is homotopic relative the base point to $\chi_t = \lambda_t U_t \cdot \omega_t = U_t \cdot \omega_t$.

Note that $\chi_t(P_1^n) > 0$ for all $t \in \interval$ and $\chi_0(P^n_1) = \chi_1(P^n_1) = 1$. Then by \cref{prop:projection_linear_interp} we have the homotopy $(t,s) \mapsto \qty[sP_1^n + (1-s)\1] \cdot \chi_t$. By  \cref{prop:invariant_state} and  \cref{prop:linear_combo_action}, this is a homotopy relative the basepoint. Finally, we note that $(P_1^n \cdot \chi_t)(P_1^n) = 1$ for all $t \in \interval$.
\end{proof}

We would now like to prove the result of \cref{lem:P1} but with $P^n_1$ replaced by $P^n_k$ for arbitrary $k \in \qty{1,\ldots, n-1}$. We note that $M_{n-k}(\bbC) \cong P^n_k M_n(\bbC) P^n_k$, where the $*$-isomorphism maps $A \in M_{n-k}(\bbC)$ to the $n \times n$ matrix with $A$ as the upper left $(n-k) \times (n-k)$ submatrix, and zeros everywhere else. If we have a state $\psi \in \sS(M_n(\bbC))$ such that $\psi(P^n_k) = 1$, then $\psi$ is essentially a state on $M_{n-k}(\bbC) \cong P^n_k M_n(\bbC) P^n_k$. 

At this point it is worth taking a step back and considering this last statement through a more abstract lens. Note that if $\fA$ is a $C^*$-algebra and $P \in \fA$ is a projection, then $P \fA P$ is a unital $C^*$-algebra with unit $P$. In fact, it is a hereditary $C^*$-subalgebra.  \cref{prop:invariant_state} has the following consequence for states on $P \fA P$.

\begin{proposition}
Let $\fA$ be a $C^*$-algebra and let $P \in \fA$ be a projection.
\begin{enumerate}
	\item If $\psi \in \sS(\fA)$, then $\psi|_{P \fA P} \in \sS(P \fA P)$ if and only if $\psi(P) = 1$.
	\item Given a state $\omega \in \sS(P\fA P)$ the function $\omega \circ \Ad(P)$ is the unique extension of $\omega$ to a state on $\fA$. 
	\item The map 
	\[
	\sS(P\fA P) \ni \omega \mapsto \omega \circ \Ad(P) \in \sS(\fA)
	\]
	is a closed embedding with respect to the weak* topologies and has image $\qty{\psi \in \sS(\fA): \psi(P) = 1}$.
\end{enumerate}
\end{proposition}

\begin{proof}
(1) This is immediate from the fact that $\psi|_{P\fA P}$ is a positive linear functional and $P$ is the unit of $P \fA P$.

(2) Since $P\fA P$ is a hereditary $C^*$-subalgebra of $\fA$ with unit $P$, this is provided by \cite[Thm.~3.3.9]{MurphyCAOT}.

(3) Using the characteristic mapping property of the weak* topology, we see that $\omega \mapsto \omega \circ \Ad(P)$ is weak* continuous. It is clear that if $\omega_1, \omega_2 \in \sS(P\fA P)$ and $\omega_1 \circ \Ad(P) = \omega_2 \circ \Ad(P)$, then $\omega_1 = \omega_2$. Since $\sS(P\fA P)$ and $\sS(\fA)$ are compact Hausdorff, this implies that $\omega \mapsto \omega \circ \Ad(P)$ is an embedding. We observe that $[\omega \circ \Ad(P)](P) = \omega(P) = 1$ for all $\omega \in \sS(P\fA P)$. On the other hand, if $\psi(P) = 1$, then \cref{prop:invariant_state} implies that $P \cdot \psi = \psi$, so $\psi = \psi|_{P\fA P} \circ \Ad(P)$ and $\psi|_{P\fA P} \in \sS(P\fA P)$. Thus, the image of the map $\omega \mapsto \omega \circ \Ad(P)$ is the set $\qty{\psi \in \sS(\fA): \psi(P) = 1}$, and this set is manifestly weak* closed.
\end{proof}

\begin{lemma}\label{lem:path_homotopy_push_forward}
  Let $\fA$ and $\fB$ be unital $C^*$-algebras, let $P \in \fB$ be a projection, let $\pi\colon \fA \to  P\fB P$ be a $*$-isomorphism, and let $\psi\colon \interval \to  \sS(\fB)$ be a weak* continuous path such that $\psi_t(P) = 1$ for all $t \in \interval$. If $A\colon \interval \times \interval \to  \fA$ is an $\fA$-homotopy relative endpoints from
  the path $t \mapsto \psi_t|_{P\fB P} \circ \pi$ to $\omega \colon  \interval \to  \sS(\fA)$,
  then $(t,s) \mapsto (\1 - P) + \pi(A_{t,s})$ is an $\fA$-homotopy relative endpoints 
  from $\psi$ to the path $t \mapsto \omega_t \circ \pi^{-1} \circ \Ad(P)$.
\end{lemma}

\begin{proof}
Note that $\pi(\1) = P$ since $\pi$ is a $*$-isomorphism. We check that $(\1 - P) + \pi(A_{t,s})$ has the desired properties. If $\pi(A_{t,s}) + \1 - P \in \fN_{\psi_t}$, then
\[
0 = \psi_t\qty[\qty(\qty(\1 - P) + \pi(A_{t,s}^*))\qty(\qty(\1 - P) + \pi(A_{t,s}))] = \psi_t(\pi(A_{t,s}^*A_{t,s})),
\] 
which implies $A_{t,s} \in \fN_{\psi_t|_{P\fB P} \circ \pi}$, which contradicts the definition of $A$. Thus, $(\1 - P) + \pi(A_{t,s}) \notin \fN_{\psi_t}$ for all $t$ and $s$.

Next, observe that for all $t$, we have 
\[
(\1 - P) + \pi(A_{t,0}) = (\1 - P) + \pi(\1) = \1 - P + P = \1.
\]

For arbitrary $t$ and $s$, we have
\begin{align*}
\qty[\qty[(\1 - P) + \pi(A_{t,s})] \cdot \psi_t](B) &= \frac{\psi_t(\pi(A_{t,s}^*)B\pi\qty(A_{t,s}))}{\psi_t(\pi(A_{t,s}^*A_{t,s}))} \\
&= \frac{\psi_t(\pi(A_{t,s}^* \pi^{-1}(PBP) A_{t,s}))}{\psi_t(\pi(A_{t,s}^*A_{t,s}))}\\
&= \qty[A_{t,s} \cdot \qty(\psi_t|_{P\fB P} \circ \pi)](\pi^{-1}(PBP))
\end{align*}
When $s = 1$, this yields
\[
\qty[\qty[(\1 - P) + \pi(A_{t,s})] \cdot \psi_t](B) = \omega_t(\pi^{-1}(PBP)),
\]
so $\qty[(\1 - P) + \pi(A_{t,s})] \cdot \psi_t = \omega_t \circ \pi^{-1} \circ \Ad(P)$. When $t = 0$ or $t = 1$, this yields
\[
\qty[\qty[(\1 - P) + \pi(A_{t,s})] \cdot \psi_t](B) = (\psi_t|_{P\fB P} \circ \pi)(\pi^{-1}(PBP)) =  \psi_t(B),
\]
so $\qty[(\1 - P) + \pi(A_{t,s})] \cdot \psi_t = \psi_t$ for all $s$ when $t = 0$ and $t = 1$. 
\end{proof}

\begin{theorem}\label{thm:M_n(C)_nulhomotopic}
  Fix $k \in \qty{1,\ldots, n-1}$. If $\omega\colon \interval \to  \sS(M_n(\bbC))$ is a loop based at $\omega_0^n$, then $\omega$ is $M_n(\bbC)$-homotopic relative the basepoint to a loop $\psi\colon \interval \to  \sS(M_n(\bbC))$ which satisfies
  $\psi_t(P^n_k) = 1$ for all $t \in \interval$.
\end{theorem}

\begin{proof}
  We prove the theorem by induction on $k$. The base case is given by \cref{lem:P1}. Suppose the theorem is true for some $k < n-1$. Then $\omega$ is $M_n(\bbC)$-homotopic relative the basepoint to a loop $\psi:\interval \to  \sS(M_n(\bbC))$
  based at $\omega^n_0$ and which fulfills $\psi_t(P^n_{k}) = 1$ for all $t \in I$. 

Now consider the $*$-isomorphism $\pi\colon M_{n-k}(\bbC) \to  P^n_k M_n(\bbC) P^n_k$ that maps $A \in M_{n-k}(\bbC)$ to the matrix
\[
\pi(A) = \mqty(A & 0\\0 & 0).
\]
We have a loop $t \mapsto \psi_t|_{P^n_k M_n(\bbC) P^n_k} \circ \pi$ based at $\omega^{n-k}_0$. By \cref{lem:P1}, this loop is
$M_{n-k}(\bbC)$-homotopic relative the basepoint to a loop $\phi\colon \interval \to  \sS(M_{n-k}(\bbC))$ such that $\phi_t(P^{n-k}_1) = 1$ for all $t \in \interval$. By \cref{lem:path_homotopy_push_forward}, $\psi$ is $M_n(\bbC)$-homotopic
relative the basepoint to $t \mapsto \phi_t \circ \pi^{-1} \circ \Ad(P^n_k)$. We observe that
\begin{align*}
(\phi_t \circ \pi^{-1} \circ \Ad(P^n_k))(P^n_{k+1}) &= (\phi_t \circ \pi^{-1} \circ \Ad(P^n_k))(\pi(P^{n-k}_1))\\
&= \phi_t(P^{n-k}_1) = 1
\end{align*}
for all $t \in \interval$. The result now follows by induction.
\end{proof}

We note that given $\omega \in \sS(M_n(\bbC))$, the statement $\omega(P^n_{n-1}) = 1$ implies that $\omega = \omega^n_0$. Thus, setting $k = n-1$ in \cref{thm:M_n(C)_nulhomotopic}, the loop $\psi$ is the constant map at $\omega^n_0$.

\subsection{Triviality of the fundamental group for UHF algebras}

We finally come to our goal to show that $\sP(\fA)$ is simply connected for a  UHF algebra $\fA$.
So let us now consider the state space of a UHF algebra $\fA$. 

It is well-known (see \cite[\S12.1]{KadisonRingroseII})  that every UHF algebra is $*$-isomorphic to the directed colimit of a directed system of tensor products 
\[
\fA_n = \bigotimes_{i=1}^n M_{k_i}(\bbC)
\]
where we have assigned a dimension $k_i \in \bbN \setminus \qty{1}$ for each $i \in \bbN$. 
For $m < n$, the embedding $\fA_m \hookrightarrow \fA_n$ of the directed system is defined by tensoring on identity matrices in the factors $M_{k_i}(\bbC)$ for $i > m$. We may therefore assume $\fA$ is of this form.

 It is known that $\sP(\fA)$ is path-connected in the weak$^*$ topology \cite{EilCCCA,ContinuousKadison}, so it suffices to show that $\pi_1(\sP(\fA), \omega_0) = 0$ for some choice of $\omega_0$. For convenience, we will choose the basepoint $\omega_0$ to be the product state represented by the first standard basis vector in each tensor factor. 

The $C^*$-algebras $M_{k_i}(\bbC)$ embed into any $\fA_n$ for $n > i$, and then embed into $\fA$. Let $\fB_i$ denote this copy of $M_{k_i}(\bbC)$ inside $\fA$. Let $P_i \in \fB_i$ be the copy of the projection $P^{k_i}_{k_i - 1}$ inside $\fB_i$. The following lemma will be useful.

\begin{lemma}[{\cite[Lem.~3.4.2]{BrownOzawa}}]\label{lem:restricted_state_factorized}
Let $\fA$ be a $C^*$-algebra and let $\fB$ be a $C^*$-subalgebra of $\fA$. If $\omega \in \sS(\fA)$ and $\omega|_{\fB} \in \sP(\fB)$, then
\[
\omega(AB) = \omega(A)\omega(B)
\]
for all $A \in \fB'$ and $B \in \fB$, where $\fB'$ is the commutant of $\fB$.
\end{lemma}

\begin{theorem}\label{thm:fundaUHFfinal}
  Let $\fA$ be a UHF algebra and let $\omega\colon \interval \to  \sS(\fA)$ be a weak* continuous loop based at the product state $\omega_0$. There exists a continuous map $A\colon \interval \times [0,1) \to  \fA$ and a homotopy $H\colon \interval \times \interval \to  \sS(\fA)$ from $\omega$ to the constant loop at $\omega_0$ such that $H(t,s) = A_{t,s} \cdot \omega_t$ for all $s < 1$. It follows that $\pi_1(\sP(\fA),\omega_0) = 0$.
\end{theorem}

The strategy is as follows. Beginning with the arbitrary loop $\omega$, we first focus on its restriction to the first ``lattice site'' $M_{k_1}(\bbC)$. Acting with unitaries and projections localized on ``site one'' (and their linear interpolations with the identity), we deform $\omega$ so that its restriction to site one becomes constant and pure, as in \cref{thm:M_n(C)_nulhomotopic}. \cref{lem:restricted_state_factorized} implies that site one is now disentangled from the rest of the sites. We then proceed to do the same thing to site two, site three, and so on. We can fit this countable infinity of homotopies in the half-open interval $[0,1)$, and then define the homotopy to be constant $\omega_0$ when the homotopy parameter reaches $s = 1$. This turns out to be weak* continuous because for a fixed local operator $A$, the expectation values of $A$ become independent of the homotopy parameter after the support of $A$ has become disentangled from the rest of the sites.

\begin{proof}[Proof of \cref{thm:fundaUHFfinal}.]
  We will follow the notation in the discussion preceding the theorem statement.
  Moreover, for any $C^*$-subalgebra $\fB \subset \fA$ we will say that an $\fA$-homotopy
  $B\colon  \interval\times\interval \to\fA$ is \emph{supported} in $\fB$ if and only if it has
  image in $\fB$. As the last preparation we set $\psi^0 = \omega$ for convenience.
  Now we  construct recursively  for each positive integer $i$ a loop $\psi^i\colon \interval \to\sS(\fA)$
  whose restriction to $\fB_i$ is a constant loop and an $\fA$-homotopy $B^i$ supported in  $\fB_i$.
    
  The restriction $t \mapsto \omega_t|_{\fB_1}$ of $\omega$ to the first tensor factor is a loop of states based at $\omega^{k_1}_0$, so it follows from \cref{thm:M_n(C)_nulhomotopic} that $\omega|_{\fB_1}$ is $\fB_1$-homotopic relative the basepoint to the constant loop at $\omega^{k_1}_0$. 
  Thus, there exists an $\fA$-homotopy $B^1\colon \interval \times \interval \to  \fA$ 
  relative the basepoint and supported on $\fB_1$ from $\psi^0 =\omega $ to a loop $\psi^1$
  such that $\psi^1_t(P_1) = 1$ for all $t \in \interval$. 

  Suppose that for some $n \in \bbN$ and all $i \in \qty{1,\ldots, n}$ we have constructed loops
  $\psi^i\colon \interval \to  \sS(\fA)$ based at $\omega_0$ and continuous maps
  $B^i\colon \interval \times \interval \to  \fB_i \subset \fA$  such that 
  $B^i$ is   an $\fA$-homotopy relative basepoints from $\psi^{i-1}$ to $\psi^i$ and such that
  $\psi^i_t(P_j) = 1$ for all $j \leq i$ and $t \in \interval$.
  Restricting $\psi^n$ to $\fB_{n+1}$ and applying \cref{thm:M_n(C)_nulhomotopic}, we obtain
  an $\fA$-homotopy  $B^{n+1}\colon \interval \times \interval \to  \fB_{n+1} \subset \fA$ relative the basepoint
  from $\psi^n$ to a loop $\psi^{n+1}$ satisfying $\psi^{n+1}_t(P_{n+1}) = 1$ for all $t \in \interval$.
  Note that the $\fA$-homotopy $B^{n+1}$ is supported on $\fB_{n+1}$.
  We observe further that for $j < n+1$, the restricted state $\psi^n_t|_{\fB_j}$ is pure (and independent of $t$)
  since $\psi^n_t(P_j) = 1$, so \cref{lem:restricted_state_factorized} implies
\begin{align*}
\psi_t^{n+1}(P_j) &= (B_{t,1}^{n+1} \cdot \psi_t^n)(P_j) \\
&= \frac{\psi_t^n((B_{t,1}^{n+1})^*P_j B_{t,1}^{n+1})}{\psi_t^n((B_{t,1}^{n+1})^*B_{t,1}^{n+1})}\\
&= \frac{\psi_t^n((B_{t,1}^{n+1})^*B_{t,1}^{n+1}) \cdot \psi_t^n(P_j)}{\psi_t^n((B_{t,1}^{n+1})^*B_{t,1}^{n+1})} = 1.
\end{align*}

By recursion we therefore obtain for all positive integers $i$ a loop $\psi^i\colon \interval \to  \sS(\fA)$ based at $\omega_0$
and an $\fA$-homotopy $B^i\colon \interval \times \interval \to  \fB_i \subset \fA$
from $\psi^{i-1}$ to $\psi^i$ such that $\psi^i_t(P_j) = 1$ for all $j \leq i$ and $t \in \interval$. 

Now define $A\colon \interval \times [0,1) \to  \fA$ as follows. Given $i \in \bbN$, define $A$ on $\interval \times [1 - 2^{-i+1}, 1-2^{-i}]$ as
\[
A_{t,s} = B_{t,r}^i B_{t,1}^{i-1} B_{t,1}^{i-2} \cdots B_{t,1}^1, \quad \tn{ where $r = \frac{s - 1 + 2^{-i+1}}{2^{-i+1} - 2^{-i}}$}.
\]
Since $B^i_{t,0} = \1$ for all $t$ and $i$, we see that 
\begin{align*}
A_{t,1-2^{-i+1}} = B_{t,1}^{i-1}B_{t,1}^{i-2}\cdots B_{t,1}^1.
\end{align*}
Note also that
\[
A_{t,1-2^{-i}} = B_{t,1}^iB_{t,1}^{i-1} \cdots B_{t,1}^1.
\]
Thus, the functions $A\colon \interval \times [1-2^{-i+1}, 1-2^{-i}] \to  \fA$ glue together to a continuous function $A\colon I \times [0,1) \to  \fA$. Note that $A_{t,0} = \1$ for all $t \in \interval$.

Finally, we define $H\colon \interval \times \interval \to  \sS(\fA)$ as
\[
H(t,s) = \begin{cases} A_{t,s} \cdot \omega_t &: s < 1 \\ \omega_0 &: s = 1 \end{cases}.
\]
Since $A_{t,0} = \1$ for all $t$, we see that $H(t, 0) = \omega_t$. Since $B^i$ fixes the endpoints of $\psi^{i-1}$, we see that $H(0,s) = H(1,s) = \omega_0$ for all $s \in \interval$. Clearly $H(t,1)$ is constant at $\omega_0$. All that remains is to show weak* continuity of $H$.

It suffices to show continuity of $H(t,s)(C)$ for fixed but arbitrary $C \in \bigcup_{n \in \bbN} \fA_n \subset \fA$. On $\interval \times [0,1)$, we have $H(t,s)(C) = (A_{t,s} \cdot \omega_t)(C)$, which is continuous since $A$ is continuous. Now, there exists $n \in \bbN$ such that $C \in \fA_n$. For $i > n+1$ and $s \in [1 - 2^{-i+1}, 1-2^{-i}]$, we have
\begin{align*}
H(t,s)(C) &= (A_{t,s} \cdot \omega_t)(C)\\
&= \qty[\qty(B_{t,r}^i B_{t,1}^{i-1} \cdots B_{t,1}^{n+1}) \cdot \qty(B_{t,1}^{n} \cdots B_{t,1}^1) \cdot \omega_t](C)\\
&= \qty[\qty(B_{t,r}^i  B_{t,1}^{i-1} \cdots B_{t,1}^{n+1}) \cdot \psi^n_t](C) = \psi^n_t(C) = \omega_0(C).
\end{align*}
The fact that $\psi^n_t(C) = \omega_0(C)$ follows from the fact that $\psi^n_t(P_i) = 1$ for all $i \leq n$, which implies $\psi^n_t|_{\fA_n} = \omega_0|_{\fA_n}$. The second to last step follows from the fact that $\psi^n_t|_{\fA_n}$ is pure, the fact that $B^i_{t,r} B_{t,1}^{i-1} \cdots B_{t,1}^{n+1} \in \fA_n'$, and an application of \cref{lem:restricted_state_factorized}. Thus, on $\interval \times [1-2^{-n-1}, 1]$, we see that $H(t,s)(C)$ is constant at $\omega_0(C)$. This proves that $H(t,s)$ is weak* continuous on all of $\interval \times \interval$, completing the proof.
\end{proof}

\appendix
\section{Homotopical background}\label{sec:appendix}
In this appendix, we review some notions from homotopy theory that are used in the paper.
We refer the reader to \cite{AguilarGitlerPrieto,MayGILS,MayEinfty,MayConcise,GoerssJardine}
for further information.
Throughout the appendix, by ``topological spaces'', we mean compactly generated topological spaces
and follow \cref{conv:top}. 
\subsection{Topological monoids and $H$-spaces}\label{sec:topmonhspace}

We start with a space  $M$ together with
a binary operation
\[\otimes \colon M \times M \to M.\]
and a map
\[\varepsilon_0 \colon \mathrm{pt} \to M \]
selecting a base point in $M$.
We have three diagrams which ``categorify'' common algebraic properties:
\begin{enumerate}
\item[(u)] (unitality) 
$
\xymatrix{M \ar[r]^-{\varepsilon_0 \times \id} \ar[dr]_-{\id} & M\times M\ar[d]^-{\otimes}  & M  \ar[l]_-{\id \times \varepsilon_0} \ar[dl]^-{\id} \\
  &M &}
$

\item[(a)] (associativity) $\xymatrix{ M\times M \times M \ar[r]^-{\otimes \times \id}\ar[d]_-{\id \times \otimes}&  M\times M \ar[d]^-{\otimes} \\
M\times M \ar[r]^-{\otimes} & M
 }
$
\item[(c)] (commutativity)
$
\xymatrix@C=1pc{ M\times M \ar[dr]_-\otimes \ar[rr]^-s  &  & M\times M \ar[dl]^-{\otimes} \\
 & M & 
 } 
$
\end{enumerate}

We say that a diagram commutes \emph{up to homotopy} if there is a homotopy between any composite of arrows around the diagram which begin at the same object and have the same endpoint. To contrast this weaker condition, we say that the diagram commutes \emph{strictly} if it is a commutative diagram, so that any composite of arrows around the diagram that start at the same place and end at the same place are equal arrows. Clearly, if a diagram commutes strictly, then it commutes up to homotopy.

\begin{defn}\label{defn:Hspace}(see e.g.\ \cite[Sec.~2.7]{AguilarGitlerPrieto})
Let $M$ be a based topological space, with base point $\varepsilon_0$, and a continuous binary operation $\otimes$.
\begin{enumerate}
\item We say that $M$ is an \emph{$H$-space} if the diagram (u) commutes up to homotopy. It is called \emph{homotopy associative} if (a) commutes up to homotopy and \emph{homotopy commutative} if (c) commutes up to homotopy.
\item We say that $M$ is a \emph{topological monoid} if (u) and (a) commute stictly. We say that it is a
  \emph{commutative topological monoid} if, in addition, (c) commutes strictly.  
\item A \emph{homotopy commutative topological monoid} is a topological monoid for which (c) commutes up to homotopy.
\end{enumerate}

A morphism of $H$-spaces is a continuous function $\phi \colon M \ra N$ so that
\[ \xymatrix{ M \times M \ar[d]^-\otimes \ar[r]^{\phi \times \phi} & N\times N \ar[d]_\otimes \\
M\ar[r]^-\phi  & N
}  \quad \quad \xymatrix{ \pt \ar[r]^{\varepsilon_0}\ar[dr]_{\varepsilon_0} & M \ar[d]^-\phi \\ & N}  \]
commute up to homotopy. For a morphism of topological monoids, the diagrams are required to commute strictly.
\end{defn}

\begin{ex}
Here are a few examples:
\begin{enumerate}
\item
Any topological group $G$ is a topological monoid. 
\item A loop space $\Omega X$ is a homotopy associative $H$-space and $\Omega^2X$ is a homotopy associative and commutative $H$-space. 
\item Any monoid $M$ is a topological monoid with the discrete topology. 
\end{enumerate}
\end{ex}


If $M$ is a homotopy associative $H$-space or a topological monoid and $X$ is a topological space, the homotopy classes of maps $[X,M]$ form a monoid (in the sense of algebra) where for $f,g\in [X,M]$,
\[(f\otimes g) (x) = f(x)\otimes g(x).\]
If $M$ is homotopy commutative, this monoid is commutative. 

\begin{rem}\label{rem:pi0monoid}
The path components of a space are the homotopy classes of maps
\[\pi_0X = [\pt, X]\]
from the one point space $\pt$. Hence, if $M$ is a homotopy associative $H$-space, then $\pi_0M$ is a monoid. In fact, if $\pi_0M$ is given the discrete topology, then so long as $M$ is locally path connected,
\[ M \xrightarrow{ \ \pi_0 \ } \pi_0M\]
is a continuous map of $H$-spaces. 
\end{rem}

\begin{defn}
If $M$ is a homotopy associative $H$-space or a topological monoid, we say that $M$ is \emph{grouplike} if $\pi_0M$ is a group. That is, if the monoid $\pi_0M$ has inverses.
\end{defn}

\subsection{Operads and $\sE_\infty$-spaces}\label{sec:opandeinfty}

One of the first appearances of $\sE_\infty$-spaces is in Boardman-Vogt \cite{BoardmanVogt} under the guise of \emph{homotopy everything spaces}. The intention was to capture the notion of an $H$-space which is commutative up to all higher homotopies. Roughly, this means that the commutativity diagram (c)  commutes up to homotopy, and that given two homotopies, there is a homotopy between them, and so on. May in \cite{MayGILS} introduced the notion of an $\sE_\infty$-operad. This allowed for a re-packaging of the structure present in a ``homotopy everything space'' into a mathematical object called an $\sE_\infty$-space. The reformulation is discussed in \cite[\S 1]{MayEinfty}.


\bigskip

An $\sE_\infty$-operad parametrizes the higher homotopies that witness commutativity. We start with some definitions
and give references for further details.
In what follows, $\Sigma_j$ denotes the symmetric group on $j$ letters, and
a $\Sigma_j$-\emph{space} is a topological space with a continuous
right action of $\Sigma_j$.

\begin{defn}[{\cite[Def. 1.1]{MayGILS}}]
An \emph{operad} $\cO$ in topological spaces consists of the following data:
\begin{enumerate}[(1)]
\item for each integer $j\geq 0$, a $\Sigma_j$-space $\cO(j)$,
\item for each tuple $(k; j_1, j_2, \ldots, j_k)$ of non-negative integers
  such that the relation $\sum\limits_{i=1}^k j_i = j$ holds continuous operations
  \[\gamma \colon \cO(k) \times \cO(j_1)\times \cdots \times \cO(j_k) \to \cO(j)\]
  which satisfy 
\begin{enumerate}[(a)]
\item the associativity formula
  \[
    \gamma (\gamma (c; d_1,\ldots , d_k); e_1,\ldots , e_j) =
    \gamma (c; f_1,\ldots , f_k)
  \]
  for all $c \in \cO(k)$, $d_s \in \cO(j_s)$, and $e_t \in \cO(i_t)$,
  where
  \[ f_s= \gamma (d_s;e_{j_1+ \ldots + j_{s-1}+1},\ldots , e_{j_1+\ldots + j_s}) \  ,\]
and
\item  the equivariance conditions
  \begin{align*}
    \gamma (c\sigma ; d_1,\ldots ,d_k) & =
    \gamma \big( c; d_{\sigma^{-1} (1)},\ldots ,d_{\sigma^{-1} (k)} \big)
    \sigma (j_1,\ldots ,j_k) \ , \\
    \gamma (c; d_1\tau_1 , \ldots , d_k\tau_k ) & =
    \gamma (c; d_1 , \ldots , d_k )  (\tau_1\times \ldots \times \tau_k) 
  \end{align*}
  for all $c \in \cO(k)$, $d_s \in \cO(j_s)$, $\sigma \in \Sigma_{k}$,
  and $\tau_s\in  \Sigma_{j_s}$, where
  $\sigma (j_1,\ldots j_k)$ denotes the permutation of $j$ letters
  given by permutation of the $k$ blocks via $\sigma$, and $\tau_1 \times \ldots \times \tau_k$
  is the image of $(\tau_1,\ldots,\tau_k)$ under the
  canonical embedding
  $\Sigma_{j_1} \times \ldots \times \Sigma_{j_k} \hookrightarrow \Sigma_j$. 
\end{enumerate}

\item a distinguished element $1 \in \cO(1)$ which acts as a unit for the operations $\gamma$ in the sense that 
\begin{align*}
\gamma(1; -) &\colon  \cO(j) \to \cO(j) \quad \text{and } \\
\gamma(-; 1^k ) &\colon \cO(k) \to \cO(k)
\end{align*} 
are the identity morphisms.
\end{enumerate}
If $\cO(0)$ is a single point, the operad is called \emph{reduced}. If the $\cO(j)$ are free $\Sigma_j$-spaces, we say that $\cO$ is \emph{$\Sigma$-free}. 
\end{defn}

\begin{defn}
An  \emph{$\mathscr{E}_\infty$-operad} is an operad $\cO$ which is reduced, $\Sigma$-free and such that the spaces $\cO(j)$ are contractible for all $j$.
\end{defn}

\begin{ex}\label{rem:alleinftyopequal}
The Barratt--Eccles operad $\mathcal{E}$ is the operad given by 
\[\mathcal{E}(j) = E\Sigma_j\] 
where $E\Sigma_j$ is the universal free contractible $\Sigma_j$-space of the symmetric group on $j$ elements. The maps $\gamma$ are given by applying the functor $E(-)$ to the group homomorphism
\[ \Sigma_k \times \Sigma_{j_1} \times \cdots \times \Sigma_{j_k} \to \Sigma_{j}\]
which sends $(\tau_k ; \sigma_{j_1} , \ldots , \sigma_{j_k})$ to the obvious block permutation. 

Given any other $\sE_\infty$-operad $\cO$, then $\mathcal{E} \times \cO$ is again an $\sE_\infty$-operad. Each projection from $\mathcal{E} \times \cO$ to $\cE$ and $\cO$ are maps of operads which are also weak equivalences on each space of the operad. This witnesses that fact that any $\sE_\infty$-operad is weakly equivalent to the Barratt--Eccles operad. Later, we see another $\sE_\infty$-operad, the linear isometry operad.  
\end{ex}


\begin{ex}[{\cite[Def. 1.2]{MayGILS}}]\label{ex:endoop}\label{rem:action}
 If $X$ is a topological space, then $X^j$ is a left $\Sigma_j$-space  if we let a permutation $\sigma$ act as
\[\sigma(x_1,\cdots, x_j) = (x_{\sigma^{-1}(1)}, \cdots, x_{\sigma^{-1}(j)}). \]
Then for any space $Y$, the space $\Map(X^j,Y)$ is made into a right $\Sigma_j$-space by letting $(f\sigma) = f\circ \sigma$.
The \emph{endomorphism operad} of $X$, denoted $\mathcal{E}(X)$ is then defined by $\mathcal{E}(X)(j) =\Map(X^j,X)$. The unit element is the identity in $\mathcal{E}(X)(1)$. The structure maps are given by $\gamma(f;g_1, \cdots, g_k) =f(g_1\times \cdots \times g_k)$ and the symmetric group acts as  $f\sigma= f\circ \sigma$.

  More generally, if $\cC$ is a topologically enriched symmetric monoidal category and $A$ is an object of $\cC$,  the endomorphism operad $\cE_\cC(A)$ is the operad in $\cC$ with $\mathcal{E}_{\cC}(A)(j) = \cC(A^{\otimes j}, A)$. The unit is again the identity morphism, and the structure maps $\gamma$ satisfy the same formula where we replace cartesian product with the monoidal product $\otimes$.
\end{ex}

\begin{defn} 
Let $\cO$ be an operad.  An \emph{$\cO$-space} consists of a 
 space $X$ and continuous maps 
\[\theta_j = \theta_j^X \colon \cO(j) \times X^j \to X\]
so that
\begin{enumerate}
\item $\theta_1(1;-)$ is the identity on $X$,
\item $\theta_j(-\sigma ; -) = \theta_j(-;\sigma -)$ for $\sigma \in \Sigma_j$, where $\Sigma_j$ acts on $X^j$ as in \cref{rem:action}.  
\item the $\theta$'s and $\gamma$'s are compatible in the sense that 
  whenever $\sum\limits_{i=1}^k j_i = j$, the following diagram commutes:
\[\xymatrix@R=0.5pc{
\cO(k) \times \cO(j_1)\times \cdots \times \cO(j_k) \times X^j \ar[rr]^-{\gamma \times \id} \ar[dd]_-{\id \times \mathrm{shuffle}} && \cO(j)\times X^j \ar[dr]^-{\theta_j}  & \\
& &  & X \\
 \cO(k)\times \cO(j_1) \times X^{j_1} \times \cdots \times  \cO(j_k) \times X^{j_k} \ar[rr]^-{\id \times \theta_{j_1} \times \cdots \times \theta_{j_k}} & & \cO(k)\times X^k \ar[ur]_-{\theta_k} & }\]
\end{enumerate} 
A map of $\cO$-spaces is a continuous map $f\colon X \ra Y$ making the diagram
\[\xymatrix{
\cO(j) \times X^j \ar[r]^-{\theta_j^X} \ar[d]_-{\id\times f^j} & X \ar[d]^-f\\
\cO(j) \times Y^j \ar[r]_-{\theta_j^Y} &  Y
}\]
commute. 

If $\cO$ is a reduced operad, then a \emph{based $\cO$-space} is a based space $X$ which is an $\cO$-space and such that $\theta_0$ is the base point of $X$.
\end{defn}

\begin{rem}
One can define the notion of a morphism of operads $\cO \ra \cP$. These are continuous maps $\cO(j) \ra \cP(j)$ that respect all of the structure. Specifiying an $\cO$-space is then equivalent to giving a morphism of operads from $\cO$ to $\cE(X)$.
\end{rem}

\begin{defn}
A space $X$ is an \emph{$\sE_\infty$-space} if it is a based $\cO$-space for some $\sE_{\infty}$-operad $\cO$. 
\end{defn}

An $\sE_\infty$-space is automatically a homotopy commutative and associative $H$-space. Choose any element $c \in \cO(2)$. We get a binary operation
\[ \theta := \theta_2(c;-) \colon X^2 \to X\]
and 
\[\varepsilon  = \theta_0 \colon \cO(0)= \mathrm{pt} \to X.\]
The properties of an $\cO$-space then imply that $X$ is an $H$-space with these maps. 

\begin{lem}
The space $X$ together with $\theta$ and $\varepsilon$ as defined above is a homotopy commutative and associative $H$-space.
\end{lem}
\begin{proof}[Proof sketch.]
See \cite[Lemma 1.9]{MayGILS} for the homotopy associativity. We will check unitality and leave commutativity as an exercise. The diagram
\[\xymatrix@R=0.5pc{
\cO(2) \times \cO(0)\times  \cO(1) \times X \ar[rr]^-{\gamma \times \id} \ar[dd] && \cO(1)\times X \ar[dr]^-{\theta_1}  & \\
& &  & X \\
 \cO(2)\times \cO(0) \times X^{0} \times   \cO(1) \times X^{1} \ar[rr]^-{\id \times \theta_{0} \times \theta_{1}} & & \cO(2)\times X^2 \ar[ur]_-{\theta_2} & }\]
commutes. It follows that, as maps $X \rightarrow X$, 
\[\theta_1(\gamma(c; \ast ,1),x) =\theta_2(c; \ast,x) = \theta(\ast,x). \]
Since $\cO(1)$ is connected, we can choose a path $\alpha \colon [0,1] \rightarrow \cO(1)$ from $\gamma(c; \ast ,1)$ to $1$. Since $\theta_1$ is continuous
\[ h\colon X \times [0,1]  \to X, \quad h(x,t )=\theta_1(\alpha(t),x)\]
is a homotopy from $\theta_1(\gamma(c; \ast ,1),x)$ to $\theta_1(1,x)$. But $\theta_1(1,x)$ is the identity on $X$, so $\theta(\ast,x)$ is homotopic to the identity.
\end{proof}

We can fix an $\sE_\infty$-operad and define the category of $\sE_\infty$-spaces to be the category of $\cO$-spaces and morphisms. As discussed in \cref{rem:alleinftyopequal}, any choice of $\sE_\infty$-operad $\cO$ gives an equivalent category of $\sE_\infty$-spaces. For this reason, we allow ourselves to be imprecise about the operad and simply talk about the \emph{category of $\sE_\infty$-spaces}.

\begin{defn} We let $\sE_\infty\text{-}\bCGTop$ denote the category of $\sE_\infty$-spaces.  
We say that an $\sE_\infty$-space $X$ is \emph{grouplike} if $\pi_0X$ is a group with respect to the induced $H$-space structure on $X$.
\end{defn}


\subsection{Simplicial sets and geometric realization}\label{sec:simplicial}
  
We review the definition of a simplicial set and its geometric realization.
References on this include  \cite{Maysset, GoerssJardine}.

\begin{defn}
A \emph{simplicial set} $X_\bullet$ is a sequence of sets 
\[X_0, X_1, X_2, \ldots,\] together with \emph{face maps} $\partial_i \colon X_k \ra X_{k-1}$ defined for 
$0\leq i \leq k$ if $k\neq 0$ and \emph{degeneracy maps} $s_j \colon X_k \ra X_{k+1}$ defined for $0\leq j \leq k$.
These data must satisfy some coherency conditions called the \emph{simplicial identities}:
\begin{align*}
\hspace{30mm}\partial_i\partial_j&=\partial_{j-1}\partial_i &&\text{if }0 \leq i<j \leq k\ , \hspace{40mm}\\
s_is_j&=s_{j+1}s_i && \text{if } 0 \leq i\leq j \leq k\ , \hspace{40mm}\\
 \partial_i s_j&=s_{j-1}\partial_i &&\text{if }  0 \leq i< j \leq k\ , \hspace{40mm}\\
     \partial_i s_j&= s_j  \partial_{i-1}&& \text{if } 1 \leq j+1 < i \leq k+1 \ , \hspace{40mm} \\
  \partial_i s_j&=1 && \text{if } i = j \text{ or } i = j+1 \ . \hspace{40mm}
\end{align*}
An element of the set $X_j$ is called a \emph{$j$-simplex}.
\end{defn}

Simplicial sets form a category, denoted $\sSets$, where a morphism 
\[f_\bullet \colon X_\bullet \to Y_\bullet\] 
is a sequence of functions $f_k \colon X_k \to Y_k$ which commute with all face and degeneracy maps, i.e.,
$\partial_i f_k = f_{k-1} \partial_i$ and $s_i f_k = f_{k+1} s_i$  for all $k$ and $i$ with $0\leq i \leq k$.

\begin{ex}
For $Y$ a topological space, we associate a 
simplicial set  $\Sing_\bullet(Y)$ as follows. 
Let 
\[\Delta_k = \big\{(t_0, \ldots, t_k) \in \R^{k+1}\big|\: \sum_{i=1}^k t_i=1 \ \text{and} \ 0\leq t_i \leq 1\big\}\] 
be the standard $k$-dimensional simplex. For $0\leq i\leq k$, let $\delta_i \colon \Delta_{k-1} \rightarrow \Delta_k$
be the map which inserts a zero in the $i$th spot, and $\sigma_i \colon \Delta_{k+1} \rightarrow \Delta_k$ be the map
that sums the $i$th and $(i+1)$th coordinates.

To define  $\Sing_\bullet(Y)$, we let
\[\Sing_k(Y) = {\Map}(\Delta_k,Y) \ .\] 
The face and degeneracy maps are the maps $\partial_i$ and $s_i$ obtained by precomposition with
$\delta_i$ and $\sigma_i$, respectively. The simplicial identities precisely reflect the relationship
imposed by the relationships satisfied by $\delta_i$ and $\sigma_i$.
\end{ex}

\begin{defn}
The \emph{geometric realization}
of $X_{\bullet}$, denoted by $|X_{\bullet}|$, is the following topological space. The space underlying $|X_\bullet |$ is the quotient of
\[\coprod_{k\geq 0} X_k \times \Delta_k,\]
by the equivalence relation generated by
\[(\partial_ix, t) \sim_\partial (x, \delta_it)\quad \text{for } x \in X_k, \: t \in \Delta_{k-1}, \text{ and } 0\leq i \leq k \ ,  \]
and
\[(s_ix, t) \sim_s (x, \sigma_it) \quad \text{for } x \in X_k, \: t \in \Delta_{k+1}, \text{ and } 0\leq i \leq k\ .\]
\end{defn}

Roughly, the set $X_k$ labels the ``$k$-cells'' of a CW complex glued along the face maps $\partial_i$. 
This is the gluing recorded by the first relation $\sim_\partial$. However, not all elements of $X_k$ truly give $k$-cells.
Some are what we call ``degenerate''. These are the simplicies in the image of the $s_i$-maps. The second relation
$\sim_s$ is telling you to crush those $k$-cells, gluing them onto a lower dimensional skeleton. Therefore,
we could have written the geometric realization by first throwing out all of the degenerate elements in $X_k$, keeping
those simplicies $X_k^{\mathrm{non\text{-}degen}}$ not in the image of the degeneracy maps $s$, and then using only the first
relation to take the quotient,
\[ |X_\bullet| \cong \left( \coprod_{k\geq 0} X_k^{\mathrm{non\text{-}degen}} \times \Delta_k\right)/\sim_\partial.\]

The importance of these construction is captured in the following adjunction, which allows one to do homotopy theory in the nicer category of CW-complexes. Indeed, let $\CWTop$ be the full subcategory of $\CGTop$ whose objects are compactly generated spaces that are homotopy equivalent to CW-complexes. 
There is an adjunction
\[\xymatrix{   |-| : \sSets   \ar@<1ex>[r] & \CGTop:  \Sing_\bullet  \ar@<1ex>[l] }\]
where $ \Sing_\bullet$ is the right adjoint. The co-unit of the adjunction
\[\varepsilon_X \colon  |\Sing_\bullet(X)| \to X\]
gives a functorial replacement for any topological space $X$ by a weakly equivalent CW-complex $|\Sing_\bullet(X)|$.


\newcommand{\etalchar}[1]{$^{#1}$}
\providecommand{\bysame}{\leavevmode\hbox to3em{\hrulefill}\thinspace}
\providecommand{\MR}{\relax\ifhmode\unskip\space\fi MR }
\providecommand{\MRhref}[2]{%
  \href{http://www.ams.org/mathscinet-getitem?mr=#1}{#2}
}
\providecommand{\href}[2]{#2}

\end{document}